\documentclass[10pt]{article}

\usepackage{graphicx}
\usepackage{tikz}
\usetikzlibrary{arrows.meta, positioning}

\usepackage[margin=1in]{geometry}
\usepackage{setspace}
\usepackage{amsmath}
\usepackage{longtable}
\usepackage{algorithm}
\usepackage[noend]{algpseudocode}

\usepackage{amsfonts}
\usepackage{amssymb}
\usepackage{bm}
\usepackage{natbib}
\usepackage{enumitem}
\usepackage{dsfont}

\newcommand{\1}{\mathds{1}}

\usepackage{subcaption}  
\usepackage{graphicx}
\usepackage{float}
\usepackage{mathtools}

\usepackage{url}
\usepackage{tikz}
\usetikzlibrary{calc, arrows.meta, positioning, shapes.geometric}

\definecolor{SWIGGray}{RGB}{200, 200, 200}
\tikzset{>=Latex}
\tikzstyle{obs} = [draw, circle, minimum size=0.9cm, inner sep=1pt, thick, fill=white]
\tikzstyle{lat} = [draw, circle, minimum size=0.9cm, inner sep=1pt, thick, dashed, fill=white]
\tikzstyle{opt} = [draw, circle, minimum size=0.9cm, inner sep=1pt, thick, dotted, fill=white]
\tikzstyle{fix} = [draw, rectangle, minimum size=0.7cm, inner sep=2pt, thick, fill=gray!20]
\tikzstyle{arr} = [->, thick, >=Latex]
\tikzstyle{bidir} = [<->, thick, >=Latex, dashed]

\usepackage{booktabs}
\usepackage{comment}
\usepackage{theoremref}
\usepackage{amsthm}
\usepackage{multirow}

\usepackage[dvipsnames]{xcolor}

\renewcommand{\P}{\mathbb{P}}

\newcommand{\E}{\mathbb{E}}
\newcommand{\R}{\mathbb{R}}
\newcommand{\Q}{\mathbb{Q}}
\newcommand{\Hh}{\mathbb{H}}
\newcommand{\G}{\mathbb{G}}

\newcommand{\calX}{\mathcal{X}}

\newcommand{\calY}{\mathcal{Y}}

\newcommand{\calI}{\mathcal{I}}

\newcommand{\calJ}{\mathcal{J}}
\newcommand{\calB}{\mathcal{B}}
\newcommand{\calT}{\mathcal{T}}
\newcommand{\calC}{\mathcal{C}}

\newcommand{\calZ}{\mathcal{Z}}
\newcommand{\calA}{\mathcal{A}}
\newcommand{\calN}{\mathcal{N}}
\newcommand{\calS}{\mathcal{S}}

\newcommand{\Var}{\mathrm{Var}}

\newcommand{\base}{\mathrm{base}}
\newcommand{\Rem}{\mathrm{Rem}}

\newcommand\independent{\protect\mathpalette{\protect\independenT}{\perp}}
\def\independenT#1#2{\mathrel{\rlap{$#1#2$}\mkern2mu{#1#2}}}

\newcommand{\wh}[1]{\widehat{#1}}
\newcommand{\wt}[1]{\widetilde{#1}}
\newcommand{\wb}[1]{\overline{#1}}

\newcommand{\corr}{\mathrm{Corr}}

\newcommand{\obs}{\mathrm{obs}}

\newcommand{\naive}{\mathrm{naive}}

\algrenewcommand\algorithmicrequire{\textbf{Input:}}
\algrenewcommand\algorithmicensure{\textbf{Output:}}
\newcommand{\intm}{\mathrm{int}}
\newcommand{\extm}{\mathrm{ext}}

\newtheorem{theorem}{Theorem}
\numberwithin{theorem}{section}
\newtheorem{lemma}[theorem]{Lemma}

\newtheorem{corollary}[theorem]{Corollary}

\newtheorem{ass}{Assumption}

\theoremstyle{definition}

\theoremstyle{remark}
\newtheorem{remark}[theorem]{Remark}

\usepackage{hyperref}
\hypersetup{
	colorlinks=false,
	bookmarks=true,
	breaklinks=true,
	hidelinks=true,
	pdfpagemode=empty,
}
\usepackage{nameref}
\usepackage{authblk}
\usepackage[noabbrev, capitalize, nameinlink]{cleveref}
\title{Measuring Gift Card Program Incrementality \\via Causal Data Fusion}
\author[1]{Justin Whitehouse\thanks{Corresponding author: contact at \texttt{jwhiteho@stanford.edu}}}
\author[2]{William Betz}
\author[2]{Yan Zhang}
\author[2]{Peter Coles}
\author[1]{\\Ramesh Johari}
\author[1]{Vasilis Syrgkanis}

\affil[1]{Management Science and Engineering, Stanford University}
\affil[2]{Airbnb}

\date{October 5, 2026}

\begin{document}

\maketitle

\begin{abstract}

Businesses regularly offer gift card programs to drive customer spending and increase engagement. A central question is how much incremental revenue these programs generate, and which channels drive it most efficiently. Measuring the incremental revenue associated with a gift card program is a challenging problem in causal inference, requiring a firm to infer how much each customer would have spent if they never received a gift card. Observational data on past customer purchasing behavior reveal possession of a gift card only when a customer makes a purchase, thus leaving a customer’s \textit{treatment status} systematically censored. 

In this paper, we develop a novel data fusion approach to overcome this missing data challenge. We identify and estimate incrementality by combining a large observational dataset with a smaller experimental dataset from a different population. Our approach relies on a mild \textit{transferability} condition, which posits that the conditional relative treatment effect of gift card receipt on the decision to purchase is invariant across the two populations. We develop a flexible, machine learning-based estimator for the incremental revenue and establish its asymptotic normality.

We apply our estimator across both first- and third-party channels through which Airbnb distributes gift cards, finding heterogeneity in incrementality across segments of the population. In particular, we find not only that third-party channels are more incremental than first-party ones, but also that ``self-gifters'' (i.e., customers likely to have purchased their own gift cards) are more incremental than the broader population.

\end{abstract}

\newpage

\section{Introduction}
\label{sec:intro}

Gift cards are stored-value instruments offered by businesses across disparate sectors, with restaurant, travel, and retail industries being notable examples. These instruments have an undeniably large economic footprint: in the United States alone the gift card market was estimated at \$447 billion in 2025 and is projected to exceed \$1.2 trillion by 2035~\citep{precedenceresearch2026giftcards}. 
Further, it is estimated that eight out of ten U.S. shoppers have purchased a gift card in the past twelve months~\citep{blackhawk2026stretched}. 
To some degree, the scale of these gift card programs is not surprising: they offer a variety of strategic benefits to consumers and firms alike. For consumers, these cards provide a convenient gift that offers increased flexibility over traditional gifts.
Further, depending on the channel of distribution, they also offer monetary incentives for self-gifting --- the phenomenon through which a consumer purchases their own gift card.
For businesses, gift cards can be \textit{incremental}, generating additional revenue on top of what the firm would have experienced without the program. This business-side potential raises an important question: how can firms rigorously assess the incrementality of their gift card programs?

Determining the incrementality of gift card programs is an essential question for businesses to address. While businesses seek to develop customer-centric programs that drive engagement and growth, it is important to note that gift card programs are not costless to operate. Each transaction carries direct expenses, including production, distribution, and payment processing, as well as substantial fees to third-party retail and online distribution partners. Further, there are fixed costs associated with running such programs, such as employee salaries and operating fees. When cards are sold at a promotional discount or bundled with loyalty incentives, the effective cost per card rises further. If a business can reliably assess the incrementality of their gift card program, they can leverage their internal cost structure to determine which channels of distribution drive profit most efficiently.

Determining the degree of gift card program incrementality is not simply a matter of measuring the raw revenue associated with gift card redemption. Rather, it requires answering a counterfactual question: how would customers have behaved had they not received a gift card? \citet{norvell2017gift} identify three distinct modes of customer interaction with gift cards. Some recipients provide an \textit{incremental visit}, making a purchase they would not have made absent the card (extensive margin effects). Others exhibit \textit{incremental spend}, purchasing a greater volume of goods than they would have otherwise (intensive margin effects). Finally, some customers exhibit \textit{cannibalization}, using the gift card to pay for a purchase they were already planning to make out-of-pocket. In particular, one might expect cannibalizing behavior to be particularly prevalent among ``self-gifters,'' or individuals who purchase cards for their own use (often, but not exclusively, to exploit promotional discounts, cash-back, or other transfers of value).
While previous works have attempted to assess the incrementality of gift card programs, most have largely relied upon survey methods, which have a natural risk of being biased.
To handle the complex set of counterfactual behaviors enumerated above, one needs to develop a rigorous, mathematical approach  for  gauging program incrementality.

In this paper, we develop a formal causal framework for gauging the incrementality of gift card programs. To evaluate a program's incrementality, or to decompose the earnings from gift card holders into intensive, extensive, and cannibalization shares, one must develop target estimands that capture the impact of receiving a gift card on an individual's spend.
In this paper, we consider two such estimands. 
The first quantity is an \textit{incrementality coefficient (IC)}, which gives the average number of incremental dollars of revenue that are generated for each one dollar of gift card value redeemed. 
The second is an \textit{incrementality ratio (IR)}, which captures the fraction of revenue amongst gift card users that is attributable to the existence of the gift card program.
Both of these quantities can be formalized in terms of the \textit{incremental revenue} of the gift card program, which is the counterfactual causal quantity defined as the expected per-customer change in revenue were the gift card program to be shut down. The former is just the incremental revenue divided by the \textit{average gift card spend amount} over customers, and the latter is the incremental revenue divided by the \textit{average spend amount for those who possess a gift card}.
Instead of focusing our work around this incremental revenue, we opt to focus on the IR and IC as these objects  offer a scale-free means of assessing the incremental revenue of a company's gift card program.

In principle, a firm could estimate the above effects by running a large-scale randomized controlled trial (RCT), in which one half of the population would be experimentally assigned to receive a gift card while the other half would not. However, distributing gift cards to a massive treated group would be prohibitively expensive. Further, the population considered for involvement in the experiment would need to be distributionally similar to the target population of interest, which may be difficult to accomplish in some settings (e.g., for recipients of third-party gift cards). Conversely, attempting to estimate the IC and IR from historical observational data presents a severe missing data problem that has been underexplored in the broader causal inference literature. In many settings, gift card possession is only revealed to the business when a customer makes a purchase; if a customer does not purchase, the firm cannot know whether they hold an unredeemed gift card.\footnote{Some businesses actually observe whether or not a customer possesses a gift card when they ``claim'' a gift card, or upload its value to a company's portal. However, throughout this paper, we view the act of claiming and redeeming (or making a purchase with) a gift card as indistinguishable.} Because the decision to purchase is endogenously determined by both customer characteristics and gift card possession, the treatment variable itself is systematically censored. Consequently, standard observational causal inference methods fail to identify the desired effects. For instance, even if one assumes standard one-sided conditional ignorability of the treatment,\footnote{That is, conditional on observable customer characteristics, the average counterfactual spend of a gift card recipient \textit{in the absence of a gift card} equals the observed spend of non-recipients with the same characteristics. This assumption allows one to impute the treated group's counterfactual behavior using comparable untreated customers.} the propensity score cannot be estimated because the number of gift card holders within each population sub-group is unobserved. We are thus at an impasse: \emph{experiments are far too costly to run at a large scale, and observational data inevitably suffers from a partially observed treatment.}

In this paper, we resolve this impasse by proposing a data fusion approach that combines a large, treatment-censored observational dataset of historical customer transactions with a small, auxiliary experimental dataset in which treatment is fully observed. The key insight is that while observational data alone cannot identify the target causal effects due to censoring, and experimental data alone is typically too small and non-representative of the broader population to yield precise estimates, the two datasets can yield proper identification when taken together. The observational data provides rich information about customer spending behavior conditional on purchase in the target population of interest, while the experimental data reveals the causal relationship between gift card possession and the decision to purchase. By fusing these two sources of information, we can recover both the incrementality coefficient and ratio (along with the raw, incremental revenue itself) without requiring a large scale experiment.
We emphasize that our particular approach to data fusion in this paper is entirely novel and not captured by existing works on missing data (see Section~\ref{sec:related} for a full discussion).

Crucially, our approach does not assume that the experimental and observational populations are similar in any strong sense. The distributions over customer characteristics, gift card assignments, and realized spend amounts can differ arbitrarily between the two datasets. We require only a single, heterogeneous \textit{transferability} assumption: the ratio between the conditional probabilities (given covariates) of someone booking when they respectively do and do not possess a gift card is invariant between the experimental and observational distributions.
As a special case, this condition captures the assumption that the conditional booking probabilities given a customer's characteristics and gift card receipt status are the same between the two populations.
We emphasize that this assumption is substantially weaker than requiring the populations to share outcome distributions or treatment assignment mechanisms, and the assumption can remain credible even when there are systematic differences in conditional booking rates between the experimental and observational datasets (say, due to temporal shifts). 

Under this transferability assumption, together with standard one-sided conditional exogeneity and positivity conditions, we prove that the incremental revenue, and hence also the IC and IR, are nonparametrically identified. The identification formula expresses both the IC and IR in terms of three nuisance functions: the \textit{experimental} gifted and ungifted conditional booking rates and the organic spend function (the expected spend conditional on booking without a gift card), which is identified from the \textit{observational} data. We further show that the identified IR admits a natural decomposition into  \textit{intensive and extensive margin shares},  capturing the fraction of incremental spend resulting from increased spend from customers who would have already made a purchase and those who purchase \textit{only because} of the gift card, respectively. This decomposition not only clarifies which components require experimental data and which are identified from the observational population alone, but also provides guidance on whether a firm's gift card program primarily drives increased out of pocket spend or deeper engagement among existing customers.\footnote{Note that both avenues of influence can exhibit cannibalization: some bookings could have happened even in the absence of a card and a fraction of the spend could have occurred without the card. The causal effects that we identify attempt to separate out the incremental influence along these two channels from the cannibalization effect.}

Based on our identification result, we prove the $\sqrt{n}$-consistency and asymptotic normality of a natural cross-fit estimator~\citep{chernozhukov2018double, van2006targeted, bang2005doubly} for our effects. Crucially, our Neyman orthogonal construction allows us to establish valid inference even when flexible machine learning methods are used to estimate the heterogeneous booking rate and organic spend functions. A distinctive feature of our analysis is that it accommodates highly unbalanced sample sizes: our coverage guarantees remain valid even when the number of experimental observations grows sub-linearly relative to the observational sample. This scaling regime is precisely the one faced in practice, where firms may have millions of historical transactions but can only afford to experiment on a few thousand customers. The Neyman orthogonal construction ensures robustness to first-order misspecification of the nuisance functions, and the cross-fitting procedure permits the use of flexible machine learning methods for nuisance estimation without introducing regularization bias. We also develop sensitivity analyses for our computed effects, which measure the degree to which our transferability assumption must be misspecified in order to overturn our findings.

We apply our methodology to evaluate the incrementality of the gift card program at Airbnb, one of the world's largest online travel platforms. Airbnb has been expanding its gift card program across multiple geographies and distribution channels over the past decade, and so measuring the degree of program incrementality across segments is particularly timely. We estimate both the IC and IR across four North American channels in which individuals can purchase gift cards: First-Party (via Airbnb's website), Retail Online, Retail In-Store, and Business-to-Business (B2B) sales. To make conditional exogeneity assumptions credible, we control for over 150 customer features capturing historical spend behavior, past interactions with gift cards, price consciousness, and website search patterns. 

Our point estimates indicate that both the incrementality coefficient (IC) and incrementality ratio (IR) are positive in Retail Online and B2B channels (around 0.20-0.22, respectively), and such positivity is statistically significant in B2B alone. On the other hand, we find that the incrementality in First-Party and Retail In-Store channels is generally neutral, with the point estimate in First-Party sales being mildly negative (around $-0.20$). 
We additionally present our estimates on a monthly basis, finding relative stability of our estimates across time.
All channels (with the exception of First-Party) possess mildly positive \textit{intensive margins}, sitting between 0.06-0.09.
This indicates that, conditional on booking, gifted customers spend slightly more than non-gifted ones.
On the other hand, extensive margins are highly variable across channels, and all confidence intervals for said margins contain zero.
We further compare our estimates to naive baselines that either do not adjust for covariates or the amount of time between gift card purchase and redemption. These baselines drastically overstate incrementality, which can in turn misinform downstream decision making tasks.
We further measure incrementality when we restrict our analysis to suspected self-gifters, who are chosen via a proxy based on the number of days between gift card purchase and redemption.
Unlike the general population, we find that self-gifters are \textit{highly incremental}, with IR estimates between 0.35 and 0.50 depending on the channel. Further, the 95\% confidence intervals for our estimates exclude zero, indicating the statistical significance of our findings. 
We find that self-gifters exhibit large intensive and extensive margins, with the latter being particularly pronounced. 
Taken together, our results show that our method can capture bespoke differences in gift card program incrementality amongst the various channels of distribution and types of customers.

Our analysis reveals several key insights regarding the mechanisms underlying these findings. First, the channels that appear most incremental (B2B and Retail Online) are precisely those offering the strongest ``effective discounts'', such as the ability to exchange reward points for gift cards, to receive a free gift card upon making a large purchase, or other transfers of value. 
Second, we surprisingly found that incrementality was \textit{highest} amongst units suspected of self-gifting, and that said incrementality was largely driven through \textit{extensive margins}, or an increase in booking probability upon receipt of a gift card.
This might seem counterintuitive, as one may suspect that those who purchase their own gift cards at a discount are likely subsidizing a trip they already had in mind, and thus are fundamentally cannibalizing out of pocket spend. 
However, we note that such behavior can arise even under simple behavioral economic models  --- see Appendix~\ref{app:discuss:econ_model} for such an example.
Third, we generally find that First-Party sales exhibit the lowest IR and IC, with all third-party channels offering much greater incrementality estimates. Although suspected self-gifters in First-Party do exhibit high incrementality estimates, with a significantly positive IR of 0.37, these units only make up a small fraction of those who actually leverage a gift card through this channel (around 11.6\%). This low prevalence is likely due to the fact that typical incentives associated with self-gifting (such as value transfers or discounts) were not offered in this channel during the study period.
All of the above findings point to a surprising conclusion: channels that entice deal-seeking and self-gifting customers seem to be of higher incrementality to Airbnb than traditional channels through which only true-gifting is likely. 

From a methodological perspective, our work contributes to the growing literature on causal inference in operations management~\citep{ho2017causal}. The censored treatment problem we address, where treatment status is revealed only upon a purchase event,  has not been analyzed in the prior data fusion literature, and our identification result provides a new empirical paradigm for analyzing such environments under flexible transferability assumptions. One promising example would be for promotional coupons distributed through third-party channels, as a business would have no way of knowing whether or not a customer received such a coupon until redemption. Hence, our identification and estimation framework can potentially be applied to these settings, providing a general-purpose tool for evaluating the causal effect of promotional instruments when treatment is partially observed.

The remainder of the paper is structured as follows. In Section~\ref{sec:related}, we discuss related work, with a primary focus on existing measurement approaches for gift card program incrementality and existing causal methods for data fusion. In Section~\ref{sec:context}, we detail the empirical context at Airbnb. Section~\ref{sec:estimands} formalizes the target causal estimands. Section~\ref{sec:model} states the identifying assumptions and presents the main identifiability theorems and decompositions. Section~\ref{sec:estimation} develops the cross-fit estimator and establishes its asymptotic properties. Section~\ref{sec:case_study} presents the empirical application to Airbnb. In Section~\ref{sec:discuss} we provide a qualitative interpretation of our findings.
Finally, we conclude in Section~\ref{sec:conclusion}.

\section{Related Work}
\label{sec:related}

Our work contributes to two distinct streams of literature: the empirical analysis of gift card programs and promotional instruments, and the methodological literature on causal inference and data fusion. We discuss each in turn, emphasizing how prior work motivates the questions studied in this paper and how our setting differs from existing analyses.

\paragraph{Empirical Analysis of Gift Cards, Promotions, and Mental Accounting} There is a large literature that studies the economic motivations behind gift card programs, including their appeal to both businesses and consumers alike. Several authors study business-side incentives for offering gift cards to customers. First, authors such as \citet{offenberg2007markets}, \citet{horne2007unredeemed}, and \citet{berg2021firms} note that gift cards can benefit businesses through prepayment, breakage, customer lock-in, and increased redemption period spending, while also offering consumers flexibility relative to traditional gifts. Further, a related behavioral literature shows that gift cards are not treated as equivalent to cash: because they are restricted-use funds, they can alter mental accounts, encourage hedonic or category-congruent spending, and interact with promotional framing~\citep{thaler1985mental,prelec1998red,white2006format,helion2014gift,reinholtz2015mental,yao2014gift,cheng2018double}. These diverse behavioral mechanisms motivate the central empirical concern in our setting: gift cards may generate incremental visits or incremental spend, but they may also cannibalize purchases that would have occurred absent the program~\citep{bawa2004effects,norvell2017gift,norvell2018bonus}. These works, however, do not provide a mathematically rigorous means of inferring counterfactual customer behavior absent a gift card, which is the primary focus of this work.

Several papers study the profitability or optimal design of gift-card programs. \citet{khouja2011analysis}, \citet{khouja2015channel}, and \citet{khouja2016effects} develop analytical models of gift-card pricing, redemption, sales thresholds, channel design, and inventory decisions. These papers characterize optimal firm policies under structural assumptions on demand and redemption behavior, whereas our goal is to identify and estimate the \textit{causal effect} of receiving a gift card on customer spend using customer-level data. Notably, our results in the sequel do not require making structural assumptions on a customer's utility model or decision making process. \citet{norvell2017gift, norvell2018bonus} propose survey-based approaches for decomposing gift-card revenue into incremental purchases, incremental spend, and cannibalization. These studies articulate the quantities that motivate our work (namely, incremental revenue as an object of study), but their measurement strategy relies on self-reported counterfactual behavior.
On the other hand, we not only provide two highly-interpretable, scale-free estimands for gauging incrementality, but also develop a rigorous causal identification/estimation strategy based on realized spend amounts instead of self-reported behavior.

The closest work in terms of methodology  is \citet{kadiyala2024dual}, who study delayed-incentive gift-card promotions at a major U.S.\ department store. In their setting, customers receive promotional emails (the treatment) that allow them to earn a future gift card by spending above tiered thresholds during a qualification window. In this work, the authors exploit discontinuities in the retailer's email targeting rule to estimate local average treatment effects (LATEs) of receiving the promotional email on downstream outcomes such as purchase probability and expenditure amongst others. The authors leverage fuzzy regression discontinuity methods to identify and estimate their targets. Their analyses show that such promotions can (1) generate incremental sales through acting as an advertisement/reminder to customers, and (2) can generate sales via gift card redemption. Our paper studies a very different problem: rather than estimating LATEs of exposure to a promotional offer among customers near targeting thresholds, we aim to estimate program level incrementality measures for actual gift-card recipients. In particular, in our setting, gift-card possession is endogenously determined and systematically censored in the observational population.  We thus must develop a data-fusion approach that combines a treatment-censored observational dataset with an auxiliary experimental one in order to identify the causal effects of interest.

\paragraph{Causal Inference Under Data Fusion} 
Methodologically, our paper contributes to the literature on causal inference with multiple imperfect data sources. In the structural causal model literature, the concepts of \textit{transportability} and \textit{data fusion} ask when a causal query in a target population can be identified by combining data collected under different populations, interventions, or sampling regimes~\citep{bareinboim2016causal}. These approaches often use selection diagrams, or causal directed acyclic graphs (DAGs) augmented with nodes indicating which structural mechanisms may differ across environments, to derive graphical and algorithmic criteria for when such transfer is possible. Also, related is the literature on \textit{generalizability}, or the transport of results or parameters from a randomized controlled trial (RCT) to a target population of interest~\citep{stuart2011use,dahabreh2019generalizing}. These works typically assume transferability/transportability of  outcome regressions from the RCT to the target population for the sake of identification. In particular, they do not permit treatment missingness as is considered in this paper.  We point to \citet{colnet2024causal} for a recent review of data fusion, which discusses weighting methods, outcome-modeling, and doubly robust methods for improving external validity.

Our work can be viewed as complementary to the literature on surrogate outcomes, which focuses on the use of short-term outcomes to identify and estimate long term causal effects. \citet{prentice1989surrogate} introduces the notion of surrogacy, in which the impact of a treatment on a long-term outcome of interest is entirely captured (up to exogenous variation) by the measurement of a short-term outcome (deemed the surrogate). 
Following this original contribution, a rich literature focused on applying and extending the notions of surrogates has developed~\citep{day1996trial, begg2000use, frangakis2002principal}.
Relevant to the present paper are the contributions of \citet{athey2025surrogate}, who define the \textit{surrogate index} as the regression of long-term outcomes onto short-term surrogates and covariates. The authors show that under assumptions of unconfoundedness, surrogacy (in the sense of Prentice), and external validity, one can identify the long term ATE via the surrogate index, and thus can define natural de-biased estimators. 
\citet{ghassami2022combining} and \citet{hu2022identification} consider settings similar to \citet{athey2025surrogate}, except where long-term observational outcomes may suffer from unobserved confounding, with each work making a distinct set of assumptions to ensure identifiability of the target causal effects on the experimental data.
\citet{price2018estimation} consider a complementary problem to the estimation of treatment effects with short-term surrogates, developing methods for determining the optimal surrogate from a vector of candidates.
Recent work due to \citet{imbens2025long} extends the identification of long term outcomes to settings where there is persistent, latent confounding impacting both long and short term outcomes. The authors attack the confounding by leveraging time-series structure in multiple, short term outcomes.
Lastly, \citet{kallus2025role} discuss how surrogates can be useful in settings where there are limited long-term outcomes. They investigate how the quality of the surrogates impacts the efficiency of their estimators. 
Unlike the aforementioned works, our paper occupies an understudied niche in the data fusion/missing data literature where the \textit{treatment} is the partially-observed variable in the target population. This setting of our paper can be thought of as the mirror image of the setting for papers focused on surrogates: our methods require a historical experiment with fully-observed treatments to identify a causal effect in observational data with partially-observed treatments. 
Heuristically, our \textit{transferability} assumption, which posits that the ratios of certain conditional probabilities are shared between the experimental and observational populations, can be viewed as an analogue of the surrogacy assumption. 
Nonetheless, our results and identification argument cannot be directly embedded in the frameworks of any of the aforementioned works.

\paragraph{Causal Inference With a Missing Treatment and Causal Estimation}

Our paper is not the first to study the estimation of causal effects under treatment missingness or errors in treatment measurements. \citet{lewbel2007estimation} studies average treatment effect (ATE) estimation  when a binary treatment is observed with \textit{misclassification}, meaning that there is some probability that the treatment bit is flipped upon observation. They show that the ATE can be identified using an instrument or a second mismeasured treatment proxy together with conditional independence restrictions. \citet{braun2017propensity} study the use of propensity score methods in a similar setting of misclassified treatment assignment. They present a likelihood-based approach for correcting the bias from misclassification. Unlike these two approaches, our setting involves \textit{partially-observed treatments}, with missingness being directly tied to a customer's booking decision. A more modern approach to studying causal inference with a missing treatment is explored in \cite{zhou2024causal}. In this work, the authors consider a setting where treatment is never observed and instead two proxy variables are measured. It is assumed that these proxy measurements have no effect on the outcome conditional on the missing treatment, and it is also assumed that they are conditionally independent of one another. Under some additional causal structural assumptions, the authors prove  that one can use such proxies to identify the causal effect of the missing treatment. However, this structure is not present in our current context, so we must develop a new identification strategy.

Lastly, we mention that our work builds upon the rich literature focused on causal inference, missing data, and data fusion that has developed over the past few decades. We first mention the foundational works on semi-parametric estimation~\citep{kosorok2008introduction, newey1994asymptotic, van2000asymptotic,robins1995analysis,robins1995semiparametric,bickel1993efficient, levit1976efficiency}, doubly-robust estimation~\citep{tsiatis2006semiparametric}, targeted maximum likelihood estimation~\citep{van2006targeted, van2011targeted, van2014targeted}, and double machine learning (DML)/cross-fitting~\citep{chernozhukov2018double, chernozhukov2022automatic, chernozhukov2022long, chernozhukov2022riesznet, chernozhukov2022nested, chernozhukov2023automatic}. Our data fusion estimator, which we develop in Section~\ref{sec:estimation}, is primarily based on the latter DML literature, and we use de-biasing techniques from \citet{chernozhukov2018double} to construct our cross-fitting estimator in Section~\ref{sec:estimation}. Because our estimator is defined in terms of two samples, an experimental dataset and an observational one, we need to leverage results on de-biasing under covariate shifts, for which \citet{chernozhukov2023automatic} and \citet{guerdan2025doubly} are two examples, with the analysis of our estimator being primarily based on the former of these two works.   While we explicitly characterize all additional nuisance functions (termed Riesz representers in the DML literature) needed for de-biasing  in our work and leverage plug-in estimates, we note that tools from the literature on automatic de-biased learning could have been used to estimate these components without direct characterization~\citep{chernozhukov2022automatic, chernozhukov2022riesznet}.

\section{Empirical Context}
\label{sec:context}

The goal of this section is to provide a concrete motivation for the abstract identification arguments we will make in the following sections. We develop our causal estimation framework with the goal of evaluating the incrementality of Airbnb's gift card program.
Airbnb offers an extensive gift card program across multiple distributional channels and geographies. In our work, we aim to assess the incrementality of the program within North America (NA) across the business's four largest channels: first-party, retail online, retail in-store, and business-to-business (B2B) sales. In particular, we measure the incrementality (via the IC and IR) on both a month-to-month basis (capturing temporal trends in customer gift card redemption behavior) and on an aggregate basis averaged over an entire year as well. We further decompose our estimates into \textit{intensive} and \textit{extensive} shares, which serves to shed light on the mechanisms through which gift cards generate incremental revenue. 
We formally define all relevant quantities, estimands, and variables in Sections~\ref{sec:estimands} and \ref{sec:model} below.

To enable our analysis, we collect rich observational customer spend data for existing Airbnb customers across all months of the year 2024. For each month, our data consists of roughly 30,000 customers who have booked with a gift card, 3.5 million customers who have booked without a gift card, and 3 million customers who haven't booked at all. We note that, due to the large number of customers at Airbnb, our dataset is subsampled from a much larger super-population, and that all estimates discussed in the sequel are adjusted via Horvitz-Thompson inverse propensity weighting to account for the fact the data has been subsampled. The monthly realized sample sizes are noted in Table~\ref{tab:obs_monthly_counts} in Appendix~\ref{app:case_study:estimation}. We reveal neither the super-population size nor the subsampling rates.

In this monthly observational data, the outcome of interest is a customer's \textit{gross booking value (GBV)}, or the amount spent by a customer when booking a property through Airbnb. The endogenously-determined treatment is correspondingly \textit{whether or not a customer possessed a gift card} at the time of their first purchase in the target month.\footnote{For customers who do not make bookings, treatment status is determined by whether or not they possessed a gift card at all in the target period.} We note that it is generally not possible to observe the treatment status of all customers. Under the mild assumption that a customer who possesses a gift card will leverage it conditional on making a purchase, we will observe the gift card receipt status of units who make bookings. However, for the customers that do not make a purchase, it is not generally possible to observe gift card receipt status.
This presents an interesting yet unstudied case of missing data in which treatments are systematically masked from the business on a subset of the population.

Despite this partial observation problem, Airbnb can still reason about customer booking behavior through historical small scale experiments.
In our analysis, we leverage an experiment from 2022 that consisted of 26,980 existing customers who were selected according to the screening detailed at the end of this section. Of these participants, 50\% were assigned a gift card (the treatment group) and 50\% were not (the control group).
The treated units received an email notification detailing that they had received a free gift card from Airbnb. Half of the treated group received a digital gift card worth \$100 while the other received a card with face value of \$200.
The control group received no change to their experience at Airbnb and received neither notification nor gift card value.
The spending and booking decisions of the customers were then monitored over the course of the following year on a monthly basis --- we leverage the granularity of these measurements in the sequel. 

This particular structure of data at Airbnb --- a large observational dataset with systematically masked treatments and a small, fully-observed experiment --- motivates our identification strategy in the following sections. 

\paragraph{Additional Discussion About the Experiment}

We now provide greater detail on Airbnb's small scale experiment from 2022. In this experiment, participants were taken to be units who appeared likely to redeem a gift card.
In more detail, a \textit{focal group} was constructed, consisting of all users who had redeemed a gift card between June, 2018 and the end of December, 2021. These users were taken from a mix of both the NA (North America) and EMEA (Europe, Middle East, and Africa) geographies, with 88\% of units in the pre-selection population from NA and 12\% from EMEA. A corresponding \textit{complement group} was constructed from the same geographies, consisting of customers who made a booking during the aforementioned period and did not redeem any gift cards. A classifier was then trained using over 500 features to discriminate between these two populations.\footnote{The classifier trained was an \texttt{XGBoost} classifier with $\ell_2$-regularized logistic regressions at each leaf} The 26,980 units with the highest predicted probabilities were then chosen as candidate controls for the experiment, receiving treatment through uniform randomization as noted above.

We note that there are three primary ways in which the above experiment deviates from an ideal experiment that one would wish to leverage for our research application.
First, units considered for the experiment were subject to the pre-selection procedure noted above, which implicitly selected for units that seemed likely to redeem a gift card.
If we were able to rerun an experiment, we would remove this pre-selection step, only filtering to include recently active users (say, individuals who have logged into their accounts in the past three years) in the experiment.
Second, as noted above, the experiment leverages a mix of units from both the NA and EMEA geographies, with 88\% of units in the pre-selection population from NA and 12\% from EMEA. 
Ideally, one would run a separate experiment for each geographical region.
Lastly, one might aim to develop a method for treatment that better captures phenomena like self-gifting.
For instance, one could imagine treating some fragment of the population by offering, say, a 5\% discount on whichever property they booked. This would simulate self-gifting through third-party channels such as B2B and Retail Online, wherein units may be able to redeem points or obtain discounted gift cards.
However, because the experiment was designed for a different purpose, we view it as a best available approximation.

\section{Target Causal Estimands}
\label{sec:estimands}

We now abstractly motivate estimands that capture precisely the incrementality of a business's gift card program. We note that although our language naturally focuses on gift card programs in the sequel, the estimands and partial observation pattern we consider can naturally be applied to other business settings, with third-party coupons, partner issued credits, and referral programs being notable examples. 

Consider a population of customers interacting with the firm and let $Y$ denote the observed spend for each customer. In the status quo, where the gift card program is active, the firm collects total expected revenue $\E[Y]$ per customer. Now, consider a counterfactual reality in which the gift card program did not exist. In this setting, every customer who had received and not yet redeemed a gift card would instead generate some \emph{counterfactual} revenue. This quantity, which we denote in potential outcome notation by $Y(0)$, captures the counterfactual spend that the customer would have generated \textit{without} the card. For some customers, it may be unchanged (i.e., $Y = Y(0)$) while for some others we may have $Y(0) = 0$ whereas $Y > 0$, implying the absence of the gift card program caused the customer to not spend any money at the business. We operate under the assumption that the behavior of non-recipients would remain unchanged.\footnote{This statement implicitly rules out equilibrium or brand effects arising from the mere existence of the program; such secondary effects are ignored in our analysis. Formally, we invoke the \textit{stable unit treatment value assumption} (SUTVA): the potential outcomes for each customer are independent of the treatment assignments of other customers \citep{rubin1980randomization}.}

The expected per-customer revenue in the case of such a shutdown can be expressed as $\E[Y(0) \cdot G + Y \cdot (1 - G)]$, where $G \in \{0,1\}$ indicates a customer's gift card receipt status. We always assume that $Y=Y(0)$ when $G=0$. It thus follows that the \textit{incremental revenue}  of the program can be expressed as the following difference $\Delta_0$:
\begin{align}
\label{eq:program-value}
\Delta_0 = \E[Y] - \E[Y(0) \cdot G + Y \cdot (1-G)] = \E[(Y - Y(0)) \cdot G] 
\end{align}
While incremental revenue $\Delta_0$ gives the raw, per-customer dollar amount gained or lost due to the shut down of the gift card program, we prefer in the sequel to work with \textit{scale-free} quantities that do not involve dollar amounts. 
We thus leverage $\Delta_0$ as a building block on which to develop more actionable estimands.

One scale-free way of determining the incrementality of a business's gift card program would be to measure the expected change in revenue for each dollar of gift card value redeemed. This quantity, which we call the \emph{incrementality coefficient} (IC) and denote by $\kappa_0$, is defined formally as follows:
\begin{align}\label{eqn:ic}
\kappa_0 := \frac{\E[(Y - Y(0))\cdot G]}{\E[A\cdot G]} = \frac{\Delta_0}{\E[A \cdot G]}.
\end{align}
Here, $A$ is a random variable capturing the \textit{realized gift card spend} of a customer, and is naturally positive when the customer makes a purchase with a gift card and zero otherwise. The IC naturally plays a role in the computation of a business's profits under a toy model: if the firm incurs a cost per-dollar of gift-card transacted $c$, and keeps a margin $m$ for each dollar spent on the platform, then the firm's profit takes on the form\footnote{We note that this simple model does not capture elements of profitability calculation such as breakage, costs for alternative forms of payments, and more. We simply provide this toy model to aid in the interpretation of our estimand.}
\begin{align}
\mathrm{Profit}  = m\cdot \E[(Y - Y(0))\cdot G] - c \cdot  \E[A \cdot G] = m \cdot \Delta_0 - c \cdot \E[A \cdot G].
\end{align}
Thus, under this simple model profitability can be determined by estimating the IC and checking whether it surpasses the ratio $c/m$. \emph{The IC can therefore be heuristically interpreted as the maximal cost-per-gift-card-dollar to margin ratio for which the program is profitable.} 
Naturally, a precise computation of profitability requires a more bespoke analysis and is not the focus of this paper.

Beyond the per-dollar IC, it is useful to express the causal effect as a fraction of total revenue generated by gift card holders. 
We define the \emph{Incrementality Ratio} (IR) as the fraction of incremental revenue attributable to the existence of the gift card program:
\begin{align}\label{eqn:ir}
\rho_0 := \frac{\E[(Y - Y(0))\,G]}{\E[Y\cdot G]} = \frac{\Delta_0}{\E[Y \cdot G]}.
\end{align}
The IR has a natural interpretation in terms of cannibalization: $1 - \rho_0$ represents the fraction of revenue generated from gift card holders that corresponds to cannibalized spend, or the fraction of the revenue the firm would have earned anyway. Thus, $\rho_0$ can be viewed as the fraction of generated revenue that is truly incremental. An IR of zero means the gift card generates no additional revenue (complete cannibalization), while an IR of one means all spend by gift card holders is incremental. We note that there is no restriction that the IR be positive: there may be settings where $Y < Y(0)$ for many customers.

\section{Identifiability}
\label{sec:model}

We now formalize the assumptions under which the incrementality coefficient (IC) $\kappa_0$ and the incrementality ratio (IR) $\rho_0$ can be identified. The identification of these quantities ultimately depends on identifying the incremental revenue $\Delta_0 := \E[(Y - Y(0))G]$. 
First, in Section~\ref{sec:model:full}, we present the standard identification for the parameters under full observability of treatment. Then, in Section~\ref{sec:model:fusion}, we show how a small auxiliary experimental dataset can be used to recover identification when the treatment is censored. The remainder of the section is then focused on defining intensive and extensive margins (Section~\ref{sec:model:margins}) and discussing how gift card type-heterogeneity may play a role in identification (Section~\ref{sec:model:heterogeneity}).

\subsection{Notation and Identification under Full Observability}
\label{sec:model:full}

We model each customer's interaction with the firm as a tuple $Z = (X, G, B, Y)$, drawn from an \textit{observational distribution} $\P_{\obs}$. Here, $X \in \calX$ denotes a vector of customer characteristics. In our application to Airbnb, these include historical purchasing tendencies, past gift card usage, and other features capturing spending habits. The variable $G \in \{0, 1\}$ indicates whether the customer received a gift card, $B \in \{0, 1\}$ records whether the customer made a purchase (a booking in the case of Airbnb), and $Y \in \R_{\geq 0}$ denotes total spend. We assume throughout that customers who do not purchase spend nothing. That is, we assume $Y = 0$ whenever $B = 0$. When relevant, we will augment this tuple with a value $A \geq 0$, which represents the number of gift card dollars redeemed by a customer.

\begin{figure}[ht]
\centering
\begin{subfigure}[b]{0.45\textwidth}
\centering
\begin{tikzpicture}[node distance=1cm and 1.5cm, line width=0.7pt]
    \node[obs] (X) at (0, 2.5) {$X$};
    \node[obs] (G) at (-1.8, 0) {$G$};
    \node[obs] (B) at (0, 0) {$B$};
    \node[obs] (Y) at (1.8, 0) {$Y$};
    
    \draw[arr] (X) -- (G);
    \draw[arr] (X) -- (B);
    \draw[arr] (X) -- (Y);
    \draw[arr] (G) -- (B);
    \draw[arr] (G) to[bend right=40] (Y);
    \draw[arr] (B) -- (Y);
\end{tikzpicture}
\caption{Fully-observed data.}
\label{fig:dag_full}
\end{subfigure}
\hfill
\begin{subfigure}[b]{0.45\textwidth}
\centering
\begin{tikzpicture}[node distance=1cm and 1.5cm, line width=0.7pt]
    \node[obs] (X) at (0, 1.5) {$X$};
    \node[lat] (G) at (-1.8, 0) {$G$};
    \node[obs] (B) at (0, 0) {$B$};
    \node[obs] (Y) at (1.8, 0) {$Y$};
    \node[obs] (S) at (-0.0, -1.8) {$S$};
    
    \draw[arr] (X) -- (G);
    \draw[arr] (X) -- (B);
    \draw[arr] (X) -- (Y);
    \draw[arr] (G) -- (B);
    \draw[arr] (G) to[bend right=30] (Y);
    \draw[arr] (B) -- (Y);
    \draw[arr] (G) -- (S);
    \draw[arr] (B) -- (S);
\end{tikzpicture}
\caption{Observational data.}
\label{fig:dag_obs}
\end{subfigure}
\caption{Causal directed acyclic graphs (DAGs) for both the fully-observed data and the partially-observed observational data.}
\label{fig:causal_dag}
\end{figure}
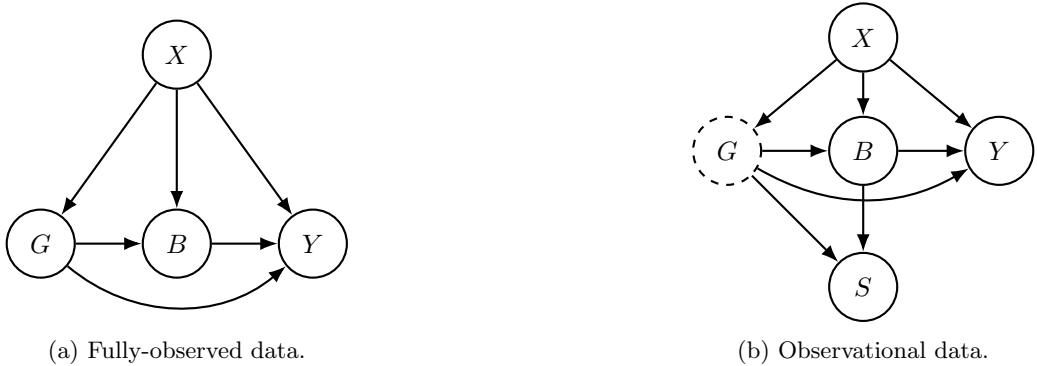

Under full observability of $G$, all three estimands are identified under two standard assumptions. The first is a one-sided conditional exogeneity condition.

\begin{ass}[One-Sided Conditional Exogeneity]
\label{ass:exogeneity}
Under $\P_{\obs}$, we have $Y(0) \independent G \mid X$.
\end{ass}

This states that, conditional on observed characteristics, the potential spend without a gift card is independent of whether a customer actually received one. We emphasize that this is a \textit{one-sided} condition: we only require independence of the untreated potential outcome $Y(0)$, which is strictly weaker than the two-sided condition $\{Y(0), Y(1)\} \independent G \mid X$ needed for identification of causal effects such as the average treatment effect (ATE). This weaker requirement is natural in our setting, as we only need to impute counterfactual outcomes for the treated group. 
In the context of our Airbnb case study (Section~\ref{sec:case_study}), we take extensive steps to ensure this assumption is credible. We construct over 150 covariates designed to capture potential confounding, which we discuss in greater detail in our empirical section. 

The second assumption is a one-sided positivity condition. This requires that, for any covariate profile observed among gift card recipients, there exist comparable customers who did \textit{not} receive a card. In a business context, this condition would fail only if certain customer segments receive gift cards with probability one, which is unlikely in practice.

\begin{ass}[One-Sided Positivity]
\label{ass:positivity}
Under $\P_{\obs}$, we have $\P_{\obs}(G = 0 \mid X = x) > 0$ for almost every $x$ in the support of $X \mid G = 1$.
\end{ass}

Under Assumptions~\ref{ass:exogeneity} and~\ref{ass:positivity}, the incremental revenue $\Delta_0$ is identified by the standard $g$-formula~\citep{hahn1998role, imbens2004nonparametric}:
\begin{equation}
\label{eq:fully-observed-ident}
\Delta_0 = \E_{\obs}[Y G] - \E_{\obs}[\mu_0(X) G] = \E_{\obs}[GY] - \E_{\obs}[\pi_0(X)\mu_0(X)],
\end{equation}
where $\pi_0(x) := \P_{\obs}(G = 1 \mid X = x)$ is the propensity score and $\mu_0(x) := \E_{\obs}(Y \mid G = 0, X = x)$ is the outcome regression for the untreated group. Both are clearly identified when $G$ is fully observed. Noting that we can write
\begin{gather}
\label{eq:full-identification}
\kappa_0 = \frac{\Delta_0}{\E_{\obs}[AG]} \qquad \text{and} \qquad
\rho_0 = \frac{\Delta_0}{\E_{\obs}[GY]},
\end{gather}
it thus follows that all three target estimands are identified in the fully observed setting.

\subsection{Data Fusion under Censored Treatment}
\label{sec:model:fusion}

In practice, however, the firm does not observe $G$ for all customers. Gift card receipt is only revealed when a customer actually makes a purchase.\footnote{We note that many companies actually offer a gift card ``claim'' step that involves uploading the card's value to their platform. We lump this stage together with redemption/purchase for simplicity and to avoid pedantry.} For non-purchasing customers, their gift card status remains unknown. We formalize this as follows.

\begin{ass}[Partial Observation of Treatment]
\label{ass:partial-obs}
The learner observes $W = (X, S, D, B, Y)$ in samples from $\P_{\obs}$, where $S := B \cdot G$ is the indicator that a customer purchases with a gift card and $D:=B(1 - G)$ is the indicator that a customer purchases without a gift card.
\end{ass}

Again, we may augment these tuples with the variable $A \geq 0$ representing the number of gift card dollars redeemed with a purchase when this information is available.
Under this censoring, the firm observes $S$ (the indicator of a customer purchase with a gift card) and $D$ (purchase without a gift card), but never $G$ on its own.\footnote{Implicit in this assumption is that a unit who possesses a gift card and makes a purchase will always leverage the gift card in said purchase.} This obstructs the identification of both nuisances in~\eqref{eq:fully-observed-ident}. The propensity $\pi_0(x)$ is unidentified because, while we can count gift card \textit{redeemers}, we cannot count the gift card \textit{holders} who chose not to transact. The outcome regression $\mu_0(x)$ is partially obstructed: since $Y = 0$ when $B = 0$, we know that:
$$\mu_0(x) = \E_{\obs}(Y \mid G = 0, X = x)=\E_{\obs}(Y \mid G = 0, B = 1, X = x)\cdot \P_{\obs}(B=1\mid G=0, X = x).$$
Even though we can identify $\E_{\obs}(Y \mid G = 0, B = 1, X = x)$ from the observed data, we cannot determine the probability an ungifted customer purchases without knowing $G$ amongst non-purchasers. 

In our approach, we show that we can overcome the identification challenge if we also have access to samples from an auxiliary \textit{experimental dataset}, in which treatment status is fully recorded. This dataset should be thought of as small relative to the observational one. For instance, in our case study with Airbnb, the experimental dataset only has around 26,000 samples, whereas the observational data have millions.
While our identification arguments naturally only depend on the distributions over the experimental and observational samples, our estimators in the sequel will maintain valid coverage even if the number of experimental samples grows much more slowly than the number of observational ones. Formally, we assume the firm has access to independent samples from an experimental distribution $\P_{\exp}$, where observations take the same form $Z = (X, G, B, Y)$ and $G$ is observed for all customers regardless of whether or not they made a purchase.

We emphasize that the experimental data need not come from the same population as the observational data, and we do not assume that the covariate $X$ or outcome $Y$ distributions coincide across the two populations.\footnote{In fact, for our methodology, observing the spending outcome $Y$ in the experimental dataset is optional.} We only require that the \textit{ratio of conditional purchase propensities} is consistent between the observational and experimental populations, which we formalize as our main assumption.

\begin{ass}[Transferable Risk Ratios]
\label{ass:transfer}
We assume that, for almost every $x \in \calX$, we have
\[
\frac{\P_{\obs}(B = 1 \mid G = 0, X = x)}{\P_{\obs}(B = 1 \mid G = 1, X = x)} = \frac{\P_{\exp}(B = 1 \mid G = 0, X = x)}{\P_{\exp}(B = 1 \mid G = 1, X = x)}.
\]
Further, we assume $\P_{\exp}(B = 1 \mid G = 1, X = x) > 0$ for almost every $x$ in the support of $X \mid G = 1$ under $\P_{\obs}$, and that $\P_{\exp}(G = g \mid X = x) > 0$ for each $g \in \{0,1\}$ and almost every $x$ under $\P_{\exp}$.
\end{ass}

We note that one way that the above assumption would be satisfied is under the assumption that the conditional purchase probabilities in the observational dataset are an $X$-dependent multiple of the analogous probabilities under the experimental distribution, i.e., that
\begin{equation}
\label{eq:f_transfer}
\P_{\obs}(B = 1 \mid G = g, X = x) = f(x) \cdot \P_{\exp}(B = 1 \mid G = g, X = x).
\end{equation}
for all $ g \in \{0, 1\}$ and almost every $x \in \calX$. In the special case that $f(x) \equiv 1$, this corresponds to the conditional probabilities being the same between the two populations. More broadly, the multiplicative function $f(x)$ allows for heterogeneous differences in purchasing behaviors between the populations, so long as these differences do not depend on gift card receipt status. These systematic fluctuations could capture geographic or temporal differences in purchasing behavior.
These assumptions may be violated, for instance, if customers under control purchased at the same rate between the two populations, but observational units had significantly higher conditional purchase rates.

We now show that, under the above assumptions, our causal estimands remain identified despite the censored treatment. The key idea is that the experimental data provides precisely ``missing pieces'' needed to complete the standard identification formula.\footnote{Our data fusion approach shares some similarities but is distinctly different than the recent literature on using surrogates and data fusion for estimating long-term effects. We discuss the relationship in Appendix~\ref{app:surrogates}.} In particular, we define several nuisance functions that collectively bridge the gap between what is observed and what is needed. We start by providing notation for key conditional purchasing probabilities:
\begin{align*}
p_0(x) &:= \P_{\exp}(B = 1 \mid G = 0, X = x) &&\text{(experimental conditional ungifted purchase rate),} \\
q_0(x) &:= \P_{\exp}(B = 1 \mid G = 1, X = x) &&\text{(experimental conditional gifted purchase rate).} 
\end{align*}
We also define the organic, counterfactual spend amount for units who purchase without a gift card:
\begin{align*}
g_0(x) &:= \E_{\obs}(Y \mid G = 0, B = 1, X = x) && \text{(avg.\ spend of those who purchase without a card)}.
\end{align*}
We emphasize that $p_0(x)$ and $q_0(x)$ are identified from the experimental distribution, which is fully observed. Thus, by Assumption~\ref{ass:transfer}, the observational risk ratio $\P_{\obs}(B = 1 \mid G = 0, X = x)/\P_{\obs}(B = 1 \mid G = 1, X = x)$ is identified as well. The counterfactual spend amount, $g_0(x)$, is identified directly from the observational data, since the event $\{D = 1\}$ (purchasing without a gift card) is observed even under the censoring described in Assumption~\ref{ass:partial-obs}.

\begin{theorem}[Identification of the IC and IR]
\label{thm:identification}
Suppose Assumptions~\ref{ass:exogeneity}--\ref{ass:transfer} hold. Then, the incremental revenue of the gift card program is identified as
\begin{align}
\label{eq:ident}
\Delta_0 = \E_{\obs}[YS] - \E_{\obs}\left[\tfrac{g_0(X)p_0(X)S}{q_0(X)}\right]
\end{align}
Consequently, the incrementality ratio and incrementality coefficient are respectively identified as
\[
\rho_0 = \frac{\E_{\obs}[YS] - \E_{\obs}\left[\tfrac{g_0(X)p_0(X)S}{q_0(X)}\right]}{\E_{\obs}\left[Y S\right]} \quad \text{and} \quad 
\kappa_0 =\frac{\E_{\obs}[YS] - \E_{\obs}\left[\tfrac{g_0(X)p_0(X)S}{q_0(X)}\right]}{\E_{\obs}[A S]}.
\]
\end{theorem}

\begin{proof}
Our strategy is to show that the components of the fully observed identification formula outlined in Equation~\eqref{eq:fully-observed-ident}, namely $\E_{\obs}[GY]$ and $\E_{\obs}[\pi_0(X)\mu_0(X)]$, can be re-expressed in terms of the observed censored data and the three identified nuisance functions $p_0(x), q_0(x)$, and $g_0(x)$.

\textit{Step 1: Identifying the treated outcome $\E_{\obs}[GY]$.} The first component of the identification formula is the average spend among gift card recipients. Although we do not observe $G$ directly, we can recover this term by exploiting the zero-spend assumption. A customer contributes to $GY$ only if they both received a gift card \textit{and} spent a positive amount. But spending a positive amount requires making a purchase ($B = 1$). Thus $GY = G \cdot B \cdot Y = S \cdot Y$, and we have
\[
\E_{\obs}[GY] = \E_{\obs}[YS].
\]
In words: the total revenue from gift card holders equals the total revenue from gift card \textit{purchasers}, because holders who do not purchase contribute zero revenue.

\textit{Step 2: Identifying the counterfactual spend $\E_{\obs}[\pi_0(X)\mu_0(X)]$.} The next component is the expected counterfactual spend that gift card recipients would have generated without the program. Recall we had defined the outcome regression under control as $\mu_0(x) = \E_{\obs}(Y \mid G = 0, X = x)$. Since customers who do not purchase spend nothing, we can decompose this regression by conditioning on the purchasing decision $B$:
\[
\mu_0(x) = \underbrace{\E_{\obs}(Y \mid G = 0, B = 1, X = x)}_{= \, g_0(x)} \cdot \P_{\obs}(B = 1 \mid G = 0, X = x).
\]
In words, a customer's expected spend without a gift card equals their expected spend \textit{given they purchase} (the organic spend $g_0$, observed in the data) times their probability of purchasing without a card.

Next, note that Bayes rule allows us to write the observational propensity $\pi_0(x)$ as 
\[
\pi_0(x) = \P_{\obs}(G = 1 \mid X = x) = \frac{\E_{\obs}(S \mid X = x)}{\P_{\obs}(B = 1 \mid G = 1, X = x)}.
\]

These two observations combined allow us to write
\begin{align*}
\E_{\obs}[\pi_0(X)\mu_0(X)] &= \E_{\obs}\left[\frac{\E_{\obs}(S \mid X)}{\P_{\obs}(B = 1 \mid G = 1, X)} \cdot g_0(X) \P_{\obs}(B = 1 \mid G = 0, X )\right] \\
&= \E_{\obs}\left[\frac{g_0(X)p_0(X)S}{q_0(X)}\right],
\end{align*}
where the final equality follows directly from our transferability assumption (Assumption~\ref{ass:transfer}).

Combining the results of Step 1 and Step 2 yields precisely the identification for $\Delta_0$ outlined in Equation~\eqref{eq:ident}. Identification for the IR $\rho_0$ similarly follows as the denominator is just the first term in $\Delta_0$. Lastly, for the IC, we simply observe that $A=0$ whenever $B=0$ and therefore $AG=ABG=AS$. Thus the denominator in Equation~\eqref{eqn:ic} is identified in the observational data as $\E_{\obs}[A\,S]$, assuming $A$ is observed for every purchase. 
\end{proof}

\subsection{Intensive and Extensive Margins} 
\label{sec:model:margins}
Using our identification formula, we can arrive at a revealing decomposition that separates the gift card's effect on \textit{how much} customers spend (intensive margin) from its effect on \textit{whether} they purchase at all (extensive margin). To formalize this decomposition, we first make the following assumption.
\begin{ass}
\label{ass:cross_positivity}
We assume that $\P_{\obs}(D = 1 \mid X = x) > 0$ for almost all $x \in \calX$ satisfying $\P_{\obs}(S = 1 \mid X = x)$.
\end{ass}
\noindent Further, we define the conditional purchasing probabilities
\begin{align*}
p_{\obs}(x) := \P_{\obs}(B = 1 \mid G = 0, X = x) \quad \text{and} \quad 
q_{\obs}(x) := \P_{\obs}(B = 1 \mid G = 1, X = x)
\end{align*}
and let the risk ratio be given as $\eta_0(x) := p_0(x)/q_0(x) \equiv p_{\obs}(x)/q_{\obs}(x)$. Note that by adding and subtracting $\E_{\obs}[S g_0(X)]$ (which is well-defined by Assumption~\ref{ass:cross_positivity}), we can rewrite the identification formula for the incremental per-customer revenue as
\begin{align}
\label{eq:decomp_weighted}
\Delta_0 = \underbrace{\E_{\obs}\left[S\left\{Y - g_0(X)\right\}\right]}_{=: \Delta_{0, \intm}} + \underbrace{\E_{\obs}\left[Sg_0(X)\left\{1  - \eta_0(X)\right\}\right]}_{=: \Delta_{0, \extm}},
\end{align}
where $\Delta_{0, \intm}$ and $\Delta_{0, \extm}$ represent intensive and extensive margins, respectively. We shall now justify this terminology. First, note that we can rewrite the intensive component/share of incremental revenue as
\[
\Delta_{0, \intm} = \left\{\E_{\obs}[Y \mid S = 1] - \E_{\obs}[g_0(X) \mid S = 1]\right\}\P_{\obs}(S = 1),
\]
which thus represents the average \textit{lift in out of pocket spend} for those who purchase with a gift card multiplied by the marginal probability of purchasing with a gift card. This very intuitively captures the fraction of incremental revenue associated with increased spend amongst those who would have purchased otherwise. Further, by noting that $1 - \eta_0(X) = \frac{q_{\obs}(X) - p_{\obs}(X)}{q_{\obs}(X)}$ and that $\E_{\obs}[S \mid X]/q_{\obs}(X) = \pi_0(X)$, we can write
\[
\Delta_{0, \extm} = \E_{\obs}\left[\pi_0(X)g_0(X)\left\{q_{\obs}(X) - p_{\obs}(X)\right\}\right].
\]
Thus, the extensive margin of incremental revenue is a weighted average (per the propensity $\pi_0(X)$) of the increase $q_{\obs}(x) - p_{\obs}(x)$ in conditional purchasing probability times the counterfactual expected spend of a unit without a gift card $g_0(X)$. In the simplified setting where neither $q_{\obs}$ nor $p_{\obs}$ depends on $X$, this simplifies to 
\[
\Delta_{0, \extm} = \E_{\obs}[\pi_0(X)g_0(X)]\left\{\P_{\obs}(B = 1 \mid G = 1) - \P_{\obs}(B = 1 \mid G = 0)\right\},
\]
which more directly captures the lift in conditional purchasing probability provided by gift card receipt. 

Heuristically, the \textit{intensive margin} $\Delta_{0, \intm}$ answers the following question: amongst customers who would have purchased irrespective of gift card receipt, how much additional revenue does the gift card generate? The \textit{extensive margin} $\Delta_{0, \extm}$ answers the following: amongst customers who purchase \textit{only because} of the gift card, what revenue do they bring in? This decomposition is informative as it reveals the distinct ways through which gift card programs can generate incremental revenue. If the intensive margin dominates, the program primarily encourages existing purchasers to spend more. If the extensive margin dominates, the program's value comes from converting non-purchasers into purchasers. 

We note that our definitions of intensive and extensive margins above just follow from algebraic manipulations. Under a mildly stronger conditional exogeneity assumption than the one outlined in Assumption~\ref{ass:exogeneity}, one can actually show that the intensive and extensive margins associated with the incremental revenue $\Delta_0$ are precisely the \textit{direct} and \textit{indirect} effects of receiving a gift card on customer spend, respectively. 
We make this connection to mediation analysis explicit in Appendix~\ref{app:margins}, where we provide a formal causal interpretation of these margins. 

\paragraph{Decomposing the IR and IC into extensive, intensive, cannibalization shares}

While $\Delta_{0, \intm}$ and $\Delta_{0, \extm}$ give a measurement of intensive and extensive margin, respectively, their magnitude depends on the overall scale of the incremental revenue. This consequently makes them difficult to interpret in a vacuum. It is thus natural to decompose both the incrementality coefficient (IC) $\kappa_0$ and the incrementality ratio (IR) $\rho_0$ into intensive and extensive shares, which we can then compute and interpret in our application to Airbnb.

For the IC, we naturally arrive at the decomposition:
\begin{align}\label{eqn:ic_decomp}
\kappa_0 = \underbrace{\frac{\E_{\obs}\left[S\left\{Y - g_0(X)\right\}\right]}{\E_{\obs}[A G]}}_{=: \kappa_{0, \intm}} + \underbrace{\frac{\E_{\obs}\left[\pi_0(X)g_0(X)\left\{q_{\obs}(X) - p_{\obs}(X)\right\}\right]}{\E_{\obs}[A G]}}_{=:\kappa_{0, \extm}}
\end{align}
Under this decomposition, $\kappa_{0, \intm}$ can be viewed as the expected change in out of pocket spend for each dollar of gift card value redeemed amongst those who would have purchased anyway. This interpretation is further cemented by noting that
\[
\kappa_{0, \intm} = \frac{\E_{\obs}[Y \mid S = 1] - \E_{\obs}[g_0(X) \mid S = 1]}{\E_{\obs}[A \mid S  =1]}.
\]
A bit more heuristically, $\kappa_{0, \extm}$ gives the expected change in spend per dollar of gift card value redeemed amongst those who purchased because of the gift card.

One can perform this decomposition equally well for the incrementality ratio, simply replacing $\E_{\obs}[AG]$ with $\E_{\obs}[Y G]$. Further, note the fact that $1 - \rho_0$ represents the \textit{cannibalization share (CS)}, or the fraction of revenue amongst gift card recipients that would have been generated without the program. 
We can thus decompose total revenue into an intensive incremental share ($\mathrm{IIS}$), an extensive incremental share ($\mathrm{EIS}$), and a cannibalized share ($\mathrm{CS}:=1 - \rho_0$), providing a complete accounting of every dollar of GBV generated by gift card holders. These shares are respectively identified as:
\begin{align*}
\mathrm{IIS} &= \frac{\E_{\obs}[Y \mid S = 1] - \E_{\obs}[g_0(X) \mid S = 1]}{\E_{\obs}[Y\mid S=1]} &
\mathrm{EIS} &= \frac{\E_{\obs}[g_0(X)\, (1 - p_0(X) / q_0(X))\mid S=1]}{\E_{\obs}[Y\mid S=1]}
\end{align*}
\begin{equation*}
\mathrm{CS} = \frac{\E[g_0(X)\, p_0(X) / q_0(X) \mid S=1]}{\E[Y\mid S=1]}
\end{equation*}
These numbers naturally sum to one. Note that there is no restriction that the various shares are positive, and in fact in some lower incrementality channels in our case study, we will see that intensive shares of revenue can actually be negative.

\paragraph{Data fusion through the lens of the decomposition.} This decomposition of incremental revenue into intensive and extensive margin also clarifies precisely which components require the auxiliary experimental data for identification and which are identified from the observational data alone. Consider the intensive margin in~\eqref{eq:decomp_weighted}: the term $\E_{\obs}[Y \mid S = 1]$ is the average spend among observed gift card purchasers, and is thus directly computable from the firm's transaction records. The term $\E_{\obs}[g_0(X) \mid S = 1]$ averages the organic spend function $g_0(x) = \E_{\obs}(Y \mid D = 1, X = x)$ over gift card purchasers, and since $g_0$ is estimated from non-gift-card purchasers in the observational data, this too requires no experimental information. Thus, the \textit{entire intensive margin is identified from observational data alone}.

The extensive margin, by contrast, fundamentally requires the experiment. The ratio between purchase rates $\eta_0(X) = \P_{\obs}(B = 1 \mid G = 0, X)/\P_{\obs}(B = 1 \mid G = 1, X)$ cannot be recovered from the observational data because $G$ is unobserved among non-purchasers. This ratio is precisely the quantity that our transferability assumption (Assumption~\ref{ass:transfer}) allows us to borrow from the experimental distribution.  In summary, the observational data identifies \textit{what} customers spend, while the experimental data identifies \textit{how often} they purchase, and it is the latter that the censoring obscures. In Section~\ref{sec:case_study}, we estimate both components and discuss their relative magnitudes.

\subsection{Incorporating Gift Card Heterogeneity}
\label{sec:model:heterogeneity}

Up until this point, we have just been considering a unit treated if they possess an unredeemed gift card. However, there are many ways in which possession of a gift card can be viewed as a \textit{multivariate treatment}. As one example, in our identification strategy above, we do not consider heterogeneity with respect to the amount of time between the purchase and redemption of the gift card. For instance, it is highly probable that a customer who received a gift card one month ago will have a different instantaneous purchasing rate than a customer who received a gift card six months ago, even if customer attributes $X$ were the same. Another example is gift card \textit{attribute heterogeneity}, which posits that different aspects of the gift card itself (e.g., face value) may impact the rate at which customers purchase. In this section, we extend our identification strategy to account for both purchase time and attribute heterogeneity simultaneously.

We must first revisit our transferability assumption. Transferability, as outlined in Assumption~\ref{ass:transfer}, ignores a key temporal distinction between the experimental and observational distributions. In the experimental distribution, all units receive a treatment assignment at the same time. Purchasing decisions are then monitored after treatment assignment at regular intervals (e.g., every few months). On the other hand, in the observational distribution treatment corresponds to holding unredeemed gift card value at the time of first purchase in the target period. In this setting, even if two observational units possess the same characteristics, if their corresponding times between gift card receipt and the start of the target period (for which we measure outcomes) differ, it may not be reasonable to assume that their conditional purchasing probabilities during that target period are the same. Thankfully, when a gift card is redeemed, the purchase date of the card is also generally revealed to the business, so the amount of time between gift card purchase and redemption is known. We can leverage this extra observation to introduce a meaningful notion of transferability in this setting. Intuitively, for an observational customer who holds a $t$ month old gift card of attribute $a$ that has yet to be spent, the natural object to transfer is the ratio between a unit's probability of purchasing given they do not possess a gift card and the \textit{conditional purchasing hazard rate} from the $t$ month mark of the experiment for cards with attribute $a$. In other words, we want the denominator of the transferred ratio to be the conditional probability that a customer who received a gift card of attribute $a$ at the start of the experiment purchased in month $t$ given they had not yet purchased in the first $t - 1$ months.

\paragraph{Definitions and Observations}
We begin by formalizing the role of time between gift card receipt and redemption, as well as the gift card attributes. Suppose the experiment was monitored over some evenly-spaced collection of discrete periods $\calT := \{1, 2, \dots, t_{\max}\}$. Canonically, we can think of these periods as months. Let $a \in \calA = \{1, \ldots, A\}$ denote the discrete set of possible gift card attributes (e.g., face value brackets, region of purchase, original merchant of distribution). Let $\calJ=\calT \times \calA$ and define the type of a card as the tuple $\tau = (t, a)\in \calJ$ that jointly captures both the elapsed time and the card attributes. 

For an experimental sample, we let $R \in \calT \cup \{\infty\}$ denote the random time at which a customer decided to purchase, with $R = \infty$ indicating that they have not purchased by the end of the experiment. We define $B_t := \1\{R = t\}$ as the indicator of an experimental customer's first purchasing event occurring at time $t$, and $B_{< t} := \sum_{s = 1}^{t-1} B_s$ as the indicator of a customer purchasing before time $t$. Moreover, we let $G_a$ denote the indicator that an experimental customer received a gift card with attribute $a$, and we assume $G_a = 1$ for at most one $a$.

\begin{ass}[Augmented Experimental Data]
\label{ass:obs-exp-time}
The learner observes $Z = (X, \{G_a\}_{a \in \calA}, B, R)$ in samples from $\P_{\exp}$, where $R$ is as defined above.
\end{ass}

For the observational population, we define $G_\tau$ to be the indicator that, at the beginning of the target period, a customer holds an unredeemed gift card of attribute $a$, purchased $t$ periods ago. We assume $G_\tau = 1$ for \textit{at most one} $\tau$. The overall gift card indicator is then $G = \sum_{\tau \in \calJ} G_\tau$. Moreover, we define the purchasing indicators:
\begin{align*}
S_\tau &:= B \cdot G_\tau, \quad \tau \in \calJ,  & 
D & := B \cdot (1 - G).
\end{align*}
Here $S_\tau$ is the indicator that a customer purchases using a gift card of type $\tau = (t, a)$, and $D$ is the indicator of purchasing without a gift card. Note that $S = \sum_{\tau \in \calJ} S_\tau$ recovers the overall gifted-purchasing indicator defined in the previous section.

\begin{ass}[Augmented Observational Data]
\label{ass:partial-obs-type}
The learner observes $W = (X, \{S_\tau\}_{\tau \in \calJ}, D, B, Y)$ in samples from $\P_{\obs}$.
\end{ass}

\paragraph{Heterogeneous Transferability via Hazard Functions}
Our identification strategy proceeds by transferring the experimental conditional hazard of purchasing for each augmented type $\tau = (t, a) \in \calJ$:
\begin{equation}
    q_0(\tau, x) := \P_{\exp}(B_t = 1 \mid B_{<t} = 0, G_a = 1, X = x).
\end{equation}
In words, this function gives the probability that a gifted customer with covariates $x$ and a gift card of attribute $a$ purchases in window $t$ given that they have not purchased in any earlier window. Our extension of transferability posits that the probability an observational unit with $G_\tau = 1$ purchases in the target period is equal to the above hazard. At first glance, it may seem that the more natural object to transfer is the following conditional probability:
\[
q_0^{\naive}(\tau, x) := \underbrace{\P_{\exp}(B_t = 1 \mid G_a = 1, X = x)}_{\text{purchases \emph{exactly} in window } t}.
\]
However, this misses an important subtlety: if a unit has $G_\tau = 1$ in the observational setting, we know by definition that they have held their card for $t$ periods without redeeming it. This behavior needs to be embedded in the object being transferred. Since we have the identity
\[
\P_{\exp}(B_t = 1 \mid G_a = 1, X = x) = q_0(\tau, x)\bar{F}_{t-1,a}(x) \leq q_0(\tau, x)
\]
where $\bar{F}_{t-1,a}(x) := \prod_{s=1}^{t-1}\big(1 - q_0((s, a), x)\big) = \P_{\exp}(B_{<t} = 0 \mid G_a = 1, X = x)$ is the survival function, we see that leveraging $q_0^{\naive}$ can lead to survivorship bias and systematic underestimation of true conditional purchasing probabilities. For the above reasons, we consider the following definition of transferability when we want to account for time between gift card purchase and redemption and card attributes.

\begin{ass}[Type-Dependent Transferability]
\label{ass:transfer-type}
For each card type $\tau = (t, a) \in \calJ$ and almost every $x \in \calX$:
\[
\frac{\P_{\obs}(B = 1 \mid G = 0, X = x)}{\P_{\obs}(B = 1 \mid G_\tau = 1, X = x)} = \frac{\P_{\exp}(B_1 = 1 \mid G = 0, X = x)}{\P_{\exp}(B_t = 1 \mid B_{<t} = 0, G_a = 1, X= x)}.
\]
We further assume $q_0(\tau, x) > 0$ for almost every $x$ in the support of $X \mid G_\tau = 1$ under $\P_{\obs}$, and that $\P_{\exp}(B_{<t}=0, G_a = 1 \mid X = x) > 0$ and $\P_{\exp}(G = 0 \mid X = x) > 0$ for each $a\in {\cal A}$, $t\in [t_{\max}]$ and almost every $x$ under $\P_{\exp}$.
\end{ass}

Intuitively, the above assumption is matching the observational subpopulation $\{G_\tau = 1\}$, or those who hold an unredeemed card of attribute $a$ purchased $t$ periods ago, with the experimental subpopulation $\{G_a = 1,\, B_{<t} = 0\}$, or those who have likewise survived to window $t$ and are considering purchasing. 
For the untreated population, we consider the experimental \emph{one-period} purchasing rate $p_0$. If a customer is selected for the control group in the experiment, they will likely not be notified that they are participating in the experiment. Rather, their purchasing behavior will just be monitored over the course of the experiment. It thus makes sense that their purchasing rate in the first period of the experiment should match their behavior during the target period, so long as these two windows of time have the same length and untreated customer purchasing behavior is relatively stationary over time.

Under the above assumption, we can prove a heterogeneous identification analogous to the one presented in Theorem~\ref{thm:identification}. This identification now depends on the collection of hazards $\{q_0(\tau, \cdot)\}_{\tau \in \calJ}$ instead of the single counterfactual purchasing probability $q_0$.
 
\begin{theorem}[Identification under Augmented Type Heterogeneity]
\label{thm:ident-type}
Suppose Assumptions~\ref{ass:exogeneity}, \ref{ass:positivity}, \ref{ass:obs-exp-time}, \ref{ass:partial-obs-type}, and \ref{ass:transfer-type} hold. Then, the incremental revenue of the gift card program is identified as
\begin{equation}
\Delta_0 := \E_{\obs}\left[Y \sum_{\tau \in \calJ} S_\tau\right] - \E_{\obs}\left[g_0(X) \, p_0(X) \, \sum_{\tau \in \calJ} \dfrac{S_\tau}{q_0(\tau, X)}\right].
\end{equation}
Consequently, the IC $\kappa_0$ and IR $\rho_0$ are identified as well.
\end{theorem}

Note that under the assumption of memoryless purchasing times and no attribute dependence, i.e., when:\footnote{Which would hold if the purchasing time $R$ of an individual follows an exponential distribution and does not depend on the card attribute $a$.}
\[
q_0(\tau, x)  \equiv q_0(x) := \P_{\exp}(B_1 = 1 \mid G = 1, X = x) \qquad \text{for all } \tau \in \calJ, x \in \calX.
\]
the identification formula reduces to that of Theorem~\ref{thm:identification}, as we have $\sum_{\tau \in \calJ} \frac{S_\tau}{q_0(\tau, X)} = \sum_{\tau \in \calJ} \frac{S_\tau}{q_0(X)} = \frac{S}{q_0(X)}$. Under this assumption, it is justified to transfer \emph{one-period} (e.g., one-month) purchasing probabilities for both the treatment and control arms and apply the analysis in the previous section.

While our framework can accommodate a highly granular type space, in practice we find that the gift card face value does not substantially alter purchasing rates within the range observed in our data. Consequently, our main specification relies primarily on the time dimension $t$. 

Lastly, we note that our notions of intensive and extensive margins can also be defined under type heterogeneity. Importantly, while the definition of $\Delta_{0, \intm}$ remains unchanged, the definition of the extensive share of incremental revenue becomes
\[
\Delta_{0, \extm} = \E_{\obs}\left[\sum_{\tau \in \calJ} S_\tau g_0(X)\left\{1 - \eta_0(\tau, X)\right\}\right],
\]
where we have defined the heterogeneous ratio $\eta_0(\tau, x) := \frac{p_0(x)}{q_0(\tau, x)}$ for each $\tau \in \calJ$. The intensive and extensive share of both the IR and IC follow using this new representation for $\Delta_{0, \extm}$ in place of the homogeneous version. Note that we would need to assume an equivalent to Assumption~\ref{ass:cross_positivity} for each $\tau \in \calJ$ in order to ensure $\Delta_{0, \intm}$ and $\Delta_{0, \extm}$ are well-defined.

\section{Estimation}
\label{sec:estimation}

In the previous section, we developed an identification strategy for the incrementality coefficient (IC) $\kappa_0$ and the incrementality ratio (IR) $\rho_0$. 
We now manifest this identification into an estimation strategy.

Our identification for $\rho_0$ and $\kappa_0$ hinged upon identifying $\Delta_0$, the incremental revenue of the business's gift card program. Our approach to estimation will be exactly analogous --- we will derive confidence intervals for $\Delta_0$ first and then leverage the delta method to produce confidence intervals for our objects of interest, namely $\kappa_0$ and $\rho_0$. Any estimation strategy for $\Delta_0$ naturally requires the estimation of the organic spend function $g_0$ and the experimental conditional booking propensities $p_0$ and $q_0$. These \textit{nuisance functions} correspond to possibly highly complex predictive models. Our goal is to develop an estimation strategy that will allow us to use flexible, modern machine learning methods (such as random forests, gradient boosted forests, deep neural networks, and penalized regressions) to estimate these complex relationships, while still providing asymptotically valid uncertainty quantification and confidence interval construction. 

The most direct approach to estimating $\Delta_0$ would be to train machine learning models for these nuisance functions and simply plug their predictions into our identification formula (Theorem~\ref{thm:ident-type}). In particular, defining the plug-in observational score $m(W; g, p, q)$ as
\begin{equation}
\label{eq:base-score}
m(W; g, p, q) := Y \sum_{\tau \in \calJ} S_\tau - g(X)\, p(X)\sum_{\tau \in \calJ} \frac{S_\tau}{q(\tau, X)} 
\end{equation}
and noting that our identification formula implies that $\Delta_0 = M(g_0, p_0, q_0)$ (where $M(g, p, q) := \E_{\obs}\left[m(W; g, p, q)\right]$), one may naively believe that $\wh{\Delta} := \frac{1}{n}\sum_{i = 1}^n m\left(W; \wh{g}, \wh{p}, \wh{q}\right)$ is a good estimate of $\Delta_0$. Here, $n$ denotes the size of the observational dataset, and $\wh{g}, \wh{p},$ and $\wh{q}$ represent trained estimates for $g_0, p_0$, and $q_0$.

Unfortunately, this naive plug-in approach fails in practice. Machine learning models inherently involve regularization and variable selection, which introduces small, systematic estimation errors (bias). In a standard plug-in estimator, these small errors in the nuisance functions directly contaminate the final estimate of $\Delta_0$, introducing substantial bias and rendering standard confidence intervals invalid. To solve this, we employ the \textit{debiased machine learning} (DML) framework \citep{chernozhukov2018double}. The core idea of DML is to mathematically ``correct'' the score so that it becomes insensitive to small errors in the machine learning predictions. When a score possesses this insensitivity property, it is called \textit{Neyman orthogonal}. By using a Neyman orthogonal score, the biases from the machine learning models become second-order, allowing us to estimate $\Delta_0$ at root-$n$ rates and compute valid confidence intervals.

\subsection{Debiasing under Flexible Logistic Booking Propensity}
\label{sec:logistic-propensity}

To construct a Neyman orthogonal score, the DML framework provides a general recipe: for each nuisance function, we must append a correction term that exactly cancels the sensitivity of the base moment to small perturbations of that nuisance. This is typically achieved by adding a weighted average of the regression residuals for the corresponding regression problem that defines each nuisance function $f\in \{g,p,q\}$, i.e., for an appropriately defined weighting function $a:\calX\to R$, we add to the original moment a term of the form
\begin{align*}
\E_{e}\left[a(X) Q (C - f(X))\right]
\end{align*}
where $e\in \{\obs, \exp\}$ represents the dataset the nuisance was trained on, $C$ is the target label that $f$ is predicting and $Q$ is a selection event defining the sub-population on which the regression is trained.

In our setting, the base moment $M$ is defined as an expectation over the observational distribution $\P_{\obs}$. However, the booking propensities $p_0$ and $q_0$ are trained on experimental sub-populations. In principle, one could construct fully nonparametric corrections by estimating the density ratio $\omega_0(X) := d\P_{\obs}/d\P_{\exp}(X)$ between the observational and experimental populations (we detail this approach in Appendix~\ref{sec:neyman_construction}). However, in practice, this density ratio can take extremely large values if regions of the covariate space are under-represented in the experiment, leading to unstable estimates and wide confidence intervals. To avoid this problematic density ratio, we adopt a semiparametric approach. We assume the conditional booking propensities follow a logistic regression model with a flexible set of engineered features $\phi(X)$:
\begin{ass}
\label{ass:logistic}
There are known maps $\phi_p : \calX \rightarrow \R^{d_p}$ and $\phi_\tau : \calX \rightarrow \R^{d_\tau}$ and unknown vectors $\beta_{0, p} \in \R^{d_p}$ and $\beta_{0, \tau} \in \R^{d_\tau}$ such that
\[
p_0(x) = \sigma(\phi_p(x)^\top \beta_{0, p}) \quad \text{and} \quad q_0(\tau, x) = \sigma(\phi_\tau(x)^\top \beta_{0, \tau}) \quad \text{for all } \tau \in \calJ,
\]
where $\sigma(u) := (1 + e^{-u})^{-1}$ is the logistic function. Let $\beta_0 = (\beta_{0, p}, \{\beta_{0, \tau}\}_{\tau \in \calJ})$.
\end{ass}

Under this assumption, we can apply a correction term that directly targets the parameters $\beta$ and avoids the dependence on the density ratio. Instead of rendering the moment insensitive to arbitrary, nonparametric perturbations of the propensity functions, it suffices to cancel the derivative of the base moment with respect to the finite-dimensional parameters $\beta$. The true parameters satisfy the first-order condition of the logistic loss, meaning the expected logistic scores on the experimental data are exactly zero. 
\begin{align}
\E_{\exp}\left[s(Z;\beta)\right] = 0, \qquad \text{with } s(Z;\beta) := Q\, (C - \sigma(\beta^\top \phi(X)))\, \phi(X).
\end{align}
where $Q$ is the selection event that defines the sub-population of the experimental sample on which each of the propensities is trained and $C$ is the target binary label. For instance, for identifying $p_0$, the selection event is $Q = 1 - G$ and the label is $C = B$. Likewise, for identifying $q_0(\tau, \cdot)$ where $\tau = (t, a)$, the selection event is $Q = (1 - B_{< t}) G_a$ and the outcome is $C =  B_t$. Therefore, we can correct our moment by adding a linear combination of these experimental logistic scores, i.e., $A\E_{\exp}\left[s(Z;\beta)\right]$. The optimal linear combination matrices $A$ are derived by setting the derivative of the corrected moment with respect to $\beta$ to zero (see Appendix~\ref{sec:neyman_construction} for explicit formulae). Crucially, for this semiparametric correction to be valid, we only require a benign \textit{column space inclusion regularity condition} (again, see Appendix~\ref{sec:neyman_construction}), which is substantially weaker than (and implied by) requiring a bounded density ratio everywhere.

Applying this recipe yields a corrected score, which consists of two parts. The \textit{observational correction} debiases the organic spend model $g_0$ by appending a term proportional to its prediction residual:
\begin{equation}
\label{eq:obs_score}
\psi_{\obs}(W;  g, \beta, \alpha_g) := m(W; g, p, q) + \alpha_g(X)\,D\left\{Y - g(X)\right\},
\end{equation}
where the correction weight $\alpha_{0,g}(x) = -\P_{\obs}(G=1\mid X=x)\,/\,\P_{\obs}(G=0\mid X=x)$ is the negative odds of being a gift-card holder. While this correction $\alpha_{0, g}(x)$ is not identifiable from observational samples alone, it can be expressed as $\alpha_{0, g}(x) = -\frac{p_0(x)}{\ell_0(x)}\sum_{\tau \in \calJ}\frac{r_0(\tau, x)}{q_0(\tau, x)}$, where $r_0(\tau, x) := \P_{\obs}(S_\tau = 1 \mid X = x)$ is the probability of observing the gift-card-holder indicator for each type $\tau$ and $\ell_0(x) := \P_{\obs}(D = 1 \mid X = x)$ is the probability of making a purchase without a gift card.

The \textit{experimental correction} de-biases the booking propensities $p_0$ and $q_0$ by adding a linear combination of the corresponding logistic scores from the experimental data:
\begin{equation}
\label{eq:exp_score_para}
\psi_{\exp}(Z; \beta, A) := A s(Z; \beta) \quad \text{where}\;\; s(Z; \beta) := \big(s_p(Z; \beta_p), \left\{s_\tau(Z; \beta_\tau) : \tau \in \calJ\right\}\big), A \in \R^{1 \times d},
\end{equation}
where $s_p$ and $s_\tau$ are the logistic score vectors from the experimental control and treated arms, respectively, and the correction matrices $A$ are computed in closed form as empirical Jacobians of the base moment and Hessians of logistic scores with respect to the parameters $\beta$ (see Appendix~\ref{sec:neyman_construction}).

\subsection{Cross-Fitting Estimation Algorithm}
\label{sec:estimation-algo}

With the Neyman orthogonal moment established, we construct our final estimator. As is typical in the debiased machine learning framework, to avoid in-sample overfitting bias of the ML models, we employ $K$-fold \emph{cross-fitting}. We partition our observational data into $K$ equal-sized folds (e.g., $K=5$). For each fold $k$, we train our observational machine learning models on all the \textit{other} folds, and then generate predictions for the data in fold $k$. For the experimental data, cross-fitting is not required, as we assume the booking propensities follow a logistic model. Algorithm~\ref{alg:dml} summarizes the estimation procedure.

To streamline the presentation of the estimator, we find it convenient to define the plug-in variance estimates as
\[
\wh{\Var}_{\obs}(\phi_i(W_i)) := \frac{1}{n}\sum_{i = 1}^n \left(\phi_i(W_i) - \P_n\left\{\phi_i(W_i)\right\}\right)^2 \qquad \wh{\Var}_{\exp}(\varphi_j(Z_j)) := \frac{1}{N}\sum_{j =1}^N \left(\varphi_j(Z_j) - \P_N \left\{\varphi_j(Z_j)\right\}\right)^2
\]
for any functions $\phi_1, \dots, \phi_n$ and $\varphi_1, \dots, \varphi_N$, where $\P_n\left\{\phi_i(W_i)\right\} := \frac{1}{n}\sum_{i = 1}^n \phi_i(W_i)$ and $\P_N$ is defined similarly.

\begin{algorithm}[ht]
\caption{Cross-Fit Incrementality Estimator}
\label{alg:dml}
\begin{algorithmic}[1]
\State \textbf{Input:} Observational data $\{W_i\}_{i=1}^n$, Experimental data $\{Z_j\}_{j=1}^N$, number of folds $K$.
\State \textbf{Partition} observational data randomly into $K$ equal-sized folds $\calI_1, \dots, \calI_K$.
\For{$k = 1, \dots, K$}
    \State Train nonparametric models $\wh{g}^{(-k)}$, $\wh{r}^{(-k)}$, and $\wh{\ell}^{(-k)}$ on all data \textit{except} fold $k$.
\EndFor
\State Learn $\wh{\beta}_p$ and $\wh{\beta}_{\tau}$ via logistic regression. 
\State For each $k \in [K]$, define $\wh{\alpha}_g^{(-k)}(x) := -\frac{\wh{p}(x)}{\wh{\ell}^(-k)(x)}\sum_{\tau \in \calJ} \frac{\wh{r}(\tau, x)}{\wh{q}(\tau, x)}$, where $\wh{p}(x) = \sigma(\phi_p(x)^\top \wh{\beta}_p)$ and $\wh{q}(\tau, x) := \sigma(\phi_\tau(x)^\top\wh{\beta}_\tau).$
\State Compute correction matrix $\wh{A}$.
\State Let $\wh{U}_i := \psi^{\obs}\left(W_i; \wh{g}^{(-k)}, \wh{\beta}, \wh{\alpha}_g^{(-k)}\right)$ for $i \in \calI_k$ and $\wh{V}_j := \psi^{\exp}(Z_j; \wh{\beta}, \wh{A})$. Define $\wh{\Delta}$ as
\begin{equation}
\label{eq:estimator}
\wh{\Delta} = \frac{1}{n}\sum_{k = 1}^K\sum_{i \in \calI_k}\wh{U}_i + \frac{1}{N}\sum_{j = 1}^N \wh{V}_j.
\end{equation}
\State Define the standard error as:
\[
\wh{\Sigma} = \sqrt{\frac{\wh{\Var}_{\obs}(\wh{U}_i)}{n} + \frac{\wh{\Var}_{\exp}(\wh{V}_j)}{N}}.
\]
\State \textbf{Return:} Point estimate $\wh{\Delta}$ and $1-\delta$ confidence interval $\wh{\Delta} \pm z_{1-\delta/2}\wh{\Sigma}$.
\end{algorithmic}
\end{algorithm}

This estimator achieves asymptotically nominal coverage under typical regularity conditions in the debiased machine learning literature. Crucially, this convergence is guaranteed as long as the \textit{minimum} of the two sample sizes grows large, making no assumptions about their relative scale. This flexibility is essential for common industrial settings where observational data is abundant but experimental samples are limited. We prove a formal asymptotic normality theorem in Appendix~\ref{sec:neyman_construction}.

\paragraph{Estimation of the incrementality coefficient and ratio.}
The IC $\kappa_0 = \Delta_0 / \E[A G]$ and the IR $\rho_0 = \Delta_0 / \E[YG]$ can be estimated by directly leveraging our point estimates for $\Delta_0$ alongside the delta method. In particular, letting $\wh{\E}_{\obs}[Y S] := \frac{1}{n}\sum_{i = 1}^n Y_i S_i$ and $\wh{\E}_{\obs}[A S] := \frac{1}{n}\sum_{i = 1}^n A_i S_i$, our point estimates are naturally given by
\[
\wh{\kappa} = \frac{\wh{\Delta}}{\wh{\E}[A S]} \quad \text{and} \quad \wh{\rho} = \frac{\wh{\Delta}}{\wh{\E}[Y S]}.
\]
These estimates are asymptotically normal by the delta method, and their variance can readily be computed as well for the sake of constructing asymptotically valid confidence intervals at a desired coverage level. In more detail, define $H_\rho(W) := Y\cdot S$, $H_\kappa(W) := A \cdot S$, and let $r \in \{\kappa, \rho\}$ be arbitrary. Define the standard error
\[
\wh{\Sigma}_r := \sqrt{\frac{\wh{\Var}_{\obs}(\wh{U}_i - \wh{r}H_r(W_i))}{n\wh{\mu}_r^2} + \frac{\wh{\Var}_{\exp}(\wh{V}_j)}{\wh{\mu}_r^2 N}}.
\]
Then, an asymptotically valid $1 - \delta$ confidence interval for $r_0 \in \{\kappa_0, \rho_0\}$ is given by
\[
\calC_{1 - \delta} := \left[\wh{r} - z_{1 - \delta/2}\wh{\Sigma}_r, \wh{r} + z_{1 - \delta/2}\wh{\Sigma}_r\right].
\]

We provide full details on the exact construction of these estimates and the corresponding intervals in Appendix~\ref{sec:neyman_construction}.

\section{Empirical Application to Airbnb's Gift Card Program}
\label{sec:case_study}

We now apply our causal framework to the task of estimating the incrementality of Airbnb's gift card program. We start by providing a brief roadmap.  In Section~\ref{sec:case_study:empirical_setup}, we briefly remind the reader of both the observational and experimental datasets we are working with, rigorously defining treatment and outcomes. In Section~\ref{sec:case_study:main_results}, we present and interpret our main findings regarding the incrementality of Airbnb's gift card program, providing channel-by-channel results. In Section~\ref{sec:case_study:margin_effects} we decompose both our incrementality coefficient (IC) and incrementality ratio (IR) estimates into intensive and extensive margins to study how gift card receipt separately impacts customer spend and booking behavior. In Section~\ref{sec:case_study:baselines}, we compare our estimates to naive baselines that one might construct for the incrementality ratio.
Lastly, in Section~\ref{sec:case_study:self_gift}, we recompute estimates for the IC and IR on a restricted population of suspected self-gifters.

\subsection{Data and Empirical Setup}
\label{sec:case_study:empirical_setup}

We provided a detailed overview of both our experimental and observational datasets in the empirical context section (Section~\ref{sec:context}). Now, with both identification and estimation details at hand, we can complete our discussion of these datasets by rigorously describing how our variables are defined in both the experimental and observational populations. We also discuss the details of estimation, including nuisance model selection and fitting.
Throughout this section, we operate under the \textit{type-heterogeneous} identification framework presented in Section~\ref{sec:model:heterogeneity} above, considering just \textit{time heterogeneity}, meaning we take into account the number of months between when a customer received and redeemed a gift card.

Throughout our application, we take the time index as the set $\calT = \{1, 2, \dots, 11, 12\}$, meaning that we monitor customer booking hazard rates on the basis of months, which thus matches the length of the experimental monitoring window. 

\paragraph{Observational Dataset}
Throughout this description, we suppose we are focusing on a single, North American channel (e.g., B2B or First-Party) and month in 2024. In our analysis, we exclude all units who booked with a gift card from a channel other than the current one being considered. Here, $B \in \{0, 1\}$ indicates  whether or not a customer made a purchase in the target month. While most customers only make at most a single booking in any given month, some customers may make multiple bookings --- we view multiple bookings as a single, large booking event. We define the treatment $G_t \in \{0, 1\}$ as the indicator of whether or not a unit possessed a gift card that was purchased between $t - 1$ and $t$ months ago at the time of their first purchase in the focal period.\footnote{If a unit did not make a purchase but possessed a gift card,  we let $G_t$ be the indicator of whether or not they possessed a gift card of age $t$ by the end of the current month. Note that we never actually observe $G_t$ for these units.}
While some units may possess multiple gift cards, we always use the age of the \textit{youngest gift card}. If a unit does not possess a gift card at the time of their first booking, we set $G = 0$. If a unit's youngest gift card is greater than twelve months old, we set $G_{12} = 1$. The outcome variable $Y$ is defined as the \textit{gross booking value (GBV)} or \textit{total amount spent} on all properties booked by a customer in the focal month.
Lastly, $A$ is defined to be the number of gift card dollars \textit{redeemed} by a unit. It is thus zero if either (a) a customer does not use a gift card in a purchase or (b) if they do not make a purchase at all.

For these units, we measure a rich covariate vector $X$ that consists of 153 features. These features are roughly organized into three categories: (i) \emph{gifting-related features}, capturing historical gift card behavior over the preceding three years;
(ii) \emph{historical spend and booking features}, measuring the volume, frequency, and dollar amount of uncanceled bookings over the previous three years; and (iii) \emph{product detail page (PDP) features} capturing historical browsing intensity and price signals from listing views. All of these features are defined prior to the first booking made by a unit in the period (or prior to the start of the period for non-booking units). Full details on these features can be found in Appendix~\ref{app:case_study:features}. 

\paragraph{Experimental Dataset}

A unit's treatment status $G \in \{0, 1\}$ is defined as whether or not they were assigned a gift card at the start of the experiment. The booking variable $B_t \in \{0, 1\}$ is then naturally defined as the indicator of whether or not a unit made their first booking between $t - 1$ and $t$ months after receiving a gift card, and the variables $B_{< t}$ are naturally defined in terms of the $B_t$'s as outlined in Section~\ref{sec:model:heterogeneity}.
For covariates $X$, we measure a set of 23 features that balance representativeness across all three groups noted above --- again see Appendix~\ref{app:case_study:features} for a full list. These features are defined for each unit prior to the start of the experiment, i.e., all measurements are taken over a period ending immediately before treatment assignment.

We note some differences between how treatment and features are defined in the experimental and observational datasets. First, in the experimental dataset, treatment is determined by random assignment at a single point in time and then booking outcomes are monitored monthly over the following year. In the observational data, treatment corresponds to \emph{possession} of an unredeemed gift card of age $t$ months at the time of booking. Our transferability assumptions do not require that treatment is specified in precisely the same manner between the populations. Further, our use of the time-heterogeneous identification/estimation strategy is meant to handle these discrepancies by, very heuristically, pairing observational units who possess a gift card of age $t$ with experimental units who make their first booking at the $t$ month time point. For observational features $X$, we note that it is generally not possible to know precisely when a  unit received their gift card, particularly amongst those who did not book. Thus, we opt to define features in a way that is maximally consistent across units in the observational dataset.  

\subsection{Main Results: Incrementality Coefficients and Ratios}
\label{sec:case_study:main_results}

We now report our main findings, which describe the Incrementality Coefficient (IC) and Incrementality Ratio (IR) segmented by the primary distribution channels in North America: First-Party (Airbnb's own platform),  Retail Online, Retail In-Store, and Business-to-Business (B2B). In this section, we show incrementality estimates aggregated over all months of 2024 in Table~\ref{tab:allunits}. In Figure~\ref{fig:au_main} of Appendix~\ref{app:case_study:monthly}, we further show our estimates plotted month-by-month in order to measure temporal changes and stability over time. Again, yearly aggregated numbers are obtained by first computing the monthly estimates in the manner described above and then appropriately averaging the influence functions per Appendix~\ref{sec:neyman_construction:time_avg}.

\begin{table}[H]
\centering
\begin{tabular}{lcccc}
\toprule
& \multicolumn{2}{c}{\textbf{IR}} & \multicolumn{2}{c}{\textbf{IC}} \\
\cmidrule(lr){2-3}\cmidrule(lr){4-5}
\textbf{Channel} & \textbf{Estimate} & \textbf{95\% CI} & \textbf{Estimate} & \textbf{95\% CI} \\
\midrule
First-Party      & $-0.199$ & $[-0.544,\ 0.146]$ & $-0.676$ & $[-1.856,\ 0.505]$ \\
Retail Online    & $0.204$  & $[-0.006,\ 0.415]$ & $0.428$  & $[-0.019,\ 0.875]$ \\
Retail In-store  & $-0.049$ & $[-0.265,\ 0.168]$ & $-0.131$ & $[-0.746,\ 0.485]$ \\
B2B              & $0.221$  & $[0.035,\ 0.406]$  & $0.545$  & $[0.083,\ 1.008]$ \\
\bottomrule
\end{tabular}
\caption{Yearly estimates and marginally valid 95\% confidence intervals for the IR and IC broken down by channel.}
\label{tab:allunits}
\end{table}

It is clear from Table~\ref{tab:allunits} that incrementality varies highly by channel. The yearly point estimates for both the IC and IR are respectively negative and neutral for First Party and Retail In-Store channels. On the other hand, the incrementality point estimates for Retail Online and B2B channels are positive, with estimates for the IR indicating that roughly 20\% of GBV amongst gift card recipients is attributable to the existence of the gift card program. Further, the confidence intervals for both the IC and IR estimates in B2B exclude zero, indicating that our findings of incrementality in this channel are statistically significant. We note that the confidence intervals for all channels are relatively wide, which is likely due to the limited sample size of the experimental sample. We hypothesize that, given a larger experiment, we would find that Retail Online would also be significantly incremental as well. 

At first, our findings here seem somewhat unintuitive. The channels exhibiting the highest incrementality (B2B and Retail Online) are arguably the ones most susceptible to self-gifting phenomena. This is because they often offer gift cards that are effectively discounted, say via rewards programs or loyalty programs. As noted earlier in the paper, one might naively believe that self-gifters are fundamentally non-incremental, simply cannibalizing and replacing potential out of pocket spend with discounted gift card value. However, this intuition is not captured by our findings. We return to investigating the incrementality of self-gifting customers in Section~\ref{sec:case_study:self_gift} below. We provide a more qualitative discussion on the incrementality of self-gifters in Section~\ref{sec:discuss}. 

\subsection{Intensive and Extensive Margin Decomposition}
\label{sec:case_study:margin_effects}

In the previous subsection, we studied the general incrementality of Airbnb's gift card program across channels. We now aim to determine the mechanisms through which gift cards impact customer spend behavior.
To understand the drivers of incrementality, we decompose our incrementality estimates into intensive margins (increased spend among those who would have booked anyway) and extensive margins (new bookings induced by the gift card).  We display the channel-by-channel decomposition of the IC and IR into intensive and extensive components in Figure~\ref{fig:au_decomp}, and Table~\ref{tab:au_ir_shares} displays point estimates for intensive and extensive shares of IR alongside marginally valid 95\% confidence intervals. Further, for the IR, we display our estimate for $1 - \rho_0$, which represents the fraction of GBV amongst those who received a gift card that was cannibalized. 

Our first interesting finding occurs by looking at our estimates of intensive margins. While we see that the intensive margin amongst First-Party gift card recipients is essentially zero, for all other channels the estimates are positive (being around 0.17-0.19 for IC and 0.06-0.09 for IR). Further, the 95\% confidence intervals for these estimates all exclude zero, indicating that the finding of positive intensive margins is statistically significant. 
We note that, regardless of channel, the confidence intervals for the intensive margins are significantly smaller than those for both the extensive margins and the overall incrementality estimates. This makes sense per our discussion in Section~\ref{sec:model:margins} --- the intensive shares/margins \textit{are identifiable from the observational population alone}, and hence the small standard errors are a consequence of the overall large dataset sizes.

Our second interesting finding is that the point estimates for extensive margins appear to be mildly negative in the First-Party and Retail In-Store channels, but positive in Retail Online and B2B channels. 
For the former two channels, these results seem to indicate that gift card receipt actually has the effect of \textit{decreasing} the probability of a unit booking in said channels.
We note that it is important to interpret these point estimates with a grain of salt, as in all channels the confidence intervals for the extensive margins are wide and contain zero. 
In reality, we believe that extensive margins should be non-negative regardless of channel: receipt of a gift card should at worst keep a unit's probability of booking the same. We speculate that extensive margins are actually high amongst units who hold recently purchased gift cards, with the intuition being that they are either self-gifting\footnote{And thus that they would have instead booked with an outside option had the gift card program not existed.} or that they are experiencing some sort of advertising effect. We investigate this later in Section~\ref{sec:case_study:self_gift}, where we re-estimate both intensive and extensive margins while restricting our attention to units who redeem a gift card within a week of purchase date.

\begin{figure}[H]
\centering
\begin{subfigure}{0.46\textwidth}
  \centering\includegraphics[width=\textwidth]{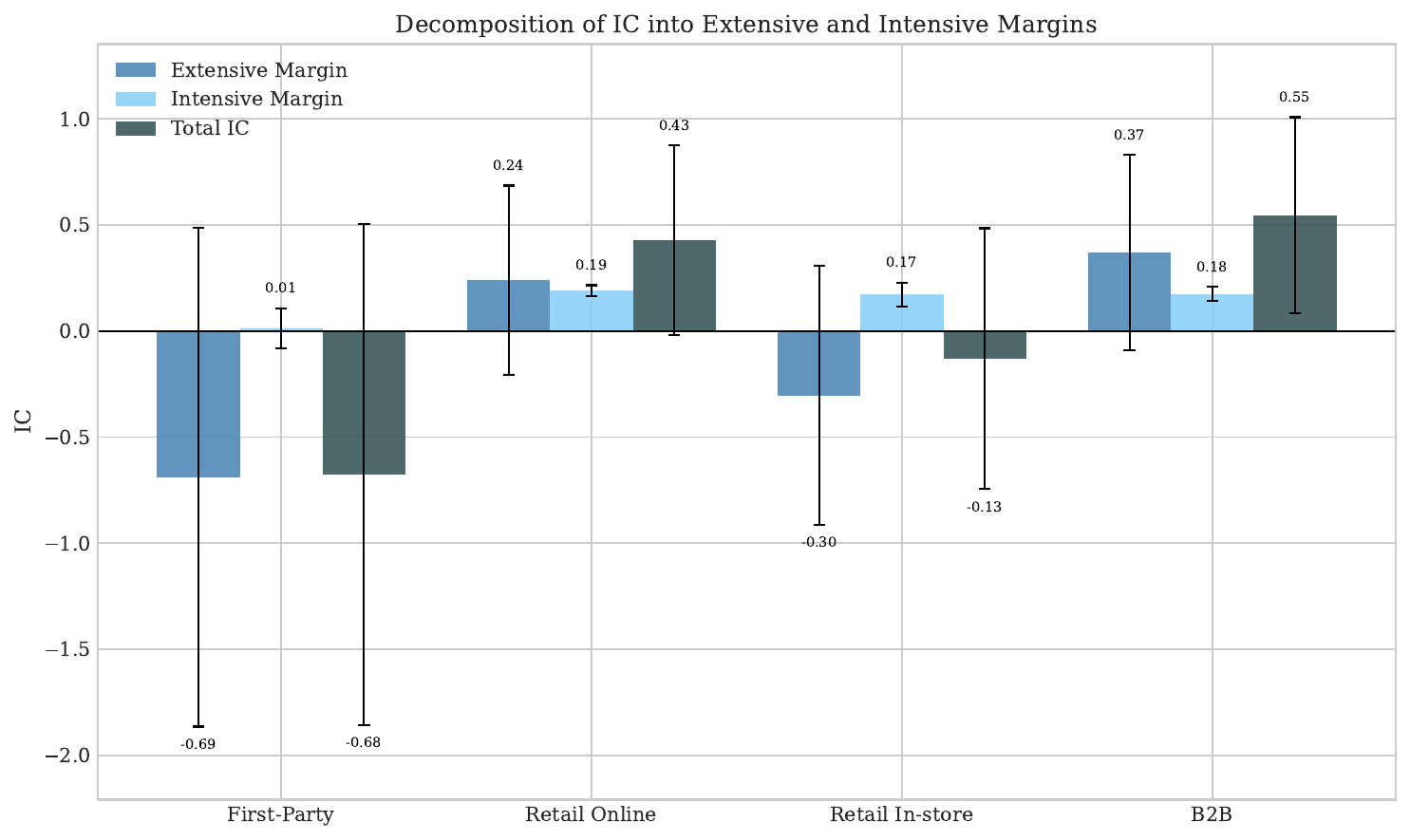}
  \caption{IC into intensive and extensive margins.}
\end{subfigure}\hfill
\begin{subfigure}{0.53\textwidth}
  \centering\includegraphics[width=\textwidth]{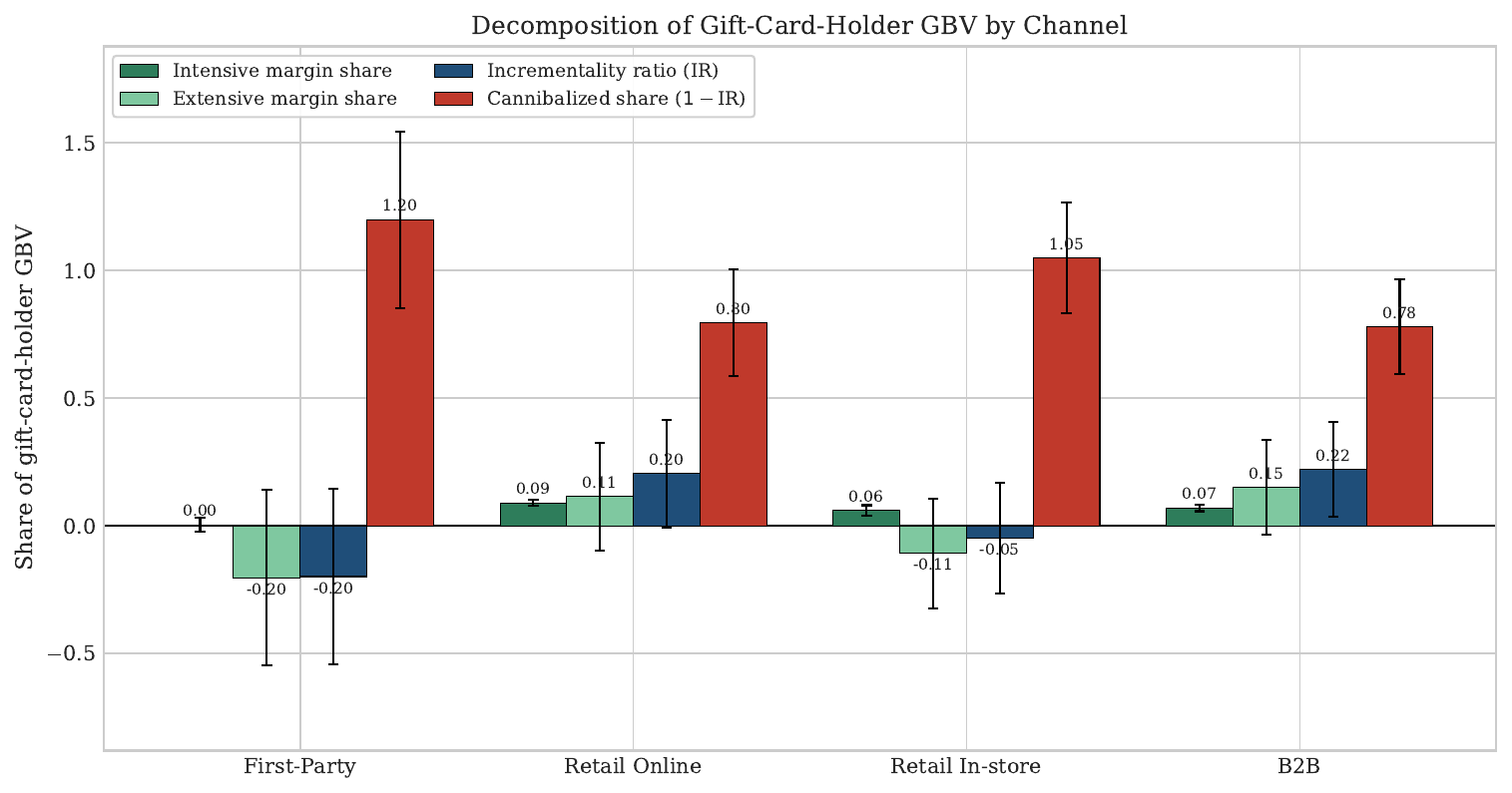}
  \caption{IR into intensive, extensive and cannibalised shares.}
\end{subfigure}
\caption{All units: decompositions, yearly time averages with 95\% confidence intervals. The
sign-safe grouped layout is used because the extensive component is negative in two channels; a
stacked layout would clip it at the axis. For the same reason no stacked cannibalisation figure is
shown for this population.}
\label{fig:au_decomp}
\end{figure}

\begin{table}[H]
\centering
\begin{tabular}{lcccc}
\toprule
& \multicolumn{2}{c}{\textbf{Intensive-Margin Share}} & \multicolumn{2}{c}{\textbf{Extensive-Margin Share}} \\
\cmidrule(lr){2-3}\cmidrule(lr){4-5}
\textbf{Channel} & \textbf{Share} & \textbf{95\% CI} & \textbf{Share} & \textbf{95\% CI} \\
\midrule
First-Party     & $0.004$ & $[-0.024,\ 0.032]$ & $-0.203$ & $[-0.546,\ 0.140]$ \\
Retail Online   & $0.089$ & $[0.077,\ 0.102]$  & $0.115$  & $[-0.096,\ 0.325]$ \\
Retail In-store & $0.060$ & $[0.040,\ 0.081]$  & $-0.109$ & $[-0.324,\ 0.106]$ \\
B2B             & $0.071$ & $[0.057,\ 0.084]$  & $0.150$  & $[-0.035,\ 0.335]$ \\
\bottomrule
\end{tabular}
\caption{All units: intensive- and extensive-margin shares of gift-card-holder GBV, with 95\%
confidence intervals (the margin-share intervals that Figure~\ref{fig:au_decomp} suppresses). }
\label{tab:au_ir_shares}
\end{table}

\subsection{Comparison with Baselines}
\label{sec:case_study:baselines}

We now compare our estimates for the incrementality ratio with several natural baselines that can be constructed from the observational and experimental datasets. These approaches represent naive approaches one may take in gauging incrementality, and we show how they can lead to misleading or incorrect conclusions regarding the incrementality of a gift card program within a target channel. We now describe the approaches, and emphasize that none of the baselines provides a valid estimate for the observational IR.

\begin{enumerate}
\item \textbf{Observational naive intensive margin:} one can compute the marginal difference between the expected spend of a unit who books with a gift card and one who books without: $$\frac{\E_{\obs}[Y \mid S = 1] - \E_{\obs}[Y\mid D{=}1]}{\E_{\obs}[Y\mid S{=}1]}.$$
\item \textbf{Experimental IR:} One can directly compute the analogue of the IR in the experiment by differencing conditional spend between those who received a gift card and those who did not:
$$\frac{\E_{\exp}[Y \mid G = 1] - \E_{\exp}[Y\mid G{=}0]}{\E_{\exp}[Y\mid G{=}1]} = \frac{\E_{\exp}[Y \mid G = 1] - \E_{\exp}[Y(0) \mid G = 1]}{\E_{\exp}[Y \mid G = 1]} = \frac{\E_{\exp}[(Y - Y(0))G]}{\E_{\exp}[Y G]}.$$ 

\item \textbf{Experimental extensive margin:} We can measure an effect that just attempts to predict extensive margin, replacing $Y$ with the one month booking indicator $B_1$ in the definition of the IR: 
$$\frac{\P_{\exp}(B_1  = 1 \mid G = 1)  - \P_{\exp}(B_1{=}1\mid G{=}0)}{\P_{\exp}(B_1{=}1\mid G{=}1)} = \frac{\E_{\exp}[(B - B(0))G]}{\E_{\exp}[BG]},$$
where $B(0)$ represents a unit's counterfactual booking decision had they not received a gift card.
\item \textbf{Experimental intensive margin:} We define the natural experimental intensive margin as well, which is just the lift in spend conditional on booking. This is the same as the observational naive intensive margin, just replacing $\E_{\obs}$ by $\E_{\exp}$.
\item \textbf{Unadjusted data fusion:} We can additionally operate under the assumption of \textit{no heterogeneity} and compute the best constant estimates of $g_0, p_0,$ and $q_0$. We can do this in both the memoryless and time-heterogeneous settings. 
\end{enumerate}

\paragraph{Comparison with experimental baselines.} We first compare our estimates to the point estimate for the IR in the experimental dataset. We note that this estimate (which is 0.076) is contained in the IR confidence intervals across four channels. This indicates that our IR estimates are not drastically larger or smaller than what one would expect from an experiment. The confidence interval for the experimental IR also contains zero, which indicates that the experimental data alone cannot confirm the incrementality of the gift card program. In particular, if one had relied on experimental estimates alone, our finding that B2B is significantly incremental would not have been supported.

Continuing with our comparison to baselines defined on the experiment, we can look at naive analogues of intensive and extensive margins. The former simply looks at the relative difference between the average realized spend for those who booked with and without a gift card (normalized by the former quantity). The latter correspondingly looks at the relative lift in one month booking probability (ignoring weighting by spend) between those who had a gift card and those who did not.  We note that, unlike the IR, the stated definitions \textit{do not} represent the corresponding direct and indirect effects of receiving a gift card. The point estimate for the naive intensive margin in the experiment is mildly negative (-0.07), which indicates that units who booked with gift card, on average, spent \textit{less} than those who booked without a gift card. On the other hand, our estimates indicate significantly positive intensive margins regardless of channel. We note that the point estimate for the experimental extensive share (0.138) is approximately the same as the extensive shares in B2B and Retail Online (0.150 and 0.115, respectively). However, we note that these quantities are not directly comparable, as the latter weights gaps in conditional booking probabilities by both the probability of booking and the expected natural spend.

\paragraph{Comparison with observational baselines.}

We first look at unadjusted (i.e., covariate independent) estimates of the IR produced assuming memorylessness transferability (i.e., Assumption~\ref{ass:transfer}) and assuming time-heterogeneous transferability (i.e., Assumption~\ref{ass:transfer-type}). We observe that the memoryless estimates are much larger than our proposed estimates, ranging from 0.38-0.43, depending on the channel. We view these estimates with a bit of skepticism, especially the large estimates in channels like First-Party (0.379) and Retail In-Store (0.354) where we believe true-gifting behavior to be common. For these channels, we actually expect the IR estimate from the experiment (0.076) to be accurate, as the treatment assignment mechanism was reminiscent of true-gifting phenomena. The time-heterogeneous unadjusted estimates (the third row of Panel A of Table~\ref{tab:allunits_baselines_ext}) provide much more conservative estimates of incrementality, and they are generally closer to our own estimates for the IR in terms of magnitude (i.e., within a margin of error of 0.1-0.15 in all but First-Party). In particular, the unadjusted time-heterogeneous estimate for the IR in First-Party is within 0.015 of the corresponding experimental estimate. However, for this estimate to yield valid inference, one would need to make strong, unconditional exogeneity assumptions on the treatment, which we find hard to justify outside of experimental settings. We also note that the unadjusted estimates for observational intensive margins, hovering between 0.25-0.32, are drastically larger than the corresponding \textit{adjusted} estimates, which cap off at 0.089 in Retail Online sales.

\begin{table}[H]
\centering
\resizebox{\textwidth}{!}{%
\begin{tabular}{lcccc}
\toprule
\multicolumn{5}{l}{\textbf{Panel A. Observational and data-fusion baselines} (IR [95\% CI])} \\
\midrule
\textbf{Estimator} & \textbf{First-Party} & \textbf{Retail Online} & \textbf{Retail In-store} & \textbf{B2B} \\
\midrule
Obs.\ intensive margin only   & 0.279 [0.269, 0.290]  & 0.285 [0.279, 0.290]  & 0.251 [0.245, 0.256]  & 0.324 [0.319, 0.329] \\
Unadj.\ fusion (memoryless)   & 0.379 [0.343, 0.414]  & 0.383 [0.349, 0.417]  & 0.354 [0.318, 0.389]  & 0.417 [0.385, 0.449] \\
Unadj.\ fusion (time-het.)    & 0.090 [0.051, 0.129]  & 0.232 [0.201, 0.264]  & 0.091 [0.054, 0.129]  & 0.315 [0.287, 0.343] \\
\textbf{Proposed --- IR}      & \textbf{-0.199} [-0.544, 0.146] & \textbf{0.204} [-0.006, 0.415] & \textbf{-0.049} [-0.265, 0.168] & \textbf{0.221} [0.035, 0.406] \\
\quad Proposed --- intensive share & 0.004 [-0.024, 0.032] & 0.089 [0.077, 0.102] & 0.060 [0.040, 0.081] & 0.071 [0.057, 0.084] \\
\quad Proposed --- extensive share & -0.203 [-0.546, 0.140] & 0.115 [-0.096, 0.325] & -0.109 [-0.324, 0.106] & 0.150 [-0.035, 0.335] \\
\midrule
\multicolumn{5}{l}{\textbf{Panel B. Experiment-only baselines} ( IR [95\% CI])} \\
\midrule
\textbf{Estimator} & \multicolumn{4}{c}{\textbf{All channels}} \\
\midrule
Experimental IR                & \multicolumn{4}{c}{0.076 [-0.010, 0.163]} \\
Experimental extensive margin  & \multicolumn{4}{c}{0.138 [0.090, 0.185]} \\
Experimental intensive margin  & \multicolumn{4}{c}{-0.071 [-0.152, 0.010]} \\
\bottomrule
\end{tabular}%
}
\caption{Naive estimates for both the IR and intensive and extensive shares alongside our proposed estimates. Panel A shows observational baselines, which generally overestimate incrementality and intensive margins. Panel B shows experimental estimates. We note that the experimental estimate for the IR does yield valid inference, but only for incrementality under $\P_{\exp}$.}
\label{tab:allunits_baselines_ext}
\end{table}

\begin{figure}[H]
\centering
\includegraphics[width=\textwidth]{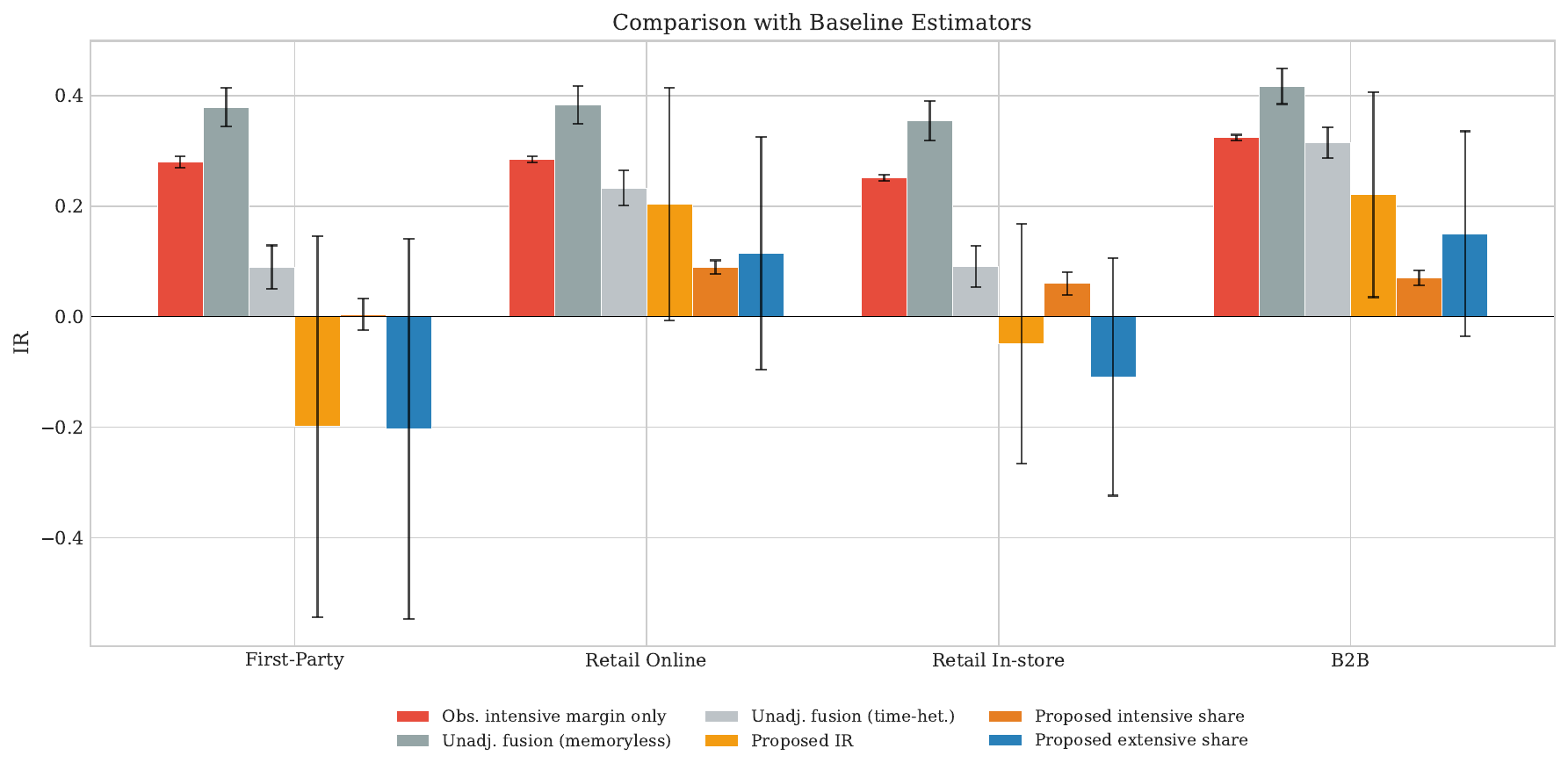}
\caption{A visualization of the observational baselines from Table~\ref{tab:allunits_baselines_ext} alongside 95\% confidence intervals.}
\label{fig:aucomp_baselines}
\end{figure}

\subsection{Incrementality Amongst Self-Gifters}
\label{sec:case_study:self_gift}

We conclude our empirical analysis by discussing how self-gifting impacts incrementality at Airbnb. Recall that, in our analysis, a \textit{self-gifter} is defined to be a customer who purchases their own gift cards to use towards a booking. This self-purchase can be motivated by perceived discounts or benefits (e.g., the ability to redeem credit card points for gift cards via certain B2B platforms), but does not strictly need to be. 
The potential contribution of self-gifters to the incrementality of a gift card program is not a priori obvious. On one hand, a firm may expect self-gifters to be non-incremental. This could be because they are likely to cannibalize out of pocket spend, finding discounted gift cards through third-party channels and redeeming them for pre-planned trips. On the other hand, one could expect these units to be highly incremental, only purchasing accommodations through Airbnb over a competitor because there were discounted gift cards available.
Based on the results from the previous section, one might speculate that the latter interpretation is correct. In particular, we have seen the highest incrementality in channels such as Retail Online and B2B, in which it is easiest for customers to leverage promotions or redeem loyalty points to obtain gift cards.

There is thus a natural question we aim to answer: is self-gifting actually incremental for businesses such as Airbnb?
We aim to answer this question in a more principled manner in this section. In particular, we show yearly estimates for both the IC and IR restricted to units who are suspected of self-gifting, with monthly estimates being presented in Appendix~\ref{app:case_study:monthly}. As before, we also interpret our findings in terms of intensive and extensive margins.

In general, it is impossible to tell if any given customer at Airbnb obtained their gift card through self-gifting. However, amongst units who book with a gift card, there exists a very natural proxy for self-gifting: the number of days between when a gift card was purchased and redeemed.
Throughout the sequel, we define any unit whose date of gift card \textit{redemption} occurs within a week of purchase to be a self-gifter. 
The intuition behind this definition is simple: self-gifters almost surely have pre-existing travel intent, and thus will likely purchase their gift cards in a short window prior to booking.
Likewise, when a unit receives a gift card through a true-gifting event (e.g., for a birthday or holiday), it is likely that either said gift card was purchased in advance (say, during holiday shopping) or that the giftee needs more than a week to plan their trip. 
While the proxy is not bullet-proof and there almost certainly exist some units who redeem a self-gifted card that is more than a week old (such as customers who self-purchase when they see a good discount), we see that the channels we earlier hypothesized offered the highest discounts in fact exhibit the highest rates of self-gifting (between 30-50\% depending on the month) under our definitions --- see Table~\ref{tab:sg_prevalence_yearly} below. 

\begin{table}[H]
\centering
\begin{tabular}{lcccc}
\toprule
\textbf{Channel} & \textbf{Min} & \textbf{Mean} & \textbf{Median} & \textbf{Max} \\
\midrule
First-Party        &  6.9 & 11.6 & 10.4 & 23.0 \\
Retail Online      & 25.4 & 45.8 & 49.3 & 56.1 \\
Retail In-store    & 13.7 & 20.4 & 20.4 & 24.4 \\
B2B                & 42.8 & 50.2 & 51.0 & 55.0 \\
\midrule
Cross-channel mean & 22.3 & 32.0 & 32.8 & 39.2 \\
\bottomrule
\end{tabular}
\caption{Yearly prevalence of suspected self-gifting by channel in percentage points. Additional aggregate statistics on monthly prevalence rates are displayed as well.}
\label{tab:sg_prevalence_yearly}
\end{table}

We now discuss our findings. We display our yearly estimates for incrementality amongst suspected self-gifters in Table~\ref{tab:sg_headline}. We see that estimates for both the IR and IC are generally much higher for self-gifters, with IR estimates being 2 and 1.5 times larger in Retail Online and B2B respectively. The IR estimates range from 0.342 in B2B to 0.504 in Retail In-Store. In other words, between a third and a half of the GBV generated by self-gifters in the considered channels is attributable to the existence of the gift card program itself.

One interesting finding is that the standard errors (and hence the confidence intervals) for our estimates are significantly smaller when restricted to suspected self-gifters. At first, this may appear counterintuitive: we are constructing our estimates with \textit{fewer} samples, and hence should have less confidence. However, this reasoning ignores the fact that we weight counterfactual outcomes for booking units by the inverse hazard rate $1/q_0(\tau, X)$. The hazard function generally \textit{decreases} as $\tau$ increases (i.e., units have a smaller per-month probability of booking), which ultimately reduces standard error despite the smaller sample size. 
Another interesting finding is that B2B, which offered the greatest estimate for the IR when all units were considered, actually exhibits the lowest IR when estimates are restricted to self-gifted units. 
Further, Retail In-Store is actually the \textit{most incremental} for self-gifters, despite having a near neutral estimate for the IR when the entire population is considered. 
This inversion is likely due to the sheer difference in suspected self-gifting prevalence between these two channels: around 50\% of units in B2B are considered to be self-gifting under our proxy, whereas only around 20\% of units self-gift for Retail In-Store.

\begin{table}[H]
\centering
\begin{tabular}{lcccc}
\toprule
& \multicolumn{2}{c}{\textbf{All units}} & \multicolumn{2}{c}{\textbf{Self-gifters only ($<7$d)}} \\
\cmidrule(lr){2-3}\cmidrule(lr){4-5}
\textbf{Channel} & \textbf{IR} & \textbf{95\% CI} & \textbf{IR} & \textbf{95\% CI} \\
\midrule
First-Party      & $-0.199$ & $[-0.544,\ 0.146]$ & $\mathbf{0.367}$ & $[0.253,\ 0.481]$ \\
Retail Online    & $0.204$  & $[-0.006,\ 0.415]$ & $\mathbf{0.408}$ & $[0.250,\ 0.565]$ \\
Retail In-store  & $-0.049$ & $[-0.265,\ 0.168]$ & $\mathbf{0.504}$ & $[0.354,\ 0.654]$ \\
B2B              & $0.221$  & $[0.035,\ 0.406]$  & $\mathbf{0.342}$ & $[0.207,\ 0.478]$ \\
\midrule
\textbf{Channel} & \textbf{IC} & \textbf{95\% CI} & \textbf{IC} & \textbf{95\% CI} \\
\midrule
First-Party      & $-0.676$ & $[-1.856,\ 0.505]$ & $\mathbf{0.934}$ & $[0.640,\ 1.227]$ \\
Retail Online    & $0.428$  & $[-0.019,\ 0.875]$ & $\mathbf{0.719}$ & $[0.441,\ 0.997]$ \\
Retail In-store  & $-0.131$ & $[-0.746,\ 0.485]$ & $\mathbf{0.899}$ & $[0.631,\ 1.167]$ \\
B2B              & $0.545$  & $[0.083,\ 1.008]$  & $\mathbf{0.798}$ & $[0.482,\ 1.114]$ \\
\bottomrule
\end{tabular}
\caption{Comparing yearly estimates of the IR and IC for all channels for self-gifters to estimates for all units. Both point estimates and marginally valid 95\% confidence intervals are shown.}
\label{tab:sg_headline}
\end{table}

Another interesting finding regarding our incrementality estimates is that both the intensive and extensive shares of the IR are much higher for self-gifters. Estimates for intensive margins increase by roughly 0.18 for First-Party and Retail In-Store, but increase to a lesser extent for Retail Online and B2B. However, what is especially interesting is that the estimates for extensive margins jump \textit{dramatically} when we restrict our attention to self-gifters. In particular, these margins are around 0.25-0.26 in Third-Party channels, indicating that there is (roughly) a 25\% increase in a unit's conditional probability of booking. We interpret this finding as follows. In general, when either all units or just self-gifting units are considered, gift cards appear to have a moderate effect on increasing a customer's out-of-pocket spend conditional on booking. However, gift cards seem to have a \textit{drastic positive effect} on a unit's conditional booking probability within the first month of receipt. For self-gifters, this might be a result of the discounts available from many Third-Party channels causing them to select Airbnb over an outside option. Further, if a unit observes that a discounted card is available through B2B or Retail Online channels, there could be some sort of advertising effect, either introducing or reminding said units about the services offered by Airbnb.

\begin{figure}[H]
\centering
\begin{subfigure}{0.46\textwidth}
  \centering\includegraphics[width=\textwidth]{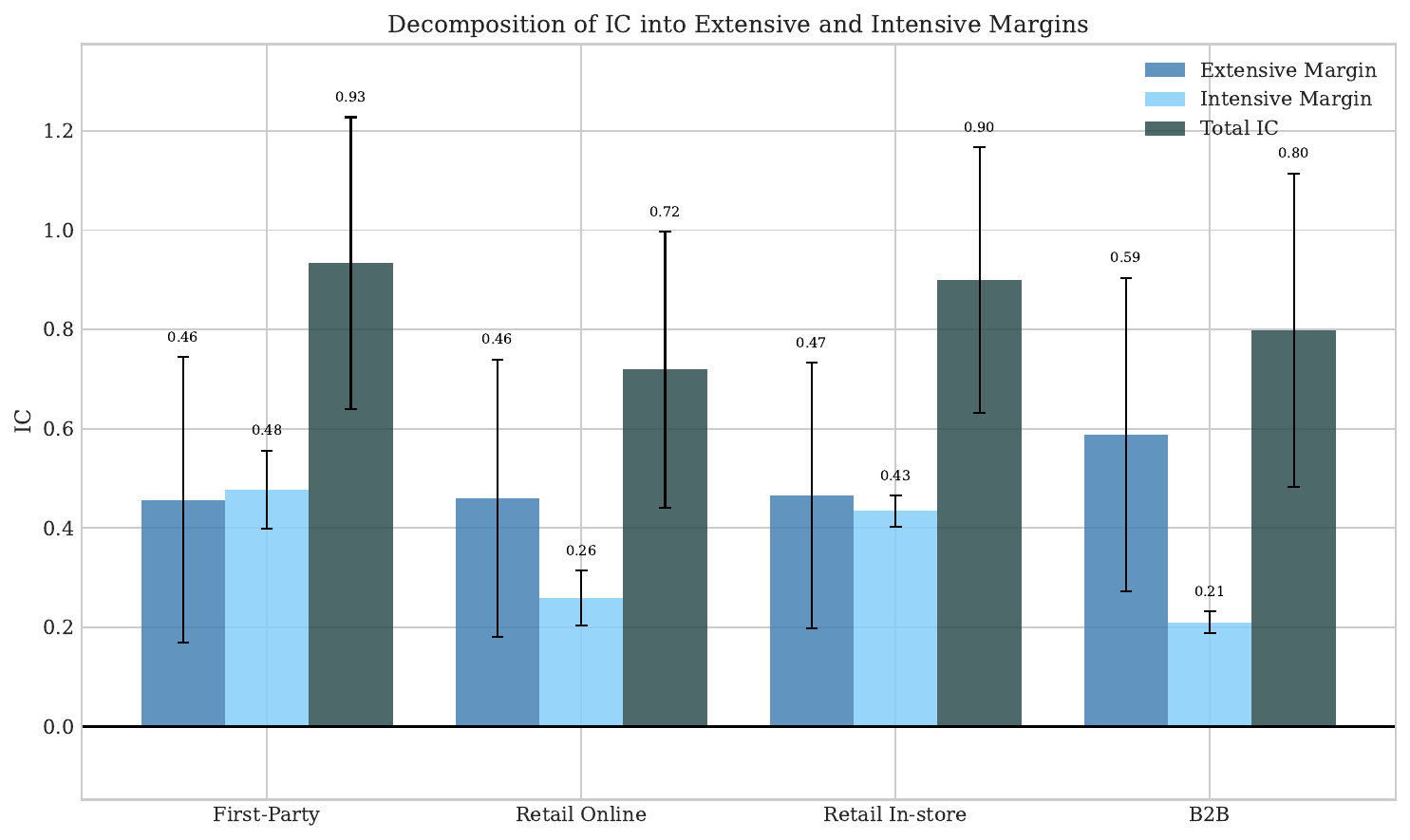}
  \caption{Decomposed IC.}
\end{subfigure}\hfill
\begin{subfigure}{0.53\textwidth}
  \centering\includegraphics[width=\textwidth]{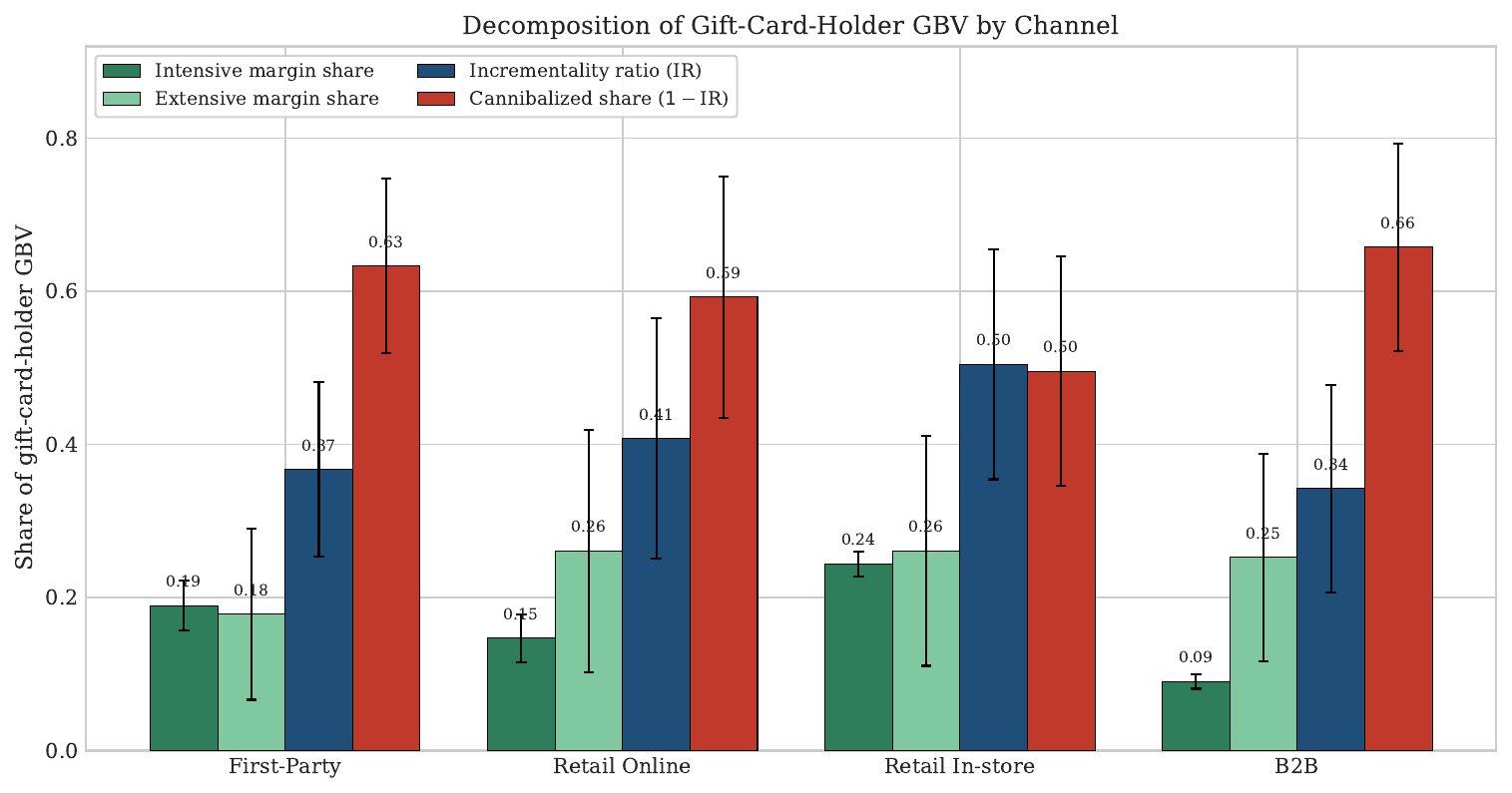}
  \caption{Decomposed IR.}
\end{subfigure}
\caption{Decompositions of the IC and IR, respectively, restricted to suspected self-gifters. Shown are both point estimates and  95\%
confidence intervals.}
\label{fig:sg_decomp}
\end{figure}

\begin{table}[H]
\centering
\resizebox{\textwidth}{!}{%
\begin{tabular}{lcccc}
\toprule
& \multicolumn{2}{c}{\textbf{Intensive-Margin Share}} & \multicolumn{2}{c}{\textbf{Extensive-Margin Share}} \\
\cmidrule(lr){2-3}\cmidrule(lr){4-5}
\textbf{Channel} & \textbf{All units} & \textbf{Self-gifters only} & \textbf{All units} & \textbf{Self-gifters only} \\
\midrule
First-Party     & 0.004 [-0.024, 0.032] & 0.189 [0.157, 0.222] & -0.203 [-0.546, 0.140] & 0.178 [0.066, 0.289] \\
Retail Online   & 0.089 [0.077, 0.102]  & 0.147 [0.116, 0.178] & 0.115 [-0.096, 0.325]  & 0.261 [0.103, 0.419] \\
Retail In-store & 0.060 [0.040, 0.081]  & 0.244 [0.228, 0.260] & -0.109 [-0.324, 0.106] & 0.261 [0.111, 0.411] \\
B2B             & 0.071 [0.057, 0.084]  & 0.090 [0.081, 0.099] & 0.150 [-0.035, 0.335]  & 0.252 [0.117, 0.388] \\
\bottomrule
\end{tabular}%
}
\caption{A comparison of intensive and extensive shares of the IR when computed on all units and suspected self-gifters, respectively. Again, point estimates are displayed alongside 95\% confidence
intervals.}
\label{tab:ir_shares}
\end{table}

\section{Qualitative Discussion of Findings}
\label{sec:discuss}

We now provide a brief, qualitative interpretation of our empirical findings from the previous section.
Appendix~\ref{app:discuss} provides an extended discussion of the behavioral mechanisms underlying our findings, the validity of our causal assumptions, and the relationship between incrementality and downstream profitability.

One important finding was that, although only B2B exhibited statistically significant incrementality, all third-party channels had small but positive intensive margins.  In words, we found that, conditional on a customer booking, the gift card program provides a small yet noticeable lift in terms of a customer's spend.
The possibility of such a lift on a customer's spend has been noted in existing works on gifting. For instance, \citet{norvell2017gift} note the possibility of a business having `` customers who would have visited
the brand anyway without the card, but end up spending more money
because of it.'' While their survey results regarding the gift card program at a national casual dining chain find that only 10\% of customers fall into this positive intensive margin category, these customers were highly incremental, spending approximately \$38.00 with a gift card where they would have spent only \$28.00 otherwise. 
Likewise, the qualitative literature on mental accounting \citep{thaler1985mental, prelec1998red, white2006format} notes that customers may end up spending more due to viewing the received gift as potential savings or may feel a lower ``pain of paying'' when leveraging gift card money. 
Our findings and methodology provide rigorous support for the insights of the aforementioned works, at least when restricted to third-party channels of gift card distribution. We also further speculate that, given a larger experimental sample, we would find both Retail Online and Retail In-Store to be positively incremental. Since intensive margins were identifiable from observational data alone, our method was generally able to produce small confidence intervals. However, extensive margins (and hence the IR and IC) depend on the experimental sample for identification, and hence the quality of our estimates for these parameters is fundamentally limited by the experimental sample size. This is further exacerbated by the fact that we divide by the conditional booking hazard $q_0(\tau, x)$, which can become quite small for some units booking in later months (e.g., for units with $\tau = 11$ or $\tau = 12$).

Our most surprising finding was that incrementality estimates (both the IC and IR) were highest when we restricted our estimation to suspected self-gifters (see Section~\ref{sec:case_study:self_gift}). We saw that all channels, even including First-Party and Retail In-Store, were highly incremental, with IR estimates indicating that one third to one half of GBV among such units was attributable to the existence of the gift card program. Not only did self-gifters genuinely spend more (i.e., had higher intensive margins), but they also exhibited a significantly higher extensive margin, or lift in booking rate. 
For these units, it is not generally reasonable to assume that the lift in booking probability comes solely from some sort of an advertising effect. In fact, many such units likely possess preexisting travel intent and then search for gift cards with effective discounts, e.g., those offered through rewards programs or those purchasable at businesses offering cash back.\footnote{We note that such advertising effect pathways are possible in self-gifting scenarios. For instance, a unit could come across a discounted gift card at a store, which thus may induce a desire to travel.}
In the absence of an advertising effect, one might expect an economically rational individual to have zero extensive margins --- the value on a gift card is generally fungible with raw cash that could be put towards a trip, and hence booking behavior should not change. 
However, this reasoning fails to capture implicit or explicit discounts that are applicable to gift cards obtained through third-party providers. We provide a toy behavioral economic model in Appendix~\ref{app:discuss:econ_model} in which positive extensive margins are realizable even \textit{without} any intensive margins. This model is not assumed throughout the paper, and is meant to merely be illustrative as to how such phenomena may arise in practice. 
Further, extensive margins may be positive even without an advertising effect if some individuals do not possess an instrument (such as a credit or debit card) that can be used to  make an online purchase. For these units, gift cards may be the only means of booking travel with Airbnb.
In sum, our findings indicate that self-gifters are \textit{highly-incremental}  for a business's gift card program, offering incremental visits/purchases and thus high extensive margins. 

The incrementality associated with effective gift card discounts and self-gifting also helps us interpret why First-Party generally exhibits the lowest incrementality. In this channel, there were no incentives (e.g., cash back programs, discounts, or reward points) for self-gifting in the study period. Thus, any lift to extensive margins should only come solely from true-gifting and the corresponding advertising effect.
Naturally, a business would desire for their first-party channels to be the most incremental, as these channels tend to exhibit lower cost structure.\footnote{As \citet{norvell2017gift} note, third-party channels generally acquire gift cards at lower than market rates. Sometimes, these channels obtain gift cards at 10-15\% discounts.}
Perhaps one means by which a business could bolster the extensive margins in their first-party offerings would be to occasionally offer discounted gift cards or to offer gift cards to users who spend over a certain threshold.
In other words, the business could attempt to improve incrementality in this channel by offering comparable value propositions and customer benefits to third-party channels.

Finally, we emphasize two qualifications to the discussion of our empirical findings. First, the causal interpretation of our estimates naturally depends on the assumptions outlined in Section~\ref{sec:model}, namely conditional exogeneity and transferability, being satisfied. 
In Appendix~\ref{app:discuss:assumptions}, we discuss the primary threats to these assumptions in the context of our application, including latent differences in self-gifting behavior and travel intent, inherent limitations in feature construction, and differences between the experimental and observational populations. Second, our analysis measures the \textit{incrementality} of gift card programs rather than their ultimate \textit{profitability}. Translating the IC and IR into profits additionally requires channel-specific information on quantities such as distribution costs, gift card discounts, transaction fees, and breakage. We discuss this distinction further in Appendix~\ref{app:discuss:profit}.

\section{Conclusion}
\label{sec:conclusion}

In this paper, we developed a formal causal framework for measuring the incrementality of a business's gift card program. We proposed two scale-free quantities for measuring the incrementality of a gift card program: the \textit{incrementality ratio}, which quantified the fraction of the program's incremental revenue attributable to the existence of the gift card program, and the \textit{incrementality coefficient}, which captured the average number of incremental dollars generated for each one dollar of gift card value redeemed. 
While neither quantity is identified in terms of the observational data a business would have access to (as it is generally not possible to know which units possess gift cards prior to redemption), we showed that identification becomes possible when the business has access to a mostly-unrelated, small-scale experiment.
These two data sources only had to be linked through a weak \textit{transferability} assumption, which posits that the \textit{ratio between certain conditional booking probabilities} is invariant between the two populations. We then derived a natural cross-fitting estimator for the aforementioned effects and showed how a learner may use flexible, ML estimators for estimating unknown nuisances.

We then applied our estimation/identification framework to a mix of observational and experimental booking data at Airbnb. In particular, we measured the incrementality of their gift card program across four large North American channels of distribution. We found that third-party channels (B2B, Retail In-Store, and Retail Online) generally exhibited higher incrementality than First-Party, with the incrementality in B2B being statistically significant.
Because our method also provided a natural decomposition into intensive and extensive margins, we were able to see that all third-party channels offered moderately positive yet significant boosts to intensive margins.  Further, we found the unexpected result that self-gifters, or units likely to have purchased gift cards for themselves, are actually highly incremental, with a large share of their incrementality contributing to extensive margin shares.

We conclude by briefly discussing future research directions that we believe to be of high impact to literature on gifting. One practically-relevant direction would be to develop policy learning algorithms for optimally assigning customers to promotional gift cards. Similarly, one could aim to understand the \textit{effect of targeting},  or the gap between the value of the optimal gift card assignment policy and a baseline policy, e.g., the one that assigns all units to the ``control'' status of not receiving a gift card. 
Lastly, one may want to develop confounding-robust identification strategies for the estimands considered in this paper, such as difference-in-differences approaches that leverage parallel trends or more general panel data methods.
Each of these approaches would naturally need to handle the  partial-observation problem encountered throughout this paper, and thus may require nontrivial identification arguments.

\paragraph{AI Usage Statement}
The authors used generative AI tools to assist with searching for related work, checking theorem statements and identifying minor errors in proofs, writing code (including code used to produce figures and tables), and copyediting. The authors reviewed and verified all AI-assisted output and remain fully responsible for the content of this manuscript.

\bibliography{bib}
\bibliographystyle{plainnat}

\appendix

\section{Causal Interpretation of Intensive and Extensive Margins}
\label{app:margins}

\newcommand{\DE}{\mathrm{DE}}
\newcommand{\IE}{\mathrm{IE}}

In this appendix, we formally connect the intensive and extensive margins $\Delta_{0, \intm}$ and $\Delta_{0, \extm}$ to the natural \textit{direct} and \textit{indirect} effects that would arise in a formal causal mediation analysis. We show this equivalence under the assumption of \textit{conditional exogeneity of the booking event}, which is slightly stronger than the assumption of conditional independence of the potential outcomes $Y(0)$ and the treatment $G$. 

To formally define the aforementioned direct and indirect effects, and to posit our assumptions, we need to introduce a broader set of potential outcomes. First, we define \textit{the potential booking decision} $B(g)$ and the indicator of a customer's booking decision had they been experimentally assigned to a gift card receipt status $g \in \{0, 1\}$. Further, we extend our indexing for \textit{potential outcomes}, letting $Y(g,b)$ denote the realized spend amount if we were to intervene and force a user to have gift card status $g \in \{0, 1\}$ and booking status $b \in \{0, 1\}$. We use the short-hand notation $Y(g)=Y(g, B(g))$, which corresponds to the potential outcome when we only intervene on the gift card receipt. We make the following consistency assumption. 

\begin{ass}
\label{ass:consistency_extended_potential_outcomes}
These potential outcomes are consistent: $Y = Y(G, B)$ and $B= B(G)$.%
\end{ass}

Using our notation, we can decompose the incremental revenue into natural direct and indirect effects. Recall that the incremental revenue was defined as
\[
\Delta_0 = \E_{\obs}[G(Y - Y(0))] = E_{\obs}\left[G\big(Y(1, B(1)) - Y(0, B(0))\big)\right],
\]
where the second equality follows from the assumption of consistency. The direct and indirect effects are now obtained by subtracting and adding $\E_{\obs}[G Y(0, B(1))]$:
\begin{align*}
\Delta_{0, \DE} &= \E_{\obs}\left[G\big(Y(1, B(1)) - Y(0, B(1))\big)\right]\\
\Delta_{0, \IE} &= \E_{\obs}\left[G\big(Y(0, B(1)) - Y(0, B(0))\big)\right].
\end{align*}
In words, the direct effect captures the expected change in spend amongst gift card recipients were their booking decisions to remain fixed but their gift card receipt status to be changed from $G = 1$ to $G = 0$. Intuitively, this quantity naturally corresponds to the intensive margin, as it isolates the effect of gift card receipt on spend keeping booking decisions frozen. Likewise, the indirect effect captures the residual change in spend due to a customer's booking decision changing when they are experimentally assigned to $G = 0$. Very heuristically, this can be viewed as measuring a weighted change in a customer's booking probability between when they are assigned a gift card and when they are not, i.e., something like an extensive margin. We now make these analogies between direct/indirect effects and intensive/extensive margins more rigorous.

\begin{lemma}
\label{lem:mediation}
Suppose Assumptions~\ref{ass:exogeneity}, \ref{ass:positivity}, \ref{ass:partial-obs}, \ref{ass:transfer}, \ref{ass:cross_positivity}, and \ref{ass:consistency_extended_potential_outcomes} hold. As before, suppose that the potential outcome for spend is zero under the no-booking intervention, i.e., $Y(g, 0)=0$ almost surely. Lastly, suppose that the following conditional mean condition holds\footnote{We note that this is implied by the more substantive assumption that $Y(g, b) \independent (B, G) \mid X$ for all $g, b \in \{0, 1\}$ under $\P_{\obs}$.}:
\begin{align}
\label{eq:mean-swap}
\E_{\obs}(Y(0, 1) \mid X = x, S = 1) = \E_{\obs}(Y(0, 1) \mid X = x, D = 1) \quad \P_{\obs}\text{-almost surely.}
\end{align}
Then:
\begin{align*}
\Delta_{0, \DE} &= \E_{\obs}\left[S\left\{Y - g_0(X)\right\}\right] \equiv \Delta_{0, \intm}\\
\Delta_{0, \IE} &= \E_{\obs}\left[S g_0(X)\left\{1 - \eta_0(X)\right\}\right]\\
&=\E_{\obs}\left[\pi_0(X)g_0(X) \big(q_{\obs}(X) - p_{\obs}(X)\big)\right] \equiv \Delta_{0, \extm}.
\end{align*}
\end{lemma}

\begin{proof}
As noted above, by the definition of $Y(0)$ and from consistency (Assumption~\ref{ass:consistency_extended_potential_outcomes}), we have the following equivalences:
\begin{align*}
\E_{\obs}[G Y(1, B(1))] = \E_{\obs}[G Y] \quad \text{and} \quad \E_{\obs}[G Y(0, B(0))] = \E_{\obs}[G Y(0)]
\end{align*}
To show the equivalence, all that remains is to show the following:
\[
\E_{\obs}\left[G \cdot Y(0, B(1))\right] = \E_{\obs}\left[S \cdot g_0(X)\right].
\]
First, note that by booking consistency and the fact $Y(g, 0)=0$ almost surely, we have
\begin{align*}
\E_{\obs}\left[G \cdot Y(0, B(1))\right] &= \sum_{b \in \{0, 1\}}\E_{\obs}\left[G \1\{B = b\} Y(0, b)\right] = \E_{\obs}[S Y(0, 1)].
\end{align*}
Using this identity and recalling the definition $r_0(x) := \P_{\obs}(S = 1 \mid X = x)$, we see that we have by the tower rule for conditional expectations that
\begin{align*}
\E_{\obs}\left[S \cdot Y(0, 1)\right] &= \E_{\obs}\left[\E_{\obs}\left(S \cdot Y(0, 1) \mid X\right)\right] \\
&= \E_{\obs}\left[r_0(X) \cdot \E_{\obs}\left(Y(0, 1) \mid X, S = 1\right)\right] \\
&= \E_{\obs}\left[r_0(X) \cdot \E_{\obs}\left(Y(0, 1) \mid G = 0, B = 1, X\right)\right] \\
&= \E_{\obs}\left[r_0(X) \cdot \E_{\obs}(Y \mid G = 0, B = 1, X)\right] \\
&= \E_{\obs}\left[r_0(X) \cdot g_0(X)\right],
\end{align*}
where the second and third equalities follow from our exchangeable mean condition (Equation~\eqref{eq:mean-swap}) and the final equality follows again from consistency. Observing the identity:
\[
\E_{\obs}\left[r_0(X) \cdot g_0(X)\right] = \E_{\obs}\left[S \cdot g_0(X)\right]
\]
finishes the proof.
\end{proof}

\section{Type-Heterogeneous Identification}
\label{app:attribute-type-proof}

In this appendix, we prove Theorem~\ref{thm:ident-type} from Section~\ref{sec:model:heterogeneity}. The proof is largely the same as the proof of Theorem~\ref{thm:identification}, modulo the fact that we must take care to handle a customer's gift card type.

\begin{proof}[Proof of Theorem~\ref{thm:ident-type}]
As before, we just need to argue that both terms $\E_{\obs}[GY]$ and $\E_{\obs}[\pi_0(X)\mu_0(X)]$ in the definition of the incremental revenue $\Delta_0$ are identified.

\textit{Step 1: Identifying the treated outcome $\E_{\obs}[GY]$.} Again, although we do not observe $G$ directly, we can recover this term by exploiting the zero-spend assumption. A customer contributes to $GY$ only if they both received a gift card \textit{and} spent a positive amount. However, spending a positive amount requires making a purchase ($B = 1$). Thus $GY = G \cdot B \cdot Y = \sum_{\tau\in\calJ} S_\tau \cdot Y$, and thus we have
\[
\E_{\obs}[GY] = \E_{\obs}\left[Y\sum_{\tau\in\calJ} S_\tau\right].
\]

\textit{Step 2: Identifying the counterfactual spend $\E_{\obs}[\pi_0(X)\mu_0(X)]$.} The final component is the expected counterfactual spend that gift card recipients would have generated without the program. The outcome regression $\mu_0(x) = \E_{\obs}(Y \mid G = 0, X = x)$ can again be written as
\[
\mu_0(x) = \underbrace{\E_{\obs}(Y \mid G = 0, B = 1, X = x)}_{= \, g_0(x)} \times \P_{\obs}(B = 1 \mid G = 0, X = x).
\]
Further, we can write $\pi_0(X)$ under type-heterogeneity as
\begin{align*}
\pi_0(X) := \P_{\obs}(G = 1 \mid X)
&= \sum_{\tau\in\calJ} \P_{\obs}(G_\tau = 1 \mid X) \\
&= \sum_{\tau\in\calJ} \frac{\P_{\obs}(B = 1, G_\tau = 1 \mid X)}{\P_{\obs}(B = 1 \mid G_\tau = 1, X)} \\
&= \sum_{\tau\in\calJ} \frac{\P_{\obs}(S_\tau = 1 \mid X)}{\P_{\obs}(B = 1 \mid G_\tau = 1, X)},
\end{align*}
where the second equality follows from the fact that $G  = \sum_{\tau \in \calJ}G_\tau$ and the third equality follows from Bayes theorem. 
Putting these two equalities together, we obtain identification as follows:
\begin{align*}
\E_{\obs}[\pi_0(X)\mu_0(X)] &= \E_{\obs}\left[\sum_{\tau \in \calJ} \frac{g_0(X)\P_{\obs}(B = 1 \mid G = 0, X)\P_{\obs}(S_\tau = 1 \mid X)}{\P_{\obs}(B = 1 \mid G_\tau = 1, X)}\right] \\
&= \E_{\obs}\left[\sum_{\tau \in \calJ} \frac{g_0(X)\P_{\exp}(B = 1 \mid G = 0, X)\P_{\obs}(S_\tau = 1 \mid X)}{\P_{\exp}(B_t = 1 \mid B_{<t} = 0, G_a = 1,  X)}\right] \\
&= \E_{\obs}\left[g_0(X)p_0(X)\sum_{\tau \in \calJ} \frac{\P_{\obs}(S_\tau = 1 \mid X)}{q_0(\tau, X)}\right] \\
&= \E_{\obs}\left[g_0(X) \, p_0(X) \, \sum_{\tau\in\calJ} \frac{S_\tau}{q_0(\tau, X)}\right].
\end{align*}
This completes the identification of $\Delta_0$. The identification for $\rho_0$ and $\kappa_0$ now follow immediately, since we have already argued that the denominator $\E_{\obs}[G \cdot Y]$ is identified from data and the identification of $\E_{\obs}[G \cdot A]$ is exactly analogous. 
\end{proof}

\section{Derivation of Neyman Orthogonal Moments}
\label{sec:neyman_construction}

\newcommand{\tr}{\mathrm{tr}}
\newcommand{\targ}{\mathrm{targ}}

The goal of this appendix is to derive and prove the convergence of a cross-fit estimator for the incremental revenue $\Delta_0 := \E_{\obs}\left[G \big(Y - Y(0)\big)\right]$. To this end, we first derive a \textit{Neyman orthogonal} population moment/estimating equation for $\Delta_0$. Then, we leverage this moment to construct the estimator outlined in Algorithm~\ref{alg:dml}.
In particular, since this estimator is based on a Neyman orthogonal score, we will be able to prove convergence under generally mild assumptions on the convergence of nuisance estimates. 
Finally, we will show convergence of the estimators for the incrementality coefficient $\kappa_0$ and the incrementality ratio $\rho_0$ using the delta method.

\subsection{A General Recipe for Neyman Orthogonal Moments}

Suppose a learner is interested in estimating the value of a population moment $\theta_0 = M^{\base}(f_0)$, where $M(f) := \E_{\targ}\left[m(Z; f)\right]$. Here, $m$ is a sample moment depending on  a random variable $Z \sim \P_{\targ}$ and an unknown, true vector of nuisance functions $f_0(X) = (f_{0, 1}(X), \dots, f_{0, p})$, where we assume  $X$ consists of a subset of the coordinates of $Z$. For simplicity, we will assume that the learner knows the moment $M(f)$ exactly for each $f$, but that the true \textit{nuisance function} $f_0(x)$ is unknown.\footnote{In practice and in the following subsections on estimation, only the sample moment $m(z; f)$ is known.  However, we aim to isolate the impact of unknown nuisances in this subsection to provide intuition for our debiasing corrections.} Instead, all the learner knows is that $f_0$ satisfies the conditional moment restriction
\[
0 = \E_{\tr,i}\left[Q_i(C_i - f_{0, i}(X)) \mid X\right] \quad \text{almost surely},
\]
where $\P_{\tr, i}$ denotes some \textit{training distribution} that may generally be different from the target distribution $\P_{\targ}$, $C$ is some real-valued random variable, and $Q \in \{0, 1\}$ is a binary random variable. 

Given this setup, the natural approach of the learner would be as follows: first, fit an estimate $\wh{f}$ of $f_0$ using the moment restriction, then produce an estimate $\wh{\theta} = M(\wh{f})$.
However, without modification, the population moment $M(f)$ may be very sensitive to errors in nuisance estimation, i.e., $M(\wh{f})$ may be very far from $M(f_0)$ even when an estimate $\wh{f}$ of the true nuisance vector $f_0$ is fairly accurate. This can be seen from a heuristic first-order Taylor expansion:
\begin{equation}
\label{eq:heuristic_taylor}
M(\wh{f}) - M(f_0) \approx D_g M(f_0)(\wh{f} - f_0) + \frac{1}{2}D_{f}^2 M(f_0)(\wh{f} - f_0) + \Rem(\wh{f}, f_0).
\end{equation}
Here, we have defined the \textit{Gateaux derivative} of $M(f)$ at $f_0$ in the direction $\nu : \calX \rightarrow \R^p$ as the following directional derivative:
\[
D_f M(f_0)(\nu) := \frac{d}{d t}\E_{\targ}\left[m(Z; f_0 + t \nu)\right]\Big|_{t = 0} 
\]
The second Gateaux derivative/Hessian is analogously defined, being given by
\[
D_f^2 M(f_0)(\nu) := \frac{\partial^2}{\partial t^2}\E_{\targ}\left[m(Z; f_0 + t\nu)\right]\Big|_{t = 0,}.
\]

Rephrasing things, under the expansion outlined in Equation~\eqref{eq:heuristic_taylor}, the best a learner can generically hope for is that the error scales \textit{linearly} in the nuisance error.
Thus, to make estimation practical, a learner should aim to derive a \textit{Neyman orthogonal} population moment, or one that is insensitive to first-order errors in nuisance estimation. One natural approach to improve the performance/convergence of the moment would be to add a \textit{correction} or \textit{de-biasing} term that weights the residuals for the regression $\wh{f} = (\wh{f}_1, \dots, \wh{f}_p)$ according to some functions $\alpha(X) = (\alpha_1(X),  \dots, \alpha_p(X))$, i.e., define 
\[
M_{\corr}(f, \alpha) := \E_{\targ}\left[m(Z; f)\right] + \sum_{j = 1}^p \E_{\tr, j}\left[\alpha_j(X)Q_j\left\{C_j  - f_j(X)\right\}\right].
\]
Noting that $\E_{\tr, j}\left[\alpha_j(X)Q_j\left\{C_j - f_j(X)\right\}\right] = 0$ at $f \equiv f_0$ for \textit{any} $\alpha_j(X)$ (which follows from the tower rule), we see that we still have $\theta_0 = M_{\corr}(f_0, \alpha)$ for any $\alpha$.

To eliminate first-order nuisance estimation errors in estimating $f_0$, we should aim to select optimal weights $\alpha_{0, j}(X)$ that ensure:
\[
D_{f}M_{\corr}(f_0, \alpha_0)(\nu) = 0 
\]
where $\nu : \calX \rightarrow \R^p$ represents a perturbation function. Under standard regularity conditions, the \textit{Riesz representation theorem}~\citep{lax2014functional} yields the existence of a function $\beta_{0, j} : \calX \rightarrow \R$ such that
\[
D_{f_j}\E_{\targ}\left[m(Z; f_0)\right](\nu_{j}) = \E_{\targ}\left[\beta_{0, j}(X)\nu_{j}(X)\right]
\]
We can use this representer $\beta_{0, j}(X)$ to derive the optimal weighting function for the $j$th residual, $\alpha_{0, j}(X)$. Note that we need the following condition to hold for all perturbations:
\[
D_{f_j}\E_{\targ}\left[m(Z; f_0)\right](\nu_j) = \E_{\tr, j}\left[\alpha_{0, j}(X)Q_j\nu_j(X)\right].
\]
If we suppose that the density ratio $\omega_{0, j}(X) := \frac{d\P_{\targ}}{d\P_{\tr, j}}(X)$ exists for each $j \in [p]$, then it is clear that we need to take $\alpha_{0, j}(X)$ to be
\[
\alpha_{0, j}(X) := \frac{\beta_{0, j}(X)\omega_{0, j}(X)}{\P_{\tr, j}(Q_j = 1 \mid X)}.
\]

In this paper, we are precisely in a special case of the framework outlined above. In the next subsection, we derive both the appropriate residuals and de-biasing weights in a \textit{nonparametric setting}, in which we place no restrictions on the functional forms of the outcome regression $g_0(X)$ or the conditional booking probabilities $p_0(X)$ and $q_0(\tau, X)$. We then consider a \textit{semiparametric setting} in the sequel, where we still permit $g_0$ to be an arbitrary function, but restrict $p_0$ and $q_0$ to be logistic regressions in a known basis expansion. 

\subsection{Nonparametric Orthogonal Scores for $\Delta_0$}

We now apply this recipe to the base moment for the incremental revenue $\Delta_0$. Recall that we could write 
\[
\Delta_0 = \E_{\obs}\left[m(Z; g_0, p_0, q_0)\right], \quad \text{where}\quad m(Z; g, p, q) := Y\sum_{\tau \in \calJ}S_\tau - g(X)p(X)\sum_{\tau \in \calJ}\frac{S_\tau}{q(\tau, X)}.
\]
Here, there are two types of nuisances, all of which are specified as regressions. We detail the derivation of their optimal weights/corresponding correction terms in sequence.

\paragraph{Correction for the Organic Spend Function $g_0(X)$.} 
First, recall that the organic spend function $g_0 : \calX \rightarrow \R$ was defined as $g_0(x) := \E_{\obs}\left(Y \mid G = 0, B = 1, X = x\right)$. Equivalently, this can be cast as the unique zero of the conditional moment restriction
\[
0 = \E_{\obs}\left(D\left\{Y - g(X)\right\} \mid X\right) \quad \text{where } D = B \cdot (1 - G).
\]
Thus, we are precisely in the setting described above, with $Q_g := D$ and $C_g := Y$. We can further directly compute the corresponding Riesz representer $\beta_{0, g}(X)$. Direct computation yields that the representer is defined as
\[
\beta_{0, g}(X) := - p_0(X)\sum_{\tau \in \calJ}\frac{r_0(\tau, X)}{q_0(\tau, X)},
\]
where we have defined $r_0(\tau, x) := \P_{\obs}(S_\tau = 1 \mid X = x)$.
Thus, since the nuisance is defined on the observational population, if we let $\ell_0(x) := \P_{\obs}(D = 1 \mid X = x)$, the optimal choice of weighting can be seen to be
\begin{align}
\label{eq:alpha_g_app}
\alpha_{0, g}(X) &:= -\frac{p_0(X)}{\ell_0(X)}\sum_{\tau \in \calJ} \frac{r_0(\tau, X)}{q_0(\tau, X)}.
\end{align}

\paragraph{Correction for Booking Probabilities $p_0$ and $q_0$.}
Both booking propensities $p_0$ and $q_0$ are defined as conditional probabilities (and hence, regressions) on the \textit{experimental population}. More precisely, the conditional booking probability \textit{without a gift card} is given by $p_0(x) := \P_{\exp}(B_1 = 1 \mid G = 0, X = x)$ and the booking hazard \textit{with a gift card} is given as $q_0(\tau, x) := \P_{\exp}\left(B_t = 1 \mid B_{< t} = 0, G_a = 1, X = x\right)$. Hence, these nuisances satisfy the following experimental moment restrictions:
\[
0 = \E_{\exp}\Big((1 - G)\left\{B_1 - p_0(X)\right\} \mid X\Big) \quad \text{and} \quad  0 = \E_{\exp}\Big(G_a (1 - B_{< t}) \left\{B_t - q_0(\tau, X)\right\} \mid X\Big).
\]
Thus, in the notation of the previous subsection, we have $Q_p = 1 - G$, $C_p = B_1$, $Q_{q, \tau} := G_a(1 - B_{< t})$, and $C_{q, \tau} = B_t$. The observational representers $\beta_{0, p}(X)$ and $\beta_{0, \tau}(X)$ for these nuisances can readily be computed, and are given as follows
\begin{align*}
\beta_{0, p}(X) &=  -g_0(X)\sum_{\tau \in \calJ}\frac{r_0(\tau, X)}{q_0(\tau, X)} \\
\beta_{0, \tau}(X) &= \frac{g_0(X)p_0(X)r_{0}(\tau, X)}{q_0(\tau, X)^2}.
\end{align*}
Now, to ensure the optimal weights $\alpha_{0, g}(X)$, $\alpha_{0, p}(X)$, and $\alpha_{0, \tau}(X)$ are well-defined, we need the following assumption, which ensures the existence of a density ratio between the observational and experimental covariate distributions.

\begin{ass}
\label{ass:overlap}
Let $P_X$ denote the marginal distribution of $X$ for $P \in \{\P_{\obs}, \P_{\exp}\}$. We assume absolute continuity: $\P_{\obs, X} \ll \P_{\exp, X}$, ensuring the density ratio $\omega_0(X) := \frac{d\P_{\obs}}{d\P_{\exp}}(X)$ is well-defined.
\end{ass}

Under this assumption, the optimal weights are given as follows:
\begin{align*}
\alpha_{0, p}(X) &:= - \frac{\omega_0(X)g_0(X)}{1 - e_0(X)}\sum_{\tau \in \calJ}\frac{r_0(\tau, X)}{q_0(\tau, X)} \\
\alpha_{0, \tau}(X) &:= \frac{\omega_0(X)g_0(X)p_0(X)r_{0, \tau}(X)}{\lambda_0(\tau, X)q_0(\tau, X)^2},
\end{align*}
where we have defined the experimental propensity as $e_0(x) := \P_{\exp}(G = 1 \mid X = x)$ and the experimental probability of not booking in the first $t - 1$ months with a gift card of type $a$ as $\lambda_0(\tau, x) := \P_{\exp}(B_{< t} = 0, G_a = 1 \mid X = x)$.

\paragraph{Putting Things Together.}

Combining the above derivations, we arrive at the nonparametric, Neyman orthogonal moment for $\Delta_0$:
\begin{align*}
M_{\corr}(g, p, q, \alpha) &:= \E_{\obs}\left[\psi_{\obs}(W; g, p, q, \alpha_g)\right] + \E_{\exp}\left[\alpha_p(X)(1 - G)\left\{B_1 - p(X)\right\}\right] \\
&\qquad + \E_{\exp}\left[\sum_{\tau \in \calJ}\alpha_{\tau}(X) G_a(1 - B_{< t})\left\{B_t - q(\tau, X)\right\}\right],
\end{align*}
where we have defined:
\[
\psi_{\obs}(W; g, p, q, \alpha_g) := m(W; g, p, q) + \alpha_g(X)D\left\{Y - g(X)\right\}.
\]

\subsection{Semiparametric Corrections under Logistic Booking Propensity}

To avoid estimating the density ratio $\Delta_0(X)$, one may assume the conditional booking propensities follow a logistic regression model with corresponding feature embeddings $\phi_p(X)$ and $\big(\phi_\tau(X) : \tau \in \calJ\big)$ (Assumption~\ref{ass:logistic}). By standard results on M-estimation, the true parameters $\beta_{0, p}$ and $(\beta_{0, \tau} : \tau \in \calJ)$ can be written as the solution of the following experimental moments:
\begin{align*}
\E\left[s_p(Z; \beta_{0, p})\right] = 0 & &\text{where }& s_p(Z; \beta_{p}) := (1 - G)\left\{B_1 - \sigma(\phi_p(X)^\top \beta_p)\right\}\phi_p(X) \\
\forall \tau \in \calJ,\;\; \E\left[s_\tau(Z; \beta_{0, \tau})\right] = 0 & &\text{where }& s_\tau(Z; \beta_\tau) := G_a(1 - B_{< t})\left\{B_t - \sigma(\phi_\tau(X)^\top\beta_\tau)\right\}\phi_\tau(X).
\end{align*}
Note that the zeros of these moments will not generally be unique. Indeed, we characterize the set of zeros in Lemma~\ref{lem:logistic_solutions} in Appendix~\ref{app:estimation_proofs}, showing that the set of all solutions forms an affine subspace of $\R^{d_p}$ and $\R^{d_\tau}$, respectively. We thus canonically refer to $\beta_{0, p}$ and $\beta_{0, \tau}$ as the minimal norm solutions, which are guaranteed to be unique. We generally let the stacked version of these moments be denoted by $s(Z; \beta) = \left(s_p(Z; \beta_p), \left\{s_\tau(Z; \beta_\tau) : \tau \in \calJ\right\}\right)$, where  $\beta = (\beta_p, \{\beta_\tau : \tau \in \calJ\})$ and $\beta_0$ denotes the minimum norm solution to the stacked estimating equations. 

In this parametric setting, instead of viewing the base moment $m(W; g, p ,q)$ as depending on nonparametric conditional booking probabilities $p(x)$ and $q(\tau, x)$, we can view them as depending on $\beta$ instead. We can simply write $p(x) := \sigma(\phi_p(x)^\top \beta_p)$ and $q(\tau, x) := \sigma(\phi_\tau(x)^\top \beta_\tau)$ for notational convenience in the sequel, defining $p_0(x)$ and $q_0(\tau, x)$ similarly in terms of $\beta_{0,p}$ and $\beta_{0, \tau}$, respectively. In this setting, instead of reweighting the experimental residuals by some function $\alpha(X)$, it suffices to consider simply a weight matrix $A \in R^{1 \times d}$, where again $d = d_p + \sum_{\tau \in \calJ} d_\tau$. In particular, we would consider the corrected moment
\[
M_{\corr}(g, \beta, \alpha_g, A) := \E_{\obs}[\psi_{\obs}(w; g, \beta, \alpha_g)] + A\E_{\exp}[s(Z; \beta)].
\]
The ideal matrix $A_0$ should be selected so as to ensure that the 
\[
D_\beta \E_{\obs}\left[\psi_{\obs}(W; g_0, \beta_0, \alpha_{0, g})\right] \equiv D_\beta \E_{\obs}\left[m(W; g_0, \beta_0)\right] = - A_0D_\beta\E_{\exp}[s(Z; \beta_0)],
\]
where ``$\equiv$'' denotes that the first two terms are equivalent, as the nonparametric correction $\alpha_g(X)D\left\{Y - g(X)\right\}$ does not depend directly on $\beta$. Here, $D_\beta$ denotes the Jacobian with respect to $\beta$.\footnote{Using the same notation as for the Gateaux derivative is for convenience. For a function $F : \R^d \rightarrow \R^{d'}$, we correspondingly write the application of the Jacobian to a vector $v \in \R^d$ at $v_0$ as $D_\beta F(v_0) (v)$.
} 
Such an $A_0$ will exist so long $D_\beta \E_{\obs}\left[m(W; g_0, \beta_0)\right]$ falls in row space of $D_\beta\E_{\exp}[s(Z; \beta_0)]$. This will always be the case if $D_\beta\E_{\exp}[s(Z; \beta_0)]$ is invertible. This may not be true if the solution to the moment $0 = \E_{\exp}[s(Z;\beta)]$ is non-unique, and hence we must explicitly assume it.
\begin{ass}
\label{ass:row_space}
The Jacobian $D_\beta\E_{\obs}\left[m(W; g_0, \beta_0)\right]$ lies in the row space of $D_\beta \E_{\exp}[s(Z; \beta_0)]$.
\end{ass}

We note that this assumption is \textit{strictly weaker} than Assumption~\ref{ass:overlap}, which we confirm in the following lemma.

\begin{lemma}
\label{lem:strong-to-weak}
Recall the ``selection variables'' $\big\{Q_\nu : \nu \in \calJ \cup \{p\}\big\}$ defined above and suppose that Assumption~\ref{ass:overlap} holds and that $\P_{\exp}(Q_\nu = 1 \mid X) > 0$ and $v_{\exp, \nu}(X) > 0$ almost surely for each $\nu \in \calJ \cup \{p\}$ and that $\Var_{\exp}(B_1 = 1 \mid G = 0, X) >0$ and $\Var_{\exp}(B_t = 1 \mid B_{< t} = 0, G_a = 1, X) >0$. Then, Assumption~\ref{ass:row_space} holds.
\end{lemma}

\begin{proof}
Define, for any $\nu \in \calJ \cup \{p\}$, the vector $\xi_\nu$ as
\[
\xi_\nu := D_{\beta_\nu}\E_{\obs}\left[m(W; g_0, \beta_0)\right] = \E_{\obs}\left[\partial_\nu m(X) v_{\exp, \nu}(X) \phi_\nu(X) \right],
\]
where 
\begin{align*}
v_{\exp, p}(x) &:= p_0(x)\left\{1 - p_0(x)\right\} = \Var_{\exp}(B_1 = 1\mid G = 0, X = x) \\
v_{\exp, \tau}(x) &:= q_{0}(\tau, x)\left\{1 - q_0(\tau, x)\right\} = \Var_{\exp}(B_t \mid B_{< t} = 0, G_a = 1, X = x).
\end{align*}
and
\[
\partial_\nu m(X) := \begin{cases}
 -g_0(X)\sum_{\tau \in \calJ}\frac{r_0(\tau, X)}{q_0(\tau, X)} &\text{if } \nu = p \\
 \frac{g_0(X)p_0(X)r_0(\tau, X)}{q_0(\tau, X)^2} &\text{for } \tau \in \calJ.
\end{cases}
\]
We argue that, for each $\nu \in \calJ \cup \{p\}$, $\xi_\nu$ is orthogonal to the null space of the Hessian $H_\nu$ given by
\[
H_\nu := - D_{\beta}\E_{\exp}\left[s_\nu(Z; \beta_{0, \nu})\right] = \E_{\exp}\left[Q_\nu v_{\exp, \nu}(X) \phi_\nu(X)\phi_\nu(X)^\top\right].
\]
To see this, note that for any $\beta_\nu$ in the nullspace of $H_\nu$, we have 
\[
0 = \beta_\nu^\top H \beta_\nu = \E_{\exp}\left[Q_\nu v_{\exp, \nu}(X) \left(\phi_\nu(X)^\top \beta_\nu\right)^2\right] = \E_{\exp}\left[\P_{\exp}(Q_\nu = 1 \mid X)v_{\exp, \nu}(X)\left(\phi_\nu(X)^\top \beta_\nu\right)^2\right].
\]
Under the non-negativity assumptions made on the experimental conditional selection probabilities and experimental conditional variances, it follows that we must have $\phi_\nu(X)^\top \beta_\nu = 0$ almost surely under $\P_{\exp}$. Thus, we have that
\begin{align*}
\beta_\nu^\top \xi_\nu &= \E_{\obs}\left[\partial_\nu m(X) v_{\exp, \nu}(X)\phi_\nu(X)^\top \beta_\nu\right] \\ 
&= \E_{\exp}\left[\omega_0(X)\partial_\nu m(X) v_{\exp, \nu}(X) \phi_\nu(X)^\top \beta_\nu\right] \\
&= 0,
\end{align*}
which shows the result. 

\end{proof}

When Assumption~\ref{ass:row_space} holds, such a matrix $A_0$ exists. If $D_\beta \E_{\exp}[s(Z; \beta_0)]$ is invertible, $A_0$ will be unique. Otherwise, if said Hessian is singular, there will be infinitely many choices for $A_0$. We can always canonically pick $A_0$ to be defined in terms of the Moore-Penrose pseudo-inverse, i.e.,
\[
A_0 = - D_\beta\E_{\obs}\left[m(W; g_0, \beta_0)\right] \left(D_\beta \E[s(Z;\beta_0)]\right)^+.
\]

We can thus derive a Neyman orthogonal moment that contains both nonparametric debiasing terms (for the counterfactual spend function $g_0$) and parametric ones as well (for the conditional probabilities parameterized by $\beta$). In sum, we arrive at a score that can be written as
\[
M_{\corr}(g, \beta, \alpha_g, A) := \E_{\obs}\left[\psi_{\obs}(W; g, \beta, \alpha)\right] + \E_{\exp}\left[\psi_{\exp}(Z; \beta, A)\right],
\]
where we let $p(x)$ and $q(\tau, x)$ be naturally defined as above in terms of $\beta_p$ and $\beta_\tau$, and where we have set $\psi_{\exp}(Z; \beta, A) := A s(Z; \beta)$.
It is straightforward to check that $\Delta_0 = M_{\corr}(g_0, \beta_0, \alpha_{0, g}, A_0)$, as desired.

\begin{remark}
\label{rem:pooled-beta}
The preceding construction allows each  conditional hazard rate
$q_0(\tau,\cdot)$ to have its own parameter $\beta_\tau$. In our empirical application, due to limited data, we instead make a stronger assumption: we assume that while there are still type-specific feature embeddings $\phi_\tau(X)$ for each $\tau \in \calT$, there is a shared slope $\beta_{0, q}$ such that
\[
q_0(\tau,x)
=
\sigma\left(\phi_\tau(x)^\top\beta_{0, q}\right),
\qquad \tau\in\calT.
\]
In our application in Section~\ref{sec:case_study}, the $\tau$th embedding was taken to be $\phi_\tau(x) := (x, e_\tau)$, where $e_\tau \in \R^{12}$ with the $\tau$th standard basis vector. We note that we did not restrict the parameter $\beta_{0,p}$ indexing $p_0$ in our application.

Using notation akin to that in the preceding lemma, we define  
\[
\xi_\tau
:=
D_{\beta_\tau}\E_{\obs}\left[m(W; g_0, \beta_0, \alpha_{0, g})\right]
\quad\text{and}\quad
H_\tau
:=
-
D_{\beta_\tau}
\E_{\exp}\!\left[s_\tau(Z;\beta_{\tau,0})\right].
\]
Under the pooled model, the chain rule yields that the Jacobian of the base moment and the (negative) experimental score with respect to
$\beta_{q,0}$ are respectively:
\[
\xi_q
=
\sum_{\tau\in\calT}\xi_\tau,
\qquad
H_q
=
\sum_{\tau\in\calT}H_\tau.
\]
Consequently, it suffices to impose the aggregate row-space condition $\xi_q\in\operatorname{Range}(H_q)$
or, equivalently, $\operatorname{Null}(H_q)
\subseteq
\operatorname{Null}(\xi_q)$.
Writing
\[
s_q(Z;\beta_q)
:=
\sum_{\tau\in\calT}s_\tau(Z;\beta_q),
\]
a valid correction under the pooled conditional hazard model is $\psi_{\exp}^{\text{pool}}(Z; \beta, A) := A_p s_p(Z; \beta_p) + A_q s_q(Z;\beta_q)$, where $A_{q, 0} :=\xi_qH_q^+$, and again $H_q^+$ denotes the Moore--Penrose pseudoinverse.

The coordinate-wise construction in the preceding discussion is stronger:
it requires
\[
\xi_\tau\in\operatorname{Range}(H_\tau)
\qquad\text{for every }\tau\in\calT.
\]
When these coordinate-wise conditions hold, the resulting moment is
orthogonal with respect to arbitrary type-specific perturbations and
therefore also with respect to the common perturbations allowed by the
pooled model. However, the converse need not hold in general. 

In our empirical application, we leverage this pooled model, since under correct specification, setting $\wh{\beta}_q$ to be a zero for $\frac{1}{N}\sum_{j = 1}^N s_q(Z_j; \wh{\beta}_q) = 0$ yields a maximum likelihood estimate for $\beta_{0, q}$~\citep{allison1982discrete}.
\end{remark}

\section{Proofs of Asymptotic Linearity and Normality}
\label{app:estimation_proofs}
In this appendix, we prove the asymptotic linearity (and consequently, normality) of our estimator outlined in Algorithm~\ref{alg:dml}. We roughly follow the strategy leveraged in the proof of Lemma 1 in \citet{chernozhukov2023automatic}. In this work, the authors provide a general asymptotic linearity result for linear functionals of regression under covariate shifts. We cannot directly apply the aforementioned lemma for two reasons. First, our estimand is not a linear functional of a regression, which is the focus of the noted work. Further, since we assume the conditional probabilities are logistic regressions, and since we fit estimates on the entire experimental samples, we need to adjust the empirical process arguments from \citet{chernozhukov2023automatic}.

Before stating our theorem, we define some notation. First, we respectively define the observational and experimental empirical distributions as
\[
\P_n := \frac{1}{n}\sum_{i = 1}^n \delta_{W_i} \quad \text{and} \quad \Q_N := \frac{1}{N}\sum_{j = 1}^N \delta_{Z_j},
\]
where $\delta_w$ (and $\delta_z$) denotes the probability measure that puts all mass on $w$ (respectively, $z$). We further define the scaled empirical process as
\[
\G_n := \left(\P_n - \E_{\obs, W}\right) \quad \text{and} \quad \Hh_N := \left(\Q_N - \E_{\exp, Z}\right),
\]
where for a function $f$ and a random variable $U$ that is independent of $Z$, $\E_{\obs, W}[f(Z, U)]$ denotes the expectation just over $W$. We define $\E_{\exp, Z}$ similarly. 

In showing the convergence of our estimator, we assume that there exists a solution to the population moment
\begin{equation}
\label{eq:logistic_moment}
0 = \E_{\exp}[s_\nu(Z; \beta_\nu^\dagger)] \quad \text{for all } \nu \in \calJ \cup \{p\}.
\end{equation}
We do not assume that the solution is unique for any such $\nu$. We can, however, simply characterize the space of solutions and note that any two solutions yield precisely the same predictions and influence functions. The following lemma captures precisely this.

\begin{lemma}
\label{lem:logistic_solutions}
Define the set of all solutions of Equation~\eqref{eq:logistic_moment} to be
\[
\calB_\nu := \left\{ \beta_\nu \in \R^{d_\nu} : \E_{\exp}[s_\nu(Z; \beta_\nu)] = 0\right\}.
\]
Further, define the nullspace associated with the $\nu$th feature embedding $\phi_\nu(X)$ as
\[
\calN_\nu := \left\{v \in \R^{d_\nu} : Q_\nu\phi_\nu(X)^\top v = 0 \;\; \P_{\exp}-\text{almost surely}\right\}.
\]
Then, letting $\beta_\nu^\dagger \in \calB_\nu$ be arbitrary, we have the equivalence
\[
\calB_\nu = \beta_\nu^\dagger + \calN_\nu.
\]
\end{lemma}

\begin{proof}
First, we argue that $\beta_\nu^\dagger + \calN_\nu \subset \calB_\nu$. To see this, let $v \in \calN_\nu$ be arbitrary. We have 
\begin{align*}
\E_{\exp}\left[s_\nu(Z; \beta_\nu^\dagger + v)\right] &= \E_{\exp}\left[Q_\nu\left\{C_\nu - \sigma\big(\phi_\nu(X)^\top (\beta_\nu^\dagger + v)\big)\right\}\phi_\nu(X)\right] \\
&= \E_{\exp}\left[Q_\nu\left\{C_\nu - \sigma\big(\phi_\nu(X)^\top \beta_\nu^\dagger\big)\right\}\phi_\nu(X)\right] \\
&= \E_{\exp}\left[s_\nu(Z; \beta_\nu^\dagger)\right] = 0.
\end{align*}

Now, we prove the reverse inclusion. First, note that because the logistic function $\sigma(u)$ is strictly increasing and differentiable, we have by a first-order Taylor expansion that
\[
(\sigma(u) - \sigma(v))(u - v) = \sigma'(\xi)(u - v)^2 \geq 0 \quad \forall u, v \in \R,
\]
where $\xi \in [u, v]$ with equality if and only if $u = v$.
Thus, letting $\beta_\nu \in \calB_\nu$ be arbitrary, we have 
\begin{align*}
0 &= (\beta_\nu - \beta_\nu^\dagger)^\top\E_{\exp}\left[s_\nu(Z; \beta_\nu) - s_\nu(Z; \beta_\nu^\dagger)\right] \\
&= \E_{\exp}\left[Q_\nu\left\{\sigma\big(\phi_\nu(X)^\top \beta_\nu^\dagger\big) - \sigma\big(\phi_\nu(X)^\top\beta_\nu\big)\right\}\phi_\nu(X)^\top\left\{\beta_\nu^\dagger - \beta_\nu\right\}\right].
\end{align*}
Thus, for the difference of moments to be zero, we must have $Q_\nu\phi_\nu(X)^\top(\beta_\nu^\dagger - \beta_\nu) = 0$ $\P_{\exp}$-almost surely. In other words, we have $\beta_\nu \in \beta_\nu^\dagger + \calN_\nu$, so we have proved the result. 
\end{proof}

It is thus immediate from the above characterization that any two solutions to the noted moment yield precisely the same predictions. Furthermore, any derivatives of $s_\nu$ with respect to $\beta_\nu$ only depend on $\beta_\nu$ through the associated predictions. To ease notation, we let $\calB$ denote the space of solutions for the entire vector $\beta = (\beta_\nu : \nu \in \calJ \cup \{p\})$ where $\beta_\nu \in \calB_\nu$ for each $\nu$. Equivalently, we can write $\calB$ as
\[
\calB = \prod_{\nu \in \calJ \cup \{p\}} \calB_\nu \quad \text{or} \quad \calB = \beta^\dagger + \calN,
\]
where $\beta^\dagger$ is an arbitrary solution and $\calN = \bigoplus_{\nu \in \calJ \cup \{p\}}\calN_\nu$. 

Since the entire solution space is affine, there is a unique, minimal norm solution, which we denote as $\beta_0$. This solution is given by the projection of any solution $\beta \in \calB$ onto the orthogonal complement of the nullspace $\calN$, i.e., the set $\calN^\bot := \{w \in \R^{d} : v^\top w = 0 \;\forall v \in \calN\}$. Equivalently, the $\nu$th sub-vector of $\beta_0$ is the projection of $\beta_\nu$ onto $\calN_\nu^\bot$ for each $\nu \in \calJ \cup \{p\}$. We do not assume in the sequel that a learner's estimate $\wh{\beta}$ converges to some fixed solution in $\calB$. Instead, we just assume that $\wh{\beta}$ converges in Minkowski distance to the set $\calB$ (or, equivalently, that $\wh{\beta}_\nu$ converges to $\calB_\nu$). In other words, letting $\Pi_{\calN^\bot}$ denote the orthogonal projection onto $\calN^\bot$ ($\Pi_{\calN^{\bot}_\nu}$ the projection onto $\calN_\nu^\bot$), we assume that $\wh{\beta}^\bot := \Pi_{\calN^\bot}\wh{\beta}$ converges to $\beta_0$ (or $\wh{\beta}_\nu^\bot := \Pi_{\calN^{\bot}_\nu} \wh{\beta}_\nu$ converges to $\beta_{0,\nu}$). 

Lastly, we define the observational and experimental influence functions as 
\[
\varphi_{\obs}(W) := \psi_{\obs}(W; g_0, \alpha_{0, g}, \beta_0) - \Delta_0 \quad \text{and} \quad \varphi_{\exp}(Z) := \psi_{\exp}(Z; \beta_0, A_0).
\]

With this discussion and notation, we can now state our main result of this appendix. To ease notation and simplify analysis, we state the following result assuming that the nonparametric nuisances are trained on a sample that is independent of the observational samples $W_1, \dots, W_n$. In this case, the estimate $\wh{\Delta}$ can be written  as
\begin{equation}
\label{eq:est_sample_split}
\wh{\Delta} := \frac{1}{n}\sum_{i = 1}^n \psi_{\obs}(W_i; \wh{g}, \wh{\alpha}, \wh{\beta}) + \frac{1}{N}\sum_{j = 1}^N \psi_{\exp}(Z_j; \wh{\beta}, \wh{A}).
\end{equation}
We note that, although our theorem is focused on the sample splitting estimator, the asymptotic linearity (and hence, normality) result applies equally well to the cross-fit estimator outlined in Algorithm~\ref{alg:dml}.

\begin{theorem}
\label{thm:asymptotic_linearity}
Suppose that $W_1, \dots, W_n \sim_{i.i.d.} \P_{\obs}$ and $Z_1, \dots, Z_N \sim_{i.i.d.} \P_{\exp}$ are mutually independent, and suppose all (type-heterogeneous) Assumptions in Section~\ref{sec:model} hold. Further, suppose Assumption~\ref{ass:logistic} holds and that $T := \min\{n, N\} \rightarrow \infty$. Also suppose the following:
\begin{enumerate}
    \item \textbf{Null-Space Containment.} The experimental null-space $\calN_\nu$ of $\phi_\nu(X)$ is contained in the observational null-space, i.e.,
    \[
    v \in \calN_\nu\;\; \Longrightarrow \;\;\phi_\nu(X)^\top v = 0 \;\; \;\;\P_{\obs}-\text{almost surely.}
    \]
    \item \textbf{Nuisance Estimation Rates. }The nuisance estimates $\wh{\beta}_\nu$ for $\nu \in \calJ \cup \{p\}$, $\wh{A}$, $\wh{g}$, and $\wh{\alpha}$ are such that, for all $\nu, \tau \in \calJ \cup \{p\}$ the errors
    \begin{gather*}
    \|\wh{\beta}_\nu^\bot - \beta_{0, \nu}\|_2\times \|\wh{\beta}_\tau^\bot - \beta_{0, \tau}\|_2, \;\; \|\wh{\beta}_\nu^\bot - \beta_{0, \nu}\| \times \left\|\wh{A} - A_0\right\|_2,\\
    \left\|\wh{g} - g_0\right\|_{L^2(\P_{\obs})}\times \left\|\wh{\alpha} - \alpha_{0, g}\right\|_{L^2(\P_{\obs})}, \;\; \left\|\wh{g} - g_0\right\|_{L^2(\P_{\obs})} \times \|\wh{\beta}_\nu^\bot - \beta_{0, \nu}\|_2,
    \end{gather*}
    are all $o_\P(T^{-1/2})$, where we define $T := \min\{n, N\}$.
    \item \textbf{Observational Nuisance Independence.} The observational estimates $\wh{g}$ and $\wh{\alpha}$ are independent of $W_1, \dots, W_n$.
    \item \textbf{Nuisance Consistency.} All nuisances are consistent in $\ell_2$/$L^2$ norms, respectively, i.e.,
    \[
    \left\|\wh{g} - g_0\right\|_{L^2(\P_{\obs})}, \;\; \left\|\wh{\alpha} - \alpha_{0, g}\right\|_{L^2(\P_{\obs})} = o_\P(1),
    \]
    and for all $\nu \in \calJ \cup \{p\}$
    \[
    \|\wh{\beta}_\nu^\bot - \beta_{0, \nu}\|_2, \;\; \left\|\wh{A} - A_0\right\|_2 = o_\P(1).
    \]
    \item \textbf{Boundedness.} We assume that there is an absolute constant $M > 0$ such that all nuisances and nuisance estimates are bounded in $L^\infty$ or $\ell_2$ norm by $M$. Further, we assume that $\|Y\|_{L^\infty(\P_{\obs})} \leq M$ and $\|\phi_\nu\|_{L^\infty(\P_{b})} \leq M$ for $b \in \{\exp, \obs\}$ and $\nu \in \calJ \cup \{p\}$.
\end{enumerate}
Then, the estimate $\wh{\Delta}$ given in Equation~\eqref{eq:est_sample_split} is asymptotically linear, i.e., as $T \rightarrow \infty$ we have
\[
\wh{\Delta} - \Delta_0 = \frac{1}{n}\sum_{i = 1}^n\varphi_{\obs}(W_i) + \frac{1}{N}\sum_{j = 1}^N \varphi_{\exp}(Z_j) + o_\P\left(T^{-1/2}\right).
\]
Further, suppose $\sigma_{\obs}^2 := \Var_{\obs}\left[\psi_{\obs}(W; g_0, \alpha_{0, g}, \beta_0)\right] > 0$ and $\sigma_{\exp}^2 :=  \Var_{\exp}\left[\psi_{\exp}(Z; \beta_0, A_0)\right] > 0$, and that $\wh{\Sigma}$ satisfies $\frac{\wh{\Sigma}}{\Sigma_0} - 1 = o_\P(1)$, where $\Sigma_0 := \sqrt{\frac{\sigma_{\obs}^2}{n} + \frac{\sigma_{\exp}^2}{N}}$. Then, we have
\[
\wh{\Sigma}^{-1}\left(\wh{\Delta} - \Delta_0\right) \Rightarrow \calN(0, 1).
\]

\end{theorem}

\begin{proof}
To ease notation, we let $\wh{\eta} = (\wh{g}, \wh{\alpha})$.
Using a similar expansion to the one leveraged in the proof of Lemma 1 in the supplement of \citet{chernozhukov2023automatic}, we get the following decomposition:
\begin{align*}
\wh{\Delta} - \Delta_0 = \left(\P_n - \E_{\obs}\right) \psi_{\obs}(W; \eta_0, \beta_0) + \Q_N \psi_{\exp}(Z; \beta_0, A_0) + R_1 + R_2,
\end{align*}
where we have defined
\begin{align*}
R_{1, \obs} &:= \E_{\obs, W}\psi_{\obs}(W; \wh{\eta}, \wh{\beta})  - \E_{\obs}\psi_{\obs}(W; \eta_0, \beta_0) \\
R_{1, \exp} &:= \E_{\exp, Z}\psi_{\exp}(Z; \wh{\beta}, \wh{A}) - \E_{\exp}\psi_{\exp}(Z; \beta_0, A_0).\\
R_{2, \obs} &:= \G_n\left\{\psi_{\obs}(W; \wh{\eta}, \wh{\beta}) - \psi_{\obs}(W; \eta_0, \beta_0)\right\} \\
R_{2, \exp} &:= \Hh_{N}\left\{\psi_{\exp}(Z; \wh{\beta}, \wh{A}) - \psi_{\exp}(z; \beta_0,  A_0)\right\}
\end{align*}
and $R_b := R_{b, \obs} + R_{b, \exp}$ for $b \in \{1, 2\}$. Thus, to show asymptotic linearity, we just need to argue that $R_1, R_2 = o_\P(T^{-1/2})$

\paragraph{Controlling $R_1$.} Observe that, performing a second order functional Taylor expansion, we get
\begin{align*}
R_{1,\obs} &= D_\eta \E_{\obs, W}\psi_{\obs}(W; \eta_0, \beta_0)(\wh{\eta} - \eta_0) + D_\beta\E_{\obs, W}\psi_{\obs}(W; \eta_0, \beta_0)(\wh{\beta} - \beta_0)  \\
&+\frac{1}{2}D_{(\eta, \beta)}^2 \E_{\obs, W}\psi_{\obs}(W; \wt{\eta}, \wt{\beta})\big((\wh{\eta}, \wh{\beta}) - (\eta_0, \beta_0)\big),
\end{align*}
where $\wt{\eta} \in [\eta_0, \wh{\eta}]$ and $\wt{\beta} \in [\beta_0, \wh{\beta}]$.\footnote{Here, the inclusion in the interval means $\wt{\eta}$ is a convex combination of $\eta_0$ and $\wh{\eta}$.}
Likewise, a similar decomposition on $R_{1, \exp}$ yields
\begin{align*}
R_{1, \exp} &= D_{\beta}\E_{\exp, Z}\psi_{\exp}(Z; \beta_0, A_0)(\wh{\beta} - \beta_0) + D_{A}\E_{\exp, Z}\psi_{\exp}(Z; \beta_0, A_0)(\wh{A} - A_0)  \\
&+\frac{1}{2}D_{(\beta, A)}^2 \E_{\exp, Z}\psi_{\exp}(Z; \wb{\beta}, \wb{A})\big((\wh{\beta}, \wh{A}) - (\beta_0, A_0)\big),
\end{align*}
where now $\wb{\beta} \in [\beta_0, \wh{\beta}]$ and $\wb{A} \in [A_0, \wh{A}]$.

Now, observe that by Neyman orthogonality (shown in Appendix~\ref{sec:neyman_construction}), we know that:
\begin{gather*}
D_\eta \E_{\obs, W}\psi_{\obs}(W; \eta_0, \beta_0)(\wh{\eta} - \eta_0) = 0,\\
D_{A}\E_{\exp, Z}\psi_{\exp}(Z; \beta_0, A_0)(\wh{A} - A_0) = 0, \;\;\text{and} \\
D_\beta\E_{\obs, W}\psi_{\obs}(W; \eta_0, \beta_0)(\wh{\beta} - \beta_0) = - D_{\beta}\E_{\exp, Z}\psi_{\exp}(Z; \beta_0, A_0)(\wh{\beta} - \beta_0).
\end{gather*}
Further, we can bound second Gateaux derivatives as well. In particular, we observe that the following second/cross Gateaux derivatives are identically zero for the observational score:
\begin{gather*}
D_{g}^2 \E_{\obs, W}\psi_{\obs}(W; \wt{\eta}, \wt{\beta})(\wh{g} - g_0) = 0 \qquad D_\alpha^2 \E_{\obs, W}\psi_{\obs}(W; \wt{\eta}, \wt{\beta})(\wh{\alpha} - \alpha_{0, g}) = 0 \\
\forall \nu, \tau \in \calJ, \nu \neq \tau, \;\; D_{\beta_\tau, \beta_\nu} \E_{\obs, W}\psi_{\obs}(W; \wt{\eta}, \wt{\beta})(\wh{\beta}_\tau - \beta_{0, \tau}, \wh{\beta}_\nu - \beta_{0, \nu}) = 0.
\end{gather*}
Now, define $\wt{p}(x) := \sigma\big(\phi_p(X)^\top \wt{\beta}_p\big)$ and $\wt{q}(\tau, x) := \sigma\big(\phi_\tau(X)^\top \wt{\beta}_\tau\big)$ for each $\tau \in \calJ$. The observational Hessian blocks that do not vanish are given by:
\begin{align*}
&D_{\alpha, g}\E_{\obs, W}\psi_{\obs}(W; \wt{\eta}, \wt{\beta})(\wh{\eta} - \eta_0) = \E_{\obs, W}\left[D(\wh{\alpha} - \alpha_{0, g})(X)(\wh{g} - g_0)(X)\right] \\
&D_{\beta_p}^2\E_{\obs, W}\psi_{\obs}(W; \wt{\eta}, \wt{\beta})(\wh{\beta}_p - \beta_{p, 0}) \\
&\qquad = -\E_{\obs}\left[\wt{g}(X)\wt{p}(X)(1 - \wt{p}(X))(1 - 2\wt{p}(X))\left\{\sum_{\tau \in \calJ}\frac{S_\tau}{\wt{q}(\tau, X)}\right\}\left(\phi_p(X)^\top(\wh{\beta}_p - \beta_{p, 0})\right)^2\right] \\
&D_{\beta_p, \beta_\tau}\E_{\obs, W}\psi_{\obs}(W; \wt{\eta}, \wt{\beta})(\wh{\beta}_p - \beta_{p, 0}, \wh{\beta}_\tau - \beta_{0, \tau}) \\
&\qquad = \E_{\obs, W}\left[\wt{g}(X)\wt{p}(X)(1 - \wt{p}(X))S_\tau \frac{1 - \wt{q}(\tau, X)}{\wt{q}(\tau, X)}\left(\phi_p(X)^\top(\wh{\beta}_p - \beta_{p, 0})\right)\left(\phi_\tau(X)^\top(\wh{\beta}_\tau - \beta_{0, \tau})\right)\right] \\
&D_{\beta_\tau}^2 \E_{\obs, W}\psi_{\obs}(W; \wt{\eta}, \wt{\beta})(\wh{\beta}_\tau - \beta_{0, \tau}) \\
&\qquad = -\E_{\obs, W}\left[\wt{g}(X)\wt{p}(X)S_\tau\frac{1 - \wt{q}(\tau, X)}{\wt{q}(\tau, X)}\left(\phi_\tau(X)^\top(\wh{\beta}_\tau - \beta_{0, \tau})\right)^2\right] \\
&D_{\beta_p, g}\E_{\obs, W}\psi_{\obs}(W; \wt{\eta}, \wt{\beta})(\wh{\beta}_p - \beta_{p, 0}, \wh{g} - g_0) \\
&\qquad = - \E_{\obs, W}\left[(\wh{g} - g_0)(X)\wt{p}(X)\left\{1 - \wt{p}(X)\right\}\sum_{\tau \in \calJ}\frac{S_\tau}{\wt{q}(\tau, X)}\varphi_p(X)^\top(\wh{\beta}_p - \beta_{p, 0})\right] \\
&D_{\beta_\tau, g}\E_{\obs, W}\psi_{\obs}(W; \wt{\eta}, \wt{\beta})(\wh{\beta}_\tau - \beta_{\tau,0}, \wh{g} - g_0) \\
&\qquad = \E_{\obs, W}\left[(\wh{g} - g_0)(X)\wt{p}(X)S_\tau \frac{1 - \wt{q}(\tau, X)}{\wt{q}(\tau, X)}\phi_\tau(X)^\top(\wh{\beta}_\tau - \beta_{0, \tau})\right]
\end{align*}
Now leveraging the fact that (1) the coordinate embeddings $\phi_\nu(X)$ for $\nu \in \{p\} \cup \calJ$, nuisances, and nuisance estimates are assumed to be almost surely bounded and (2) the assumption that the nullspace of $\phi_\nu(X)$ contains $\calN_\nu$ for each $\nu \in \calJ \cup \{p\}$, we have that
\begin{align*}
\left|D_{\alpha, g}\E_{\obs, W}\psi_{\obs}(W; \wt{\eta}, \wt{\beta})(\wh{\eta} - \eta_0)\right| &\lesssim  \left\|(\wh{\alpha} - \alpha_{0, g})(\wh{g} - g_0)\right\|_{L^2(\P_{\obs})} \\
\left|D_{\beta_p}^2\E_{\obs, W}\psi_{\obs}(W; \wt{\eta}, \wt{\beta})(\wh{\beta}_p - \beta_{p, 0})\right| &\lesssim \|\wh{\beta}_p^\bot - \beta_{p, 0}\|_2^2\\
\left|D_{\beta_p, \beta_\tau}\E_{\obs, W}\psi_{\obs}(W; \wt{\eta}, \wt{\beta})(\wh{\beta}_p - \beta_{p, 0}, \wh{\beta}_\tau - \beta_{0, \tau})\right| &\lesssim  \|\wh{\beta}_p^\bot - \beta_{p, 0}\|_2\times \|\wh{\beta}_\tau - \beta_{0, \tau}\|_2\\
\left|D_{\beta_\tau}^2 \E_{\obs, W}\psi_{\obs}(W; \wt{\eta}, \wt{\beta})(\wh{\beta}_\tau - \beta_{0, \tau})\right| &\lesssim \|\wh{\beta}_\tau^\bot - \beta_{0, \tau}\|_2^2 \\
\left|D_{\beta_p, g}\E_{\obs, W}\psi_{\obs}(W; \wt{\eta}, \wt{\beta})(\wh{\beta}_p - \beta_{p, 0}, \wh{g} - g_0)\right| &\lesssim  \left\|\wh{g} - g_0\right\|_{L^2(\P_{\obs})} \left\|\wh{\beta}_p^\bot - \beta_{p, 0}\right\|_2 \\
|D_{\beta_\tau, g}\E_{\obs, W}\psi_{\obs}(W; \wt{\eta}, \wt{\beta})(\wh{\beta}_\tau - \beta_{\tau,0}, \wh{g} - g_0)| &\lesssim \left\|\wh{g} - g_0\right\|_{L^2(\P_{\obs})} \times \|\wh{\beta}_\tau^\bot - \beta_{0, \tau}\|_2.
\end{align*}

We can repeat the same exercise for the experimental score. In particular, the following blocks of the experimental second Gateaux derivative exactly vanish:
\begin{gather*}
D_A^2\E_{\exp, Z}\psi_{\exp}(Z; \wb{\beta}, \wb{A})(\wh{A} - A_0) = 0 \\
\forall \nu, \tau \in \calJ \cup \{p \},\;\;D_{\beta_\nu, \beta_\tau}\E_{\exp, Z}\psi_{\exp}(Z; \wb{\beta}, \wb{A})(\wh{\beta}_\nu - \beta_{0, \nu}, \wh{\beta}_\tau - \beta_{0, \tau}) = 0.
\end{gather*}

We now can compute the non-zero elements of the second Gateaux derivative. To simplify computation, note for any $A \in \R^{1 \times d}$ with $d = \sum_{\nu \in \calJ \cup \{p\}}d_\nu$ that we can write $A = \left(A_\nu : \nu \in \calJ \cup \{p\}\right)$, where $A_\nu \in \R^{1 \times d_\nu}$. Using this, we have
\[
A s(Z; \beta) = A_p s_p(Z; \beta_p) + \sum_{\tau \in \calJ}A_\tau s_\tau(Z; \beta_\tau) \;\; \forall \beta \in \R^d.
\]
Using this notation, and letting $\wb{q}(\nu, x) := \sigma(\phi_\nu(X)^\top \wt{\beta}_\nu)$ for any such $\nu$, we can express the second Gateaux derivatives as
\begin{align*}
&D_{\beta,A}\E_{\exp, Z}\psi_{\exp}(Z; \wb{\beta}, \wb{A})(\wh{\beta} - \beta_0, \wh{A} - A_0) \\
&\qquad = -\sum_{\nu \in \calJ \cup \{p\}}\E_{\exp, Z}\left[Q_\nu\wb{q}(\nu; X)(1 - \wb{q}(\nu;X))(\wh{A}_\nu - A_{\nu, 0})\phi_\nu(X)\phi_\nu(X)^\top(\wh{\beta}_\nu - \beta_{0, \nu})\right] \\
&D_{\beta_\nu}^2 \E_{\exp, Z}\psi_{\exp}(Z; \wb{\beta}, \wb{A})(\wh{\beta}_\nu - \beta_{0, \nu}) \\
&\qquad = -\E_{\exp, Z}\left[Q_\nu \wb{q}(\nu, X)\left\{1 - \wb{q}(\nu, X)\right\}\left\{1 - 2\wb{q}(\nu, X)\right\}\wb{A}_\nu\phi_\nu(X)\left\{\phi_\nu(X)^\top (\wh{\beta}_\nu - \beta_{0, \nu})\right\}^2\right].
\end{align*}
Leveraging again the (1) almost sure boundedness of feature embeddings, nuisances, and nuisance estimates and (2) the fact that $Q_\nu \phi_\nu(X)^\top \wh{\beta}_\nu = Q_\nu \phi_\nu(X)^\top \wh{\beta}_\nu^\top$ almost surely for all $\nu$, we obtain the following upper bounds.
\begin{align*}
\left|D_{\beta,A}\E_{\exp, Z}\psi_{\exp}(Z; \wb{\beta}, \wb{A})(\wh{\beta} - \beta_0, \wh{A} - A_0)\right| &\lesssim  \left\|\wh{A} - A_0\right\|_2\left\|\wh{\beta}^\bot - \beta_0\right\|_2 \\
\left|D_{\beta_\nu}^2 \E_{\exp, Z}\psi_{\exp}(Z; \wb{\beta}, \wb{A})(\wh{\beta}_\nu - \beta_{0, \nu})\right| &\lesssim \left\|\wh{\beta}_\nu^\bot - \beta_{0, \nu}\right\|_2^2,
\end{align*}
all of which are $o_\P(T^{-1/2})$ under our nuisance estimation rate assumptions. 

Thus, putting everything together and noting that $\|\wh{\beta}_\nu^\bot - \beta_{0, \nu}\|_2 \leq \|\wh{\beta}^\bot - \beta_0\|_2$ for any $\nu \in \calJ \cup \{p\}$, we have that 
\begin{align*}
\left|R_1\right| &\leq \left|D^2_{(\eta, \beta)}
\E_{\obs, W}\psi_{\obs}(W; \wt{\eta}, \wt{\beta})\big((\wh{\eta}, \wh{\beta}) - (\eta_0, \beta_0)\big)\right| +  \\
&\qquad \qquad + \left|D^2_{(A, \beta)}\E_{\exp, Z}\psi_{\exp}(Z; \wb{\beta}, \wb{A})\big((\wh{\beta}, \wh{A}) - (\beta_0, A_0)\big)\right| \\
&\lesssim \left\|\wh{\alpha} - \alpha_{0, g}\right\|_{L^2(\P_{\obs})} \times \left\|\wh{g} - g_0\right\|_{L^2(\P_{\obs})} + \left\|\wh{\beta}^\bot - \beta_0\right\|_2^2 \\
&\qquad \qquad + \left\|\wh{A} - A_0\right\|_2 \times \left\|\wh{\beta}^\bot - \beta_0\right\|_2 + \left\|\wh{g} - g_0\right\|_{L^2(\P_{\obs})} \times \left\|\wh{\beta}^\bot - \beta_0\right\|_2\\
&= o_\P(T^{-1/2}),
\end{align*}
where the final equality follows from our nuisance estimation rate assumption.

\paragraph{Controlling $R_2$.}

We now argue that $R_{2, \obs} = o_\P(n^{-1/2})$ and $R_{2, \exp} = o_\P(N^{-1/2})$, starting with the former term. Since the nuisance estimates $\wh{\beta}$ and $\wh{\eta} = (\wh{g}, \wh{\alpha})$ are assumed independent of the samples $W_1, \dots, W_n$, we can control $R_{2, \obs}$ with standard machinery. In particular, defining $\wh{\zeta} = (\wh{\eta}, \wh{\beta})$,  $\zeta_0 = (\eta_0, \beta_0)$, and $\Delta(W; \wh{\zeta}, \zeta_0) = \psi_{\obs}(W; \wh{\zeta}) - \psi_{\obs}(W; \zeta_0)$, if we can show that 
\begin{align*}
\Var_{\obs, W}\left[\sqrt{n}\G_n\Delta(W; \wh{\zeta}, \zeta_0)\right] = o_\P(1)
\end{align*}
then applying Chebyshev's/Markov's inequality conditionally in the same manner as \citet{chernozhukov2023automatic} will yield $\sqrt{n}\G_n \Delta(W; \wh{\zeta}, \zeta_0) = o_\P(1)$, i.e., $R_{2, \obs} = o_\P(n^{-1/2})$.

To see the desired convergence of the above conditional variance, let $\wb{\Delta} := \E_{\obs, W} \Delta(W; \wh{\zeta}, \zeta_0)$ and note that we have
\begin{align*}
\Var_{\obs, W}\left[\sqrt{n}\G_n \Delta(W; \wh{\zeta}, \zeta_0)\right] &= \E_{\obs, W}\left[\left(\sqrt{n}\G_n\Delta(W; \wh{\zeta}, \zeta_0)\right)^2\right] \\
&= \frac{1}{n} \E_{\obs, W}\left[\sum_{i, j = 1}^n \left\{\Delta(W_i; \wh{\zeta}, \zeta_0) - \wb{\Delta}\right\}\left\{\Delta(W_j; \wh{\zeta}, \zeta_0) - \wb{\Delta}\right\}\right] \\
&= \frac{1}{n} \E_{\obs, W}\left[\sum_{i = 1}^n \left\{\Delta(W; \wh{\zeta}, \zeta_0) - \wb{\Delta}\right\}^2\right] \\
&= \E_{\obs, W}\left[\left(\Delta(W; \wh{\zeta}, \zeta_0) - \wb{\Delta}\right)^2\right] \\
&\leq \E_{\obs, W}\left[\left(\psi_{\obs}(W; \wh{\zeta}) - \psi_{\obs}(W; \zeta_0)\right)^2\right].
\end{align*}

We can show that the last term in the previous display converges in probability to zero by leveraging nuisance consistency alongside our boundedness assumption. In more detail, the parallelogram inequality yields 
\begin{align*}
&\E_{\obs, W}\left[\left(\psi_{\obs}(W; \wh{\zeta}) - \psi_{\obs}(W; \zeta_0)\right)^2\right]  \\
&\qquad \lesssim \E_{\obs, W}\left[\left(\wh{g}(X)\wh{p}(X)\sum_{\tau \in \calJ}\frac{S_\tau}{\wh{q}(\tau, X)} - g_0(X)p_0(X)\sum_{\tau \in \calJ} \frac{S_\tau}{q_0(\tau, X)}\right)^2\right] \\
&\qquad \qquad + \E_{\obs, W}\left[\left(\wh{\alpha}(X)D\left\{Y - \wh{g}(X)\right\} - \alpha_{0, g}(X)D\left\{Y - g_0(X)\right\}\right)^2\right].
\end{align*}
First, note that $\ell_2$-consistency for $\wh{\beta}$ alongside the fact that $\phi_\nu(X)$ is $\P_{\obs}$-almost surely bounded by $M$ implies that the random functions $\wh{p}(x) := \sigma(\wh{\phi}_p(x)^\top \wh{\beta}_p)$ and $\wh{q}(\tau, x) := \sigma(\wh{\phi}_\tau(x)^\top \wh{\beta}_\tau)$ are (i) $L^2(\P_{\obs})$-consistent and (ii) are strictly bounded between zero and one. Since we have also assumed that all nonparametric nuisances, nonparametric nuisance estimates, and the outcomes $Y$ are also bounded by $M$, consistency of $\wh{\alpha}$ and $\wh{g}$ alongside the noted consistency of $\wh{p}$ and $\wh{q}(\tau, \cdot)$ implies that
\begin{gather*}
\left\|\wh{g}\wh{p}\sum_{\tau}\frac{S_\tau}{\wh{q}} - g_0p_0\sum_{\tau}\frac{S_\tau}{q_0}\right\|_{L^2(\P_{\obs})} = o_\P(1) \\
\Big\|\wh{\alpha}D\{Y - \wh{g}\} - \alpha_{0, g}D\{Y - g_0\}\Big\|_{L^2(\P_{\obs})} = o_\P(1),
\end{gather*}
which thus shows the desired convergence result for the variance.

We now argue that $R_{2, \exp} = o_\P(N^{-1/2})$. Because $\wh{\beta}$ is fit on the experimental samples, arguing that $\R_{2, \exp} = o_\P(N^{-1/2})$ takes a bit more care. First, observe that we can write 
\[
R_{2, \exp} = \underbrace{(\wh{A} - A_0)\Hh_Ns(Z; \wh{\beta})}_{=: T_1} + \underbrace{A_0 \Hh_N\left\{s(Z; \wh{\beta}) - s(Z; \beta_0)\right\}}_{=: T_2},
\]
which follows from adding and subtracting $A_0\Hh_N s(Z; \wh{\beta})$.
Next, observe that by Example 19.7 of \citet{van2000asymptotic}, because $\phi_\nu(X), \beta_\nu,$ and $\wh{\beta}_\nu$ are assumed bounded for each $\nu \in \calJ \cup \{p\}$ and because, for each such $\nu$,\footnote{This follows from the fact that $\sigma'(u) \leq \frac{1}{4}$ for all $u \in \R$.}
\[
\left\|s(z; \beta) - s(z; \beta')\right\|_2 \leq \sup_{t \in \R}\sigma'(t) \max_{\nu \in \calJ \cup \{p\}}\left\|\phi_\nu(x)\right\|_2\times \left\|\beta - \beta'\right\|_2 \leq \frac{1}{4}M\left\|\beta - \beta'\right\|_2,
\]
it follows that the class of functions
\[
\calS := \left\{s(\cdot; \beta) : \calZ \rightarrow \R^d : \left\|\beta\right\|_2 \leq M\right\}
\]
is $\P_{\exp}$-Donsker.
We now use this information to control the growth of $T_1$ and $T_2$.
First, we control $T_1$. Since the class $\calS$ is $\P_{\exp}$-Donsker and since we have assumed $\left\|\wh{A} - A_0\right\|_2 = o_\P(1)$, we have 
\begin{align*}
\sqrt{N}|T_1| &= \left|(\wh{A} - A_0)\sqrt{N}\Hh_Ns(Z; \wh{\beta})\right| \\
&\leq \left\|\wh{A} - A_0\right\|_2 \times \sup_{s \in \calS} \left\|\sqrt{N}\Hh_Ns(Z)\right\| \\
&= o_\P(1) \times O_\P(1) = o_\P(1),
\end{align*}
and hence $T_1 = o_\P(N^{-1/2})$. 

Next, note that the consistency assumption $\|\wh{\beta} - \beta_0\|_2 = o_\P(1)$ implies that $\E_{\exp, Z}\left[\left(s(Z; \wh{\beta}) - s(Z; \beta_0)\right)^2\right] = o_\P(1)$. Consequently, Lemma 19.24 of \citet{van2000asymptotic} implies that
\[
\sqrt{N}\Hh_N\left\{s(Z; \wh{\beta}) - s(Z; \beta_0)\right\} = o_\P(1).
\]
This in turn yields that
\begin{align*}
\sqrt{N}|T_2| &= \left|A_0 \sqrt{N}\Hh_N\left\{s(Z; \wh{\beta}) - s(Z; \beta_0)\right\}\right| \\
&\leq \left\|A_0\right\|_2 \left\|\sqrt{N}\Hh_N\left\{s(Z; \wh{\beta}) - s(Z; \beta_0)\right\}\right\|_2 \\
&= O(1) \cdot o_\P(1) = o_\P(1).
\end{align*}
Thus, we have shown that $R_{2, \exp} = o_\P(N^{-1/2})$, which completes the proof of asymptotic linearity.

\paragraph{Arguing Asymptotic Normality}

Now, to prove asymptotic normality, we leverage the Lindeberg-Feller central limit theorem for triangular arrays (Theorem 3.4.2 of \citep{durrett2019probability}). Let $(n_k)_{k \geq 1}$ and $(N_k)_{k \geq 1}$ be sequences of sample sizes such that the sequence $(T_k)_{k \geq 1}$ defined as $T_k := \min\{n_k, N_k\}$ satisfies $T_k \xrightarrow[k \rightarrow \infty]{} \infty$, and let $W_1^k, \dots, W_{n_k}^k \sim \P_{\obs}$ and $Z_1^{k}, \dots, Z_{N_k}^k \sim \P_{\exp}$ be random variables that are all mutually independent. Further, define $m_k := n_k + N_k$. For each $k \geq 1$, define the random variables $U_{k, 1}, \dots, U_{k, m_k}$ by
\[
U_{k, t} := \begin{cases}
\frac{1}{n_k}\Sigma_{k,0}^{-1}\varphi_{\obs}(W_t^{k})  &\text{for } t \in \{1, \dots, n_k\}, \\
\frac{1}{N_k}\Sigma_{0, k}^{-1}\varphi_{\exp}(Z^k_{t - n_k}) &\text{for } t \in \{n_k + 1, \dots, n_k + N_k\},
\end{cases}
\]
where we have defined the oracle standard error as 
\[
\Sigma_{0, k} := \sqrt{\frac{\sigma_{\obs}^2}{n_k} + \frac{\sigma_{\exp}^2}{N_k}}
\]
for each $k \geq 1$.

Clearly, the random variables $U_{k, 1}, \dots, U_{k, m_k}$ are all mutually independent for each fixed $k \geq 1$. Further, we have by construction that
\[
\sum_{t = 1}^{m_k}\Var\left[U_{k, t}\right] = \frac{\sigma_{\obs}^2}{n\Sigma_{0, k}^2} + \frac{\sigma_{\exp}^2}{N\Sigma_{0, k}^2} = 1.
\]
Lastly, under our boundedness assumption, it is straightforward to see that, for any $\epsilon > 0$, there is some sufficiently large $K_\epsilon$ such that $|U_{k, t}| \leq \epsilon$ for all $k \geq K_\epsilon$. Consequently, we have
\[
\lim_{k \rightarrow \infty} \sum_{t = 1}^{m_k}\E\left[|U_{k, t}|^2\1\big\{|U_{k, t}| > \epsilon\big\}\right] = 0,
\]
where here $\E[\cdots]$ denotes the expectation over the randomness in all samples. Thus, the Lindeberg condition is satisfied, and we have the convergence:
\[
\frac{1}{\Sigma_{0, k}}\left\{\wh{\Delta} - \Delta_0\right\} \Rightarrow \calN(0, 1).
\]
Lastly, we see that we have 
\begin{align*}
\frac{1}{\Sigma_{0, k}}\left\{\wh{\Delta} - \Delta_0\right\} &= \frac{1}{\wh{\Sigma}}\left(\frac{\wh{\Sigma}}{\Sigma_{0, k}} - 1\right)\left\{\wh{\Delta} - \Delta_0\right\}  + \frac{1}{\wh{\Sigma}}\left\{\wh{\Delta} - \Delta_0\right\} \\
&= \frac{1}{\wh{\Sigma}}\left\{\wh{\Delta} - \Delta_0\right\} + o_\P(1),
\end{align*}
and consequently $\frac{1}{\wh{\Sigma}}\left\{\wh{\Delta} - \Delta_0\right\} \Rightarrow \calN(0, 1)$, which proves asymptotic normality of the studentized estimate.
\end{proof}
\subsection{Asymptotic Linearity and Normality for Estimates of $\kappa_0$ and $\rho_0$.}

In the main body of the paper, we primarily report two scale-free estimands: the \textit{incrementality coefficient (IC)} $\kappa_0 := \Delta_0/\E_{\obs}[G A]$ and the \textit{incrementality ratio (IR)} $\rho_0 := \Delta_0/\E_{\obs}[G Y]$. We provide a corollary of Theorem~\ref{thm:asymptotic_linearity} that more or less leverages the delta method to show the asymptotic normality of the studentized estimates of these quantities.

Before we cite and prove the corollary, we provide some notation. First, for $r \in \{\kappa, \rho\}$, we define the influence functions
\begin{align*}
\varphi_{\obs, r}(W) &:= \frac{1}{\mu_0}\left\{\psi_{\obs}(W; \eta_0, \beta_0) - \Delta_0\right\} - \frac{\Delta_0}{\mu_0^2}\left\{H_r(W) - \mu_0\right\} \\
\varphi_{\exp, r}(Z) &:= \frac{1}{\mu_0} \psi_{\exp}(Z; \beta_0, A_0),
\end{align*}
where we recall that $H_\kappa(W) := S \cdot A$ and $H_\rho(W) := S \cdot Y$ and have set $\mu_0 := \E_{\obs}[H_r(W)]$. Further, we define the observational and experimental standard errors respectively as 
\[
\Sigma_{r, \obs} := \sqrt{\frac{\Var[\varphi_{\obs, r}(W)]}{n}} \quad \text{and} \quad \Sigma_{r, \exp} := \sqrt{\frac{\Var[\varphi_{\exp, r}(Z)]}{N}},
\]
and further define the joint standard error as
\[
\Sigma_{r} := \sqrt{\frac{\Var[\varphi_{\obs, r}(W)]}{n} + \frac{\Var[\varphi_{\exp, r}(Z)]}{N}} \equiv \sqrt{\Sigma_{r, \obs}^2 + \Sigma_{r, \exp}^2}.
\]
Lastly, we define the estimate for $r_0 \in \{\kappa_0, \rho_0\}$ as $\wh{r} = \frac{\wh{\Delta}}{\P_n H_r(W)}$, where $\wh{\Delta}$ is as defined in Algorithm~\ref{alg:dml}.

\begin{corollary}
\label{cor:ic_ir_convergence}
Let $r \in \{\kappa, \rho\}$, and suppose the identifying assumptions in Section~\ref{sec:model} hold. Then, the estimator $\wh{r}$ outlined above is asymptotically linear, taking the form
\[
\wh{r} - r_0 = \frac{1}{n}\sum_{i = 1}^n \varphi_{\obs, r}(W_i) + \frac{1}{N}\sum_{j = 1}^N \varphi_{\exp, r}(Z_j) + o_\P(T^{-1/2}),
\]
where again $T := \min\{n, N\}$.
Consequently, if $\wh{\Sigma}_r$ is any estimate of $\Sigma_r$ satisfying the consistency relation
\[
\frac{\wh{\Sigma}_r}{\Sigma_r} - 1 = o_\P(1),
\]
then we have asymptotic normality of the studentized statistic:
\[
\frac{\wh{r} - r_0}{\wh{\Sigma}_r} \Rightarrow \calN(0, 1).
\]
\end{corollary}

\begin{proof}
Define $\wh{\mu} := \P_n H_r(W)$, and observe that the following hold:
\begin{gather*}
\wh{\mu} - \mu_0 = \P_n\left\{H_r(W) - \mu_0\right\} + o_\P(n^{-1/2}) \\
\wh{\Delta} - \Delta_0 = \P_n\left\{\psi_{\obs}(W; \eta_0, \beta_0)- \Delta_0\right\} + \Q_N\psi_{\exp} + o_\P(T^{-1/2}). 
\end{gather*}
Note that $\wh{r} = f(\wh{\Delta}, \wh{\mu})$ and $r_0 = f(\Delta_0, \mu_0)$ where $f(\Delta, \mu) = \frac{\Delta}{\mu}$. Further, we have $\nabla f(\Delta, \mu) = (1/\mu, -\Delta/\mu^2)$. Thus, a second order Taylor expansion with mean value theorem remainder yields
\begin{align*}
\wh{r} - r_0 &= \frac{1}{\mu_0}\left\{\wh{\Delta} - \Delta_0\right\} - \frac{\Delta_0}{\mu_0^2}\left\{\wh{\mu} - \mu_0\right\} + \frac{1}{2}\begin{pmatrix} \wh{\Delta} - \Delta_0 \\ \wh{\mu} - \mu_0\end{pmatrix}^\top \nabla^2 f(\wt{\Delta}, \wt{\mu})\begin{pmatrix} \wh{\Delta} - \Delta_0 \\ \wh{\mu} - \mu_0\end{pmatrix},
\end{align*}
for some $\wt{\Delta} \in [\Delta_0, \wh{\Delta}]$ and $\wt{\mu} \in [\mu_0, \wh{\mu}]$. Here, the Hessian is given by 
\[
\nabla^2 f(\Delta, \mu) = \begin{pmatrix}
0 & -\mu^{-2} \\
-\mu^{-2} & 2\Delta\mu^{-3}
\end{pmatrix}.
\]
Since $\mu_0 > 0$ and $\wh{\mu} \xrightarrow[]{\P} \mu_0$, it thus follows that 
\begin{align*}
\wh{r} - r_0 &= \frac{1}{\mu_0}\left\{\wh{\Delta} - \Delta_0\right\} - \frac{\Delta_0}{\mu_0^2}\left\{\wh{\mu} - \mu_0\right\} + O_\P\left(|\wh{\Delta} - \Delta_0||\wh{\mu} - \mu_0| + |\wh{\mu} - \mu_0|^2\right) \\
&= \P_n \varphi_{\obs, r}(W) + \Q_N \varphi_{\exp, r}(Z) + o_\P(T^{-1/2}),
\end{align*}
which follows by the definitions of $\varphi_{\obs, r}$, $\varphi_{\exp, r}$ and because $|\wh{\mu} - \mu_0| = O_\P(n^{-1/2})$ and $|\wh{\Delta} - \Delta_0| = O_\P(T^{-1/2})$. Hence, asymptotic linearity holds.  

Now, to obtain asymptotic normality, note that the central limit theorem implies that
\[
U_n := \frac{1}{\Sigma_{\obs, r}}\P_n \varphi_{\obs, r}(W) \Rightarrow U \quad \text{and} \quad V_N := \frac{1}{\Sigma_{\exp, r}}\Q_N \varphi_{\exp, r}(Z) \Rightarrow V,
\]
where $U, V \sim \calN(0, 1)$ are standard normal random variables. Further, because the observational samples $W_1, \dots, W_n$ and the experimental samples $Z_1, \dots, Z_N$ are independent, it actually follows that $(U_n, V_N) \Rightarrow \calN(0, I_2)$. Consequently, we have
\[
\frac{\P_n \varphi_{\obs, r}(W) + \Q_N \varphi_{\exp, r}(Z)}{\Sigma_r} = \frac{\Sigma_{r, \obs} U_n + \Sigma_{r, \exp}V_N}{\Sigma_r} \Rightarrow \calN(0, 1).
\]
Asymptotic normality of the studentized statistic now follows from the assumption that the standard error estimate $\wh{\Sigma}_r$ is consistent in the noted sense.

\end{proof}

\subsection{Inference for the Average Monthly Estimand}
\label{sec:neyman_construction:time_avg}

In our empirical results, we estimate and report the target parameter $r\in \{ \Delta, \kappa, \rho\}$ separately for each target period (in this case, month) $t \in \calT := \{1, \dots, t_{\max}\}$ and also report our estimates for the time-averaged estimand:
\[
\bar{r}_0 := \frac{1}{T_{\max}} \sum_{t=1}^{T_{\max}} r_{0,t},
\]
where $r_{0,t}$ denotes the true value of the parameter for month $t$. The natural estimator is $\wh{\bar{r}} := \frac{1}{T_{\max}} \sum_{t=1}^{T_{\max}} \wh{r}_t$, where each $\wh{r}_t$ is produced by Algorithm~\ref{alg:dml} applied to the period $t$ data. In this subsection, we derive the influence function of $\wh{\bar{r}}$ and present its asymptotic normality.

\paragraph{Data structure and notation.}
Throughout, we operate under the assumption that the observational units are independent across periods, whereas the experimental units are the same. This, roughly, reflects the actual setting of our application: there is a single experiment whose data is shared across months, and the units considered within any month are subsampled from a much larger overpopulation, and hence should be (roughly) independent. In the sequel, for each $t \in \calT$, we let $o_t$ denote the set of observational sample unit indices for period $t$, $n_t := |o_t|$, let $o := \bigcup_{t=1}^{T_{\max}} o_t$, and set $n := |o| = n_1 + \cdots + n_{t_{\max}}$. For $t \in \calT$ and $i \in o_t$, we let $W_{i, t}$ denote the partial observation for unit $i$ in period $t$. Further, we let $T_t := \min\{n_t, N\}$ (where again $N$ denotes the experimental sample size), we let $\varphi_{r,\obs}^{(t)}$ denote the period $t$ influence functions, and let $\varphi_{r, \exp}^{(t)}$ denote the experimental influence function.\footnote{Note that although we leverage the same sample $Z_1, \dots, Z_N$ for each period $t \in \calT$, for $r \in \{\rho, \kappa\}$, the experimental influence function naturally depends on the distribution of $W_t$. Hence, we must index the experimental influence function by the period.} Further, we set $T_{\min} := \min_{t \in \calT} T_t$ to be the minimum sample size throughout the $t$ periods. 

\paragraph{Asymptotic linearity of the average.}
From Corollary~\ref{cor:ic_ir_convergence}, the month-specific estimator satisfies:
\[
\wh{r}_t - r_{0,t} = \frac{1}{n_t}\sum_{i \in o_t} \varphi_{r,\obs}^{(t)}(W_{i,t}) + \frac{1}{N}\sum_{j = 1}^N \varphi_{r,\exp}^{(t)}(Z_{j}) + o_\P(T_t^{-1/2}).
\]
Using this expansion, we can reason about the asymptotic linearity of the aggregated estimator. In particular, we have
\begin{align*}
\wh{\bar{r}} - \bar{r}_0 &= \frac{1}{t_{\max}} \sum_{t=1}^{t_{\max}} \left(\frac{1}{n_t}\sum_{i \in o_t} \varphi_{r, \obs}^{(t)}(W_{i,t}) + \frac{1}{N}\sum_{j = 1}^N \varphi_{r, \exp}^{(t)}(Z_j)\right) + o_\P(T_{\min}^{-1/2}) \\
&= \frac{1}{t_{\max}}\sum_{t = 1}^{t_{\max}}\left(\frac{1}{n_t}\sum_{i \in o_t} \varphi_{r, \obs}^{(t)}(W_{i, t}) \right) + \frac{1}{N}\sum_{j = 1}^N \frac{1}{t_{\max}}\sum_{t = 1}^{t_{\max}}\varphi_{r, \exp}^{(t)}(Z_j) + o_\P(T_{\min}^{-1/2})\\
&= \frac{1}{n}\sum_{i \in o} \frac{n}{t_{\max}n_{t(i)}}\varphi_{r, \obs}^{(t(i))}(W_{i, t(i)}) + \frac{1}{N}\sum_{j = 1}^N \wb{\varphi}_{r, \exp}(Z_j)  + o_\P(T_{\min}^{-1/2}),
\end{align*}
where $t(i)$ denotes the unique period unit $i$ belongs to and we have defined the average experimental influence function $\wb{\varphi}_{r, \exp}$ as
\[
\wb{\varphi}_{r, \exp}(Z) := \frac{1}{t_{\max}}\sum_{t = 1}^{t_{\max}}\varphi_{r, \exp}^{(t)}(Z).
\]
Note that all observational units associated with period $t$ have precisely the same influence function, which is their original influence function appropriately reweighted to adjust for relative prevalence of period $t$ samples amongst the total aggregated population of $n$ units.

\paragraph{Standard error and asymptotic normality.}
Since all units are assumed to be independent of one another, and since observational units in the same period share the same influence function, the variance of the leading term can be computed as follows:
\begin{align*}
\Sigma_{\bar{r}}^2 &= \frac{1}{t_{\max}^2}\sum_{t = 1}^{t_{\max}}  \frac{1}{n_t}\Var_{\obs, t}\left[\varphi_{r, \obs}^{(t)}(W_t)\right] + \frac{1}{N}\Var_{\exp}\left[\wb{\varphi}_{r, \exp}(Z)\right],
\end{align*}
where $\Var_{\obs, t}$ denotes the distribution over the observations $W_t$ in period $t$. Let $\wh{\Var}_{\obs, t}$ denote the recommended estimator for the observational variance from Algorithm~\ref{alg:dml} for period $t$, and let $\wh{\Var}_{\exp}$ be defined as
\[
\wh{\Var}_{\exp} := \frac{1}{N}\sum_{j = 1}^N\left(\wh{\bar{\varphi}}_{r, \exp}(Z_j) - \frac{1}{N}\sum_{j = 1}^N \wh{\bar{\varphi}}_{r, \exp}(Z_j)\right)^2.
\]
Here, we have defined the plug-in influence function $\wh{\bar{\varphi}}_{r, \exp}$ naturally as
\[
\wh{\bar{\varphi}}_{r, \exp}(Z) := \begin{cases}
\frac{1}{t_{\max}}\sum_{t = 1}^{t_{\max}}\wh{A}^{(t)}s(Z; \wh{\beta}) &\text{if } r = \Delta \\
\frac{1}{t_{\max}}\sum_{t = 1}^{t_{\max}}\frac{\wh{A}^{(t)}s(Z; \wh{\beta})}{\wh{\mu}_r^{(t)}} &\text{if } r \in \{\kappa, \rho\},
\end{cases}
\]
where $\wh{A}^{(t)}$ and $\wh{\mu}_r^{(t)}$ denote estimates associated with period $t$.
The natural plug-in estimator for the standard error of the aggregated estimate is given by
\begin{align}
\label{eq:agg_std_error}
\wh{\Sigma}_{\bar{r}} &= \sqrt{\frac{1}{t_{\max}^2}\sum_{t = 1}^{t_{\max}}\frac{1}{n_t}\wh{\Var}_{\obs, t} + \frac{1}{N}\wh{\Var}_{\exp}}.
\end{align}

The convergence of the studentized statistic is guaranteed by Corollary~\ref{cor:avg_monthly} below. The proof follows readily using the same techniques as in the proofs of Theorem~\ref{thm:asymptotic_linearity} and Corollary~\ref{cor:ic_ir_convergence} above, and thus we do not prove the corollary.

\begin{corollary}[Asymptotic Normality of the Average Monthly Estimand]
\label{cor:avg_monthly}
Suppose the conditions of Theorems~\ref{thm:asymptotic_linearity} and Corollary~\ref{cor:ic_ir_convergence} hold when applied to each period $t \in \calT$. Further, suppose that $T_{\min} \rightarrow \infty$. Then, the average monthly estimator $\wh{\bar{r}}$ satisfies $\wh{\Sigma}_{\bar{r}}^{-1}(\wh{\bar{r}} - \bar{r}_0) \Rightarrow \calN(0, 1)$.
\end{corollary}

\section{Relation to Data Fusion for Long-Term Effects via Surrogates}
\label{app:surrogates}

Our data fusion identification result shares conceptual similarities with the literature on long-term effects via surrogates~\citep{athey2025surrogate, prentice1989surrogate, price2018estimation}. Such settings are characterized by having access to an observational and an experimental dataset governed by the causal graphs depicted in Figure~\ref{fig:data_environments_surrogates}.
\begin{figure}[ht]
\centering
\begin{subfigure}[b]{0.45\textwidth}
\centering
\begin{tikzpicture}[node distance=1cm and 1.5cm, line width=0.7pt]
    \node[obs] (X) at (0, 2.5) {$X$};
    \node[lat] (G) at (-1.8, 0) {$G$};
    \node[obs] (B) at (0, 0) {$B$};
    \node[obs] (Y) at (1.8, 0) {$Y$};
    \draw[arr] (X) -- (G);
    \draw[arr] (X) -- (B);
    \draw[arr] (X) -- (Y);
    \draw[arr] (G) -- (B);
    \draw[arr] (B) -- (Y);
\end{tikzpicture}
\caption{Observational data containing short-term surrogate outcomes $B$ and long term outcome $Y$.}
\label{fig:dag_obs}
\end{subfigure}
\hfill
\begin{subfigure}[b]{0.45\textwidth}
\centering
\begin{tikzpicture}[node distance=1cm and 1.5cm, line width=0.7pt]
    \node[obs] (X) at (0, 2.5) {$X$};
    \node[obs] (G) at (-1.8, 0) {$G$};
    \node[obs] (B) at (0, 0) {$B$};
    \node[lat] (Y) at (1.8, 0) {$Y$};
    \draw[arr] (X) -- (G);
    \draw[arr] (X) -- (B);
    \draw[arr] (X) -- (Y);
    \draw[arr] (G) -- (B);
    \draw[arr] (B) -- (Y);
\end{tikzpicture}
\caption{Experimental data containing treatment $G$ and short-term surrogate outcomes $B$.}
\label{fig:dag_exp}
\end{subfigure}
\caption{Causal directed acyclic graphs (DAGs) assumed by the surrogate outcome framework. In the surrogate setting, the intermediate variable $B$ is fully observed in both datasets and serves as a mediator between treatment $G$ and long-term outcome $Y$. In our setting, $B$ is observed in both datasets but we do not have $Y \independent G \mid (X, B)$. Further, the treatment $G$ is only partially observed in the observational data.}
\label{fig:data_environments_surrogates}
\end{figure}
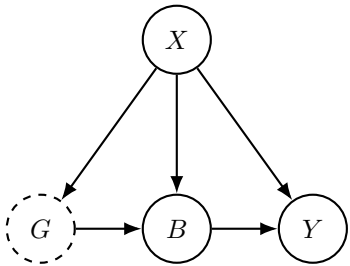
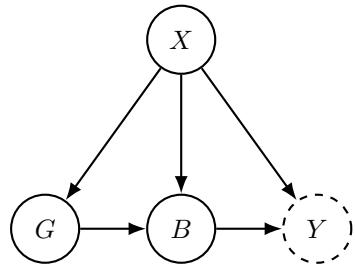
In both settings, experimental data is combined with observational data to identify a causal effect that neither dataset could identify alone. One might intuitively view the booking decision $B$ as a short-term ``surrogate'' and spend $Y$ as a long-term outcome. However, our setting departs from the standard surrogate framework in several crucial ways. First, $B$ is not a complete mediator (the Prentice criterion \cite{prentice1989surrogate} fails): the gift card $G$ directly affects the spend amount $Y$ even conditional on booking. In standard surrogate settings, identification typically fails if the surrogate is incomplete. We achieve identification despite this incompleteness because we have partial observability of the treatment in the observational data via the indicator $S = B \cdot G$, and because booking is a strict prerequisite for spending ($Y=0$ when $B=0$). 
Further, we are interested in understanding the effect of gift card receipt $G$ on the outcome $Y$ in the \textit{observational population} as opposed to the experimental one. Thus, our result provides a distinct identification pathway for settings with censored treatments, expanding the empirical paradigm of data fusion beyond traditional surrogate applications.

\section{Further Details on the Airbnb Case Study}
\label{app:case_study}

In this appendix, we provide additional details regarding the application of our methodology to data at Airbnb. In Section~\ref{app:case_study:estimation}, we provide implementation details for the estimators for the IR and IC.  In Section~\ref{app:case_study:features}, we describe the various features leveraged in estimation. Lastly, in Section~\ref{app:case_study:sensitivity}, we conduct a sensitivity analysis for our findings.

\subsection{Estimation Details}
\label{app:case_study:estimation}

\begin{table}[H]
\centering
\small
\setlength{\tabcolsep}{4pt}
\begin{tabular}{lrrrr}
\toprule
\textbf{Month} & \textbf{$D=1$} & \textbf{$S=1$} & \textbf{$B=0$} & \textbf{Total} \\
\midrule
2024-01 & 3,533,578 & 38,941 & 3,063,986 & 6,636,505 \\
2024-02 & 3,441,571 & 32,838 & 3,076,209 & 6,550,618 \\
2024-03 & 3,754,586 & 33,131 & 3,062,708 & 6,850,425 \\
2024-04 & 3,600,800 & 29,953 & 3,075,835 & 6,706,588 \\
2024-05 & 3,688,202 & 30,865 & 3,075,055 & 6,794,122 \\
2024-06 & 3,652,893 & 29,672 & 3,080,809 & 6,763,374 \\
2024-07 & 3,869,517 & 29,246 & 3,073,336 & 6,972,099 \\
2024-08 & 3,721,814 & 29,173 & 3,083,947 & 6,834,934 \\
2024-09 & 3,492,342 & 25,260 & 3,096,820 & 6,614,422 \\
2024-10 & 3,506,309 & 24,769 & 3,099,847 & 6,630,925 \\
2024-11 & 3,216,049 & 22,600 & 3,116,727 & 6,355,376 \\
2024-12 & 3,198,623 & 26,827 & 3,119,746 & 6,345,196 \\
\bottomrule
\end{tabular}
\caption{Sample sizes for the observational dataset by month and customer
behavior. $D=1$ denotes booking without a gift card, $S=1$ denotes booking
with a gift card across all distribution channels, and $B=0$ denotes no
booking. The counts represent subsampled populations, not raw
superpopulation sizes.}
\label{tab:obs_monthly_counts}
\end{table}

To apply our method, we use the semiparametric cross-fitting estimator described in Algorithm~\ref{alg:dml} in Section~\ref{sec:estimation} with $K = 5$ folds. We compute this estimator separately in each month and each channel, and perform separate nuisance model selection in each segment as well, which we detail below. We leverage stratified splitting in each segment on the observational samples in any month to ensure that the number of non-booking units ($B = 0$), units that book without a gift card ($D = 1$), and units that book with a gift card ($S = 1$) are approximately the same across each fold. We then aggregate the monthly estimates per channel to get a single number representing the incrementality of each channel. We describe the precise details on how monthly effects are combined into yearly ones by appropriately averaging influence functions in Appendix~\ref{sec:neyman_construction:time_avg}.

We fit all experimental nuisances once using the \texttt{LogisticRegression} class from the \texttt{sklearn} Python package~\citep{pedregosa2011scikit} without any regularization.
We simply take the feature embedding as $\phi_\nu(X) \equiv (1, X)$ for each $\nu \in \calT \cup \{p\}$. To fit an estimate $\wh{p}(X)$ of $p_0(X)$, we regress $B_1$ onto $X$ amongst ungifted units (i.e., those with $G = 0$). Given limited experimental samples, we strengthen Assumption~\ref{ass:logistic} to be that $q_0(\tau, x) = \sigma\big(X^\top \beta_0 + \gamma_{\tau, 0}\big)$ for each $\tau \in \calT$, i.e., that there is a shared slope vector for the hazards but month-specific intercepts. To fit estimates $\wh{q}(\tau, X) = \sigma\big(X^\top \wh{\beta} + \wh{\gamma}_\tau\big)$ of $q_0(\tau, X)$, we fit a joint logistic regression model with shared slope estimate $\wh{\beta}$ and per-month intercepts $\wh{\gamma}_\tau$ for $\tau \in \calT$ using the method of \citet{allison1982discrete}. Briefly, this method works by including $t$ entries for a customer who books for the first time in month $t \in \calT$, and augments the base features $X$ with ``dummy'' indicators attached to each month associated with the booker up to the month in which they booked.

We perform model selection for the nonparametric observational nuisances $\ell_0(x), g_0(x),$ and $r_0(x)$ using the \texttt{AutoML} class from the \texttt{flaml} library~\citep{wang2021flaml}. We select the $g_0$ model once per month amongst all units satisfying $D  = B(1 - G) = 1$ (i.e., those who book without a gift card). We select the models $r_0$ and $\ell_0$ once per each channel and month pair using all relevant data (i.e., all non-bookers and bookers without gift card for the month, plus all units who booked with a gift card from the desired channel).  We have the \texttt{AutoML} class select between the classes \texttt{LGBMClassifier} and \texttt{XGBClassifier} in the case of $r_0$ and $\ell_0$, and \texttt{LGBMRegressor} and \texttt{XGBRegressor} for $g_0$, all using the default hyperparameter grid associated with the classes in flaml~\citep{ke2017lightgbm, chen2016xgboost}. We allow \texttt{AutoML} ten minutes to select parameters for each model. Models with the selected hyperparameters are then fit using $K$-fold cross-fitting as noted above. We note that using the entire dataset to select just the model is justified theoretically by~\citep{van2006oracle} and practically by \citep{bach2024hyperparameter}.

Lastly, we note the exact sample sizes leveraged in estimation for each month in Table~\ref{tab:obs_monthly_counts} above. We again emphasize that the numbers displayed in the above table represent counts computed \textit{after subsampling}. We do not directly reveal raw counts nor subsampling rates in the paper.

\subsection{Monthly Estimates of Incrementality}
\label{app:case_study:monthly}

In the main empirical section of the paper, we presented and interpreted yearly estimates for the IR and IC across the four channels of interest. 
We computed these estimates by first producing monthly estimates of incrementality in each channel using the estimator outlined in Algorithm~\ref{alg:dml} and the surrounding text. Then, we aggregated these monthly estimates into yearly ones leveraging the machinery from Appendix~\ref{sec:neyman_construction:time_avg}. 
Our focus in this appendix is to interpret the constructed time series of incrementality estimates, both when all units are considered (Figure~\ref{fig:au_main}) as well as when attention is restricted to suspected self-gifters (Figure~\ref{fig:sg_main}).

\begin{figure}[htb]
\centering
\includegraphics[width=\textwidth]{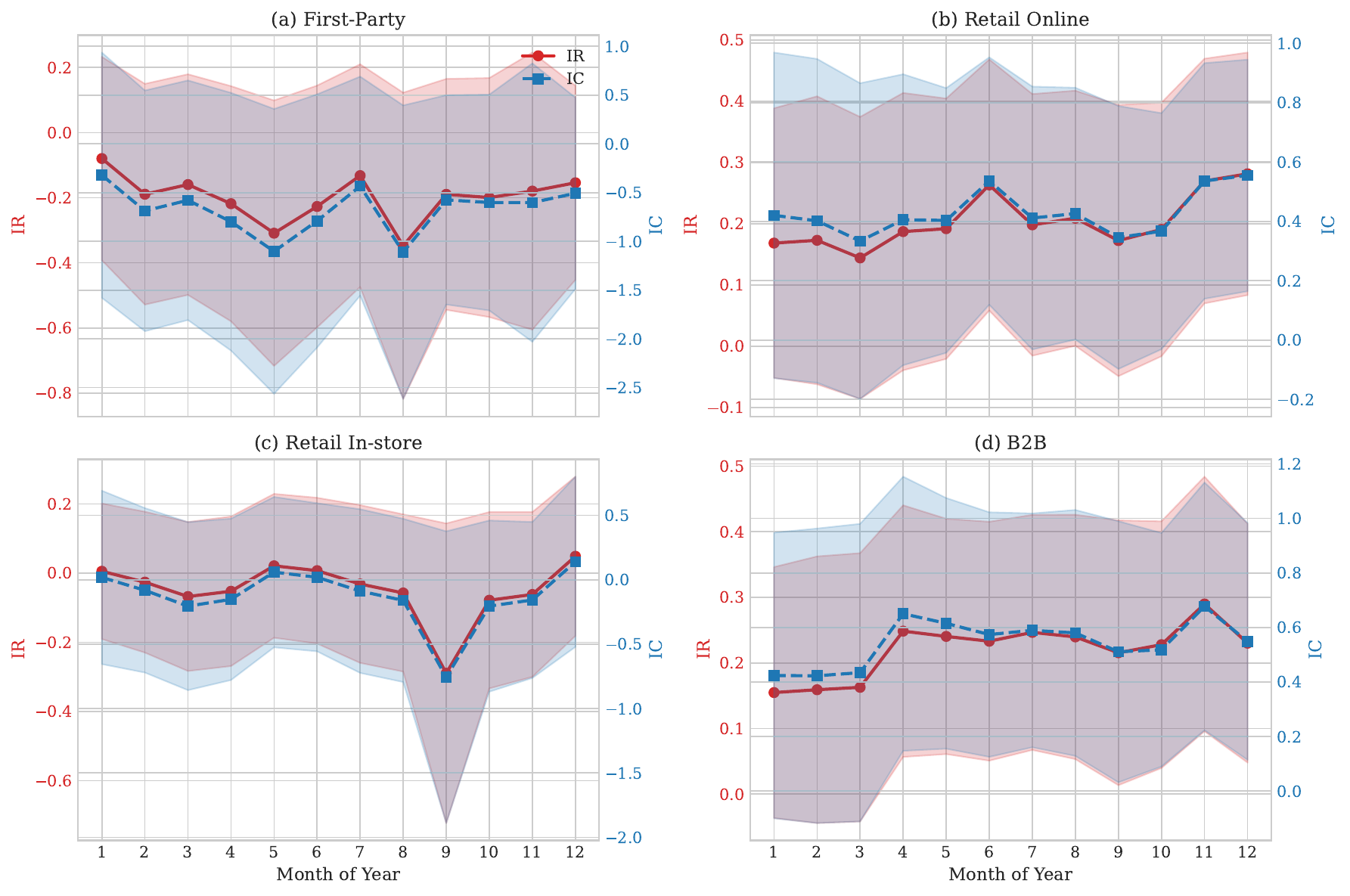}
\caption{Monthly estimates and marginally valid 95\% confidence intervals for the IR and IC broken down by channel.}
\label{fig:au_main}
\end{figure}

We first interpret the numbers shown in Figure~\ref{fig:au_main}. We note that our yearly findings for both First-Party and Retail In-Store from Section~\ref{sec:case_study:main_results} are reflected in the monthly incrementality estimates as well: there does not exist a month where point estimates for the IC or IR in First-Party are positive. While the point estimates for incrementality can be mildly positive in some months for Retail In-Store, the confidence intervals for incrementality always contain zero. 
Interestingly, for Retail Online and B2B, there are some months where the confidence intervals contain zero, and some months where they do not. In more detail, the gift card program is significantly incremental for three months in Retail Online and nine months in B2B. Further, these channels appear least incremental in the first three months of the year and most incremental in the final two. If we suspect that cards from these channels are regularly purchased by self-gifting units, this might reflect an increase in self-gifted bookings prior to the holiday season.

 \begin{figure}[H]
\centering
\includegraphics[width=\textwidth]{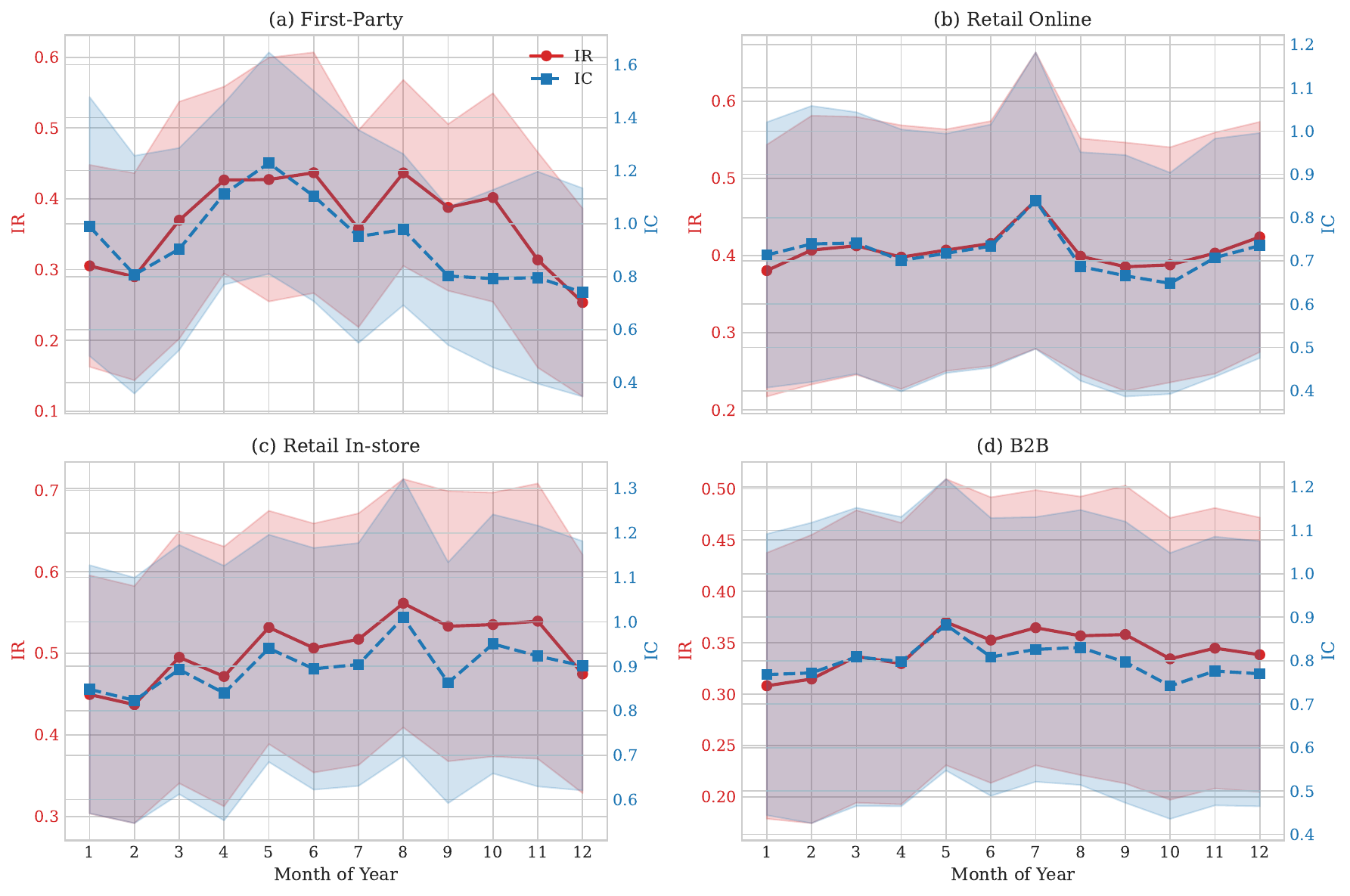}
\caption{Monthly estimates alongside 95\% confidence intervals for both the IR and IC by channel amongst suspected self-gifters.}
\label{fig:sg_main}
\end{figure}

We see that the time series estimates when restricted to self-gifting units largely support our findings from Section~\ref{sec:case_study:self_gift}. We see that all channels exhibit positive incrementality regardless of month, and that all confidence intervals comfortably exclude zero. This thus reaffirms our finding of significant incrementality amongst self-gifting units. 
We see that the point estimates for the IR and IC are largely stable over time for all channels except First-Party, in which there seem to be substantial dips in incrementality in the later months of the year (November, December) as well as in January and February. We do not attempt to identify the specific drivers of these fluctuations.

\subsection{Sensitivity to Transferability Violations}
\label{app:case_study:sensitivity}

We now assess the sensitivity of our estimates to violations of transferability using a marginal sensitivity model akin to those in \cite{rosenbaum2002observational, tan2006distributional, kallus2019interval, zhao2019sensitivity}. We introduce a sensitivity parameter $\Lambda \geq 1$ that bounds the \textit{multiplicative degree of violation of transferability} between the conditional booking odds ratio between the observational dataset and experiment. In more detail, defining the ratio 
\[
\eta_0(\tau, x) := \frac{\P_{\obs}(B = 1 \mid G = 0, X = x)}{\P_{\obs}(B = 1 \mid G_\tau = 1, X = x)},
\]
we are interested in weakening our assumption to that of 
\begin{equation}
\label{eq:sensitivity_bd}
\frac{1}{\Lambda}\leq \eta_0(X)\Big/ \frac{p_0(X)}{q_0(\tau, X)}\leq \Lambda.
\end{equation}
Under this assumption, the incremental revenue $\Delta_0$ is no longer point identified, and hence the IR and IC are not as well.
Nonetheless, partial identification still remains possible. In particular, respective upper and lower bounds on the incremental revenue are given by
\begin{align*}
\Delta_{\mathrm{LB}} &:= \E_{\obs}\left[Y \sum_{\tau \in \calJ} S_\tau\right] - \Lambda \E_{\obs}\left[g_0(X)p_0(X)\sum_{\tau \in \calJ}\frac{S_\tau}{q_0(\tau, X)}\right] \\ 
\Delta_{\mathrm{UB}} &:= \E_{\obs}\left[Y \sum_{\tau \in \calJ} S_\tau\right] - \frac{1}{\Lambda} \E_{\obs}\left[g_0(X)p_0(X)\sum_{\tau \in \calJ}\frac{S_\tau}{q_0(\tau, X)}\right],
\end{align*}
which follows because the formula for the incremental revenue is monotonically increasing in $\Lambda$.
These partial identification bounds naturally imply corresponding upper and lower bounds on both the IC and IR:
\begin{align*}
\rho_0 \in \left[\frac{\Delta_{\mathrm{LB}}}{\E_{\obs}[Y G]} , \frac{\Delta_{\mathrm{UB}}}{\E_{\obs}[Y G]}\right] \quad \text{and} \quad \kappa_0 \in \left[\frac{\Delta_{\mathrm{LB}}}{\E_{\obs}[A G]} , \frac{\Delta_{\mathrm{UB}}}{\E_{\obs}[A G]}\right].
\end{align*}

Using these partial identification bounds, we can compute corresponding \textit{sensitivity values} per channel, which are defined  as \textit{the largest value $\Lambda$ such that} $\Delta_{\mathrm{LB}} \geq 0$. Intuitively, the larger $\Lambda$ is, the more robust the sign of our estimate within a channel is to violations of transferability. 

We report our sensitivity findings when considering both all units in Figure~\ref{fig:au_sens_ic}. We note that, because our point estimates for the incrementality coefficient are negative in First-Party and Retail In-Store channels, we can only apply our sensitivity bounds to Retail Online and B2B. The figure shows the upper and lower bound on the IC as a function of the sensitivity parameter $\Lambda$, and includes a vertical bar noting the breakdown point (i.e., the point at which the lower bound on the IC point estimates becomes zero).  We note that while our plot focuses on the IC $\kappa_0$, the numbers are exactly analogous for the IR $\rho_0$. 
We note that the breakdown points for both of the aforementioned channels sit at around 1.25, indicating that our point estimates are only mildly insensitive to violations of transferability.\footnote{We briefly note how we compute the upper sensitivity bound $\Delta_{\mathrm{UB}}$ for the incremental revenue at a fixed parameter $\Lambda \geq 1$. The case of the lower bound is exactly analogous. To obtain the upper bound, one naturally needs to scale the second summand in the moment $m(W; g, p, q)$ and the $g$-representer term (i.e., the term involving $\alpha_g$) by $\Lambda$, as these terms both naturally involve the ratio $p_0/q_0$. Since the matrices $A_{0, p}$ and $A_{0, q}$ leveraged in the experimental correction naturally depend on these scaled quantities, they are correspondingly modified as well.}

\begin{figure}[H]
\centering
\includegraphics[width=\textwidth]{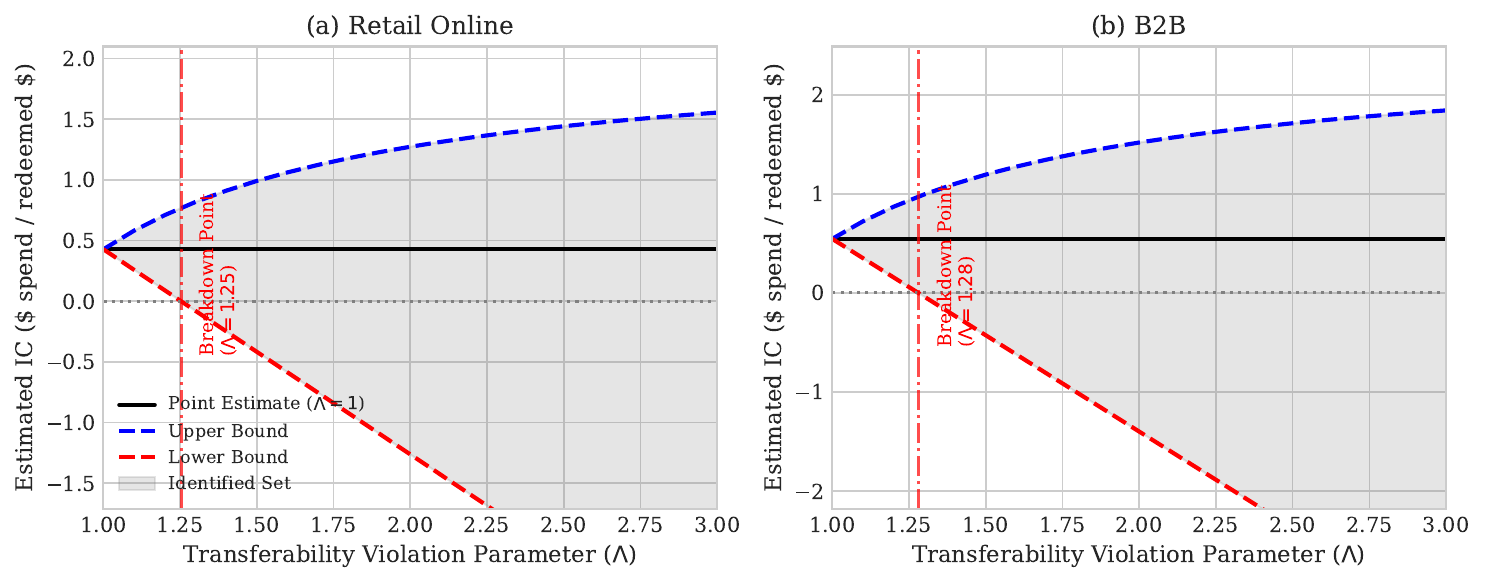}
\caption{Sensitivity to transferability violations for the IC.}
\label{fig:au_sens_ic}
\end{figure}

Figure~\ref{fig:sg_sens_ic} likewise displays the breakdown points for the IC when we restrict our estimates to suspected self-gifters, as defined in Section~\ref{sec:case_study:self_gift}. In this setting, since all four point estimates are positive regardless of channel, we can actually compute all four breakdown points. We see that breakdown point for First-Party, B2B, and Retail Online are all roughly the same, sitting between 1.52 and 1.69, and the breakdown point for Retail In-Store sits at around 2.01. This indicates that our findings for self-gifters are generally much more robust to violations than those presented for the entire population.

\begin{figure}[H]
\centering
\includegraphics[width=\textwidth]{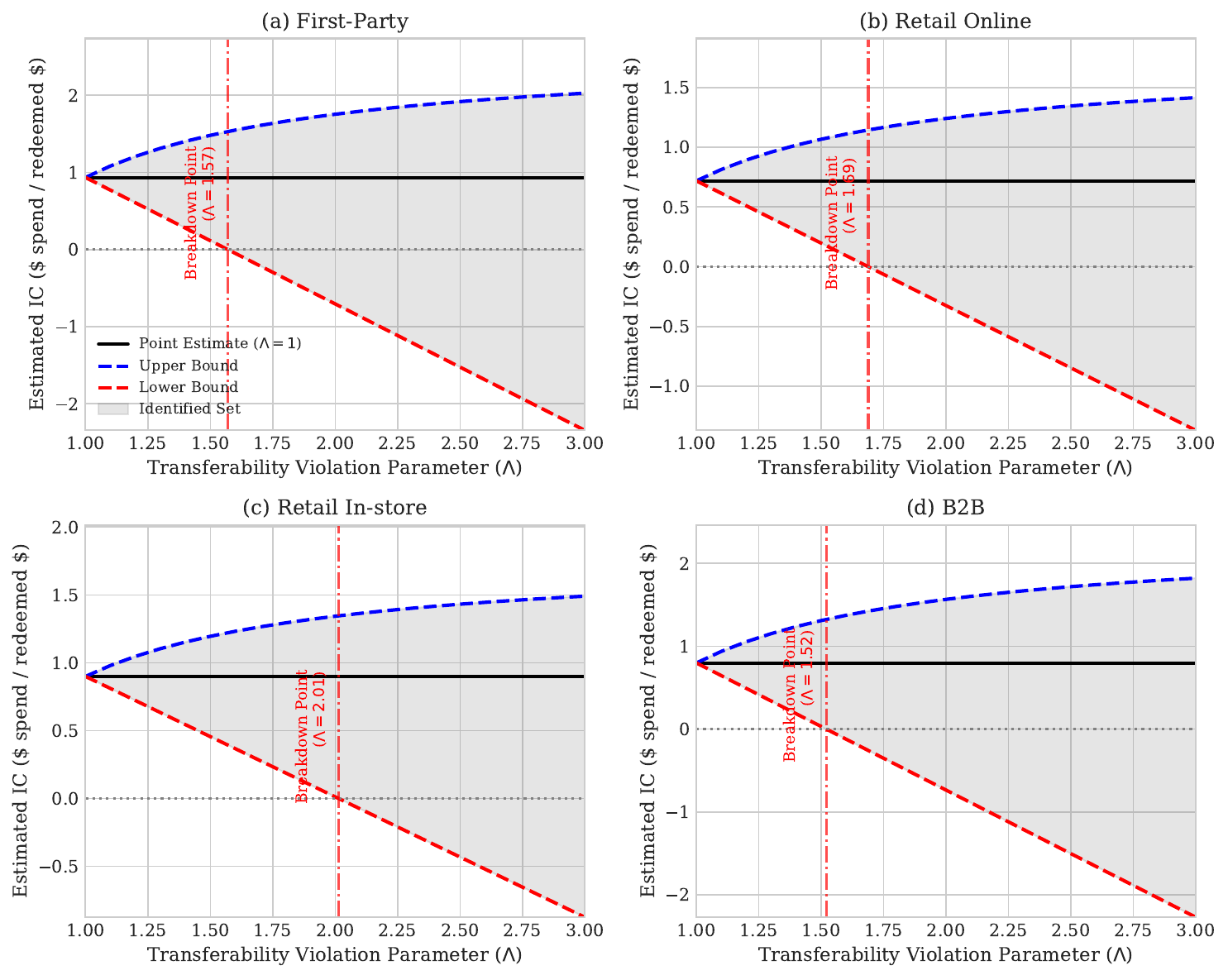}
\caption{Transferability sensitivity for IC, restricted to suspected self-gifters.}
\label{fig:sg_sens_ic}
\end{figure}

\subsection{Features and Grouping}
\label{app:case_study:features}

In our case study, we leverage 153 features  that are designed to capture a wealth of information about a customer's historical interactions with Airbnb. Our goal in gathering features is to explain customer behavior that can potentially impact gift card receipt probability, customer booking probability, and also realized gross booking value (GBV). As noted earlier, these customer behaviors are diverse, and can include self-gifting, price consciousness, pre-existing travel intent, and more. We thus find it fruitful to gather as many features as possible that may serve as proxies for these behaviors. Our features can be roughly grouped into three categories as follows.

\begin{description}
  \item[Gifting-related.] Gifting-related features capture a user's historical interactions with Airbnb's gift card program. These features include a mixture of both average and aggregate statistics, with some examples including the total number of gift cards used (broken down by region and channel), the average and minimum number of days between gift card purchase and redemption, and fraction of historical trips booked using gift card value, for example. See Table~\ref{tab:features_gifting} for an entire enumeration. For the experimental data, these features are measured over the three years leading up to the date of gift card assignment at the start of the experiment. For the observational dataset, these features are measured up to the point of a customer's first booking in the month under consideration, with the convention that features are constructed up to the start of the month for customers that do not make a purchase. 
  We note that, although the window of time used to define the observational features may include the date of gift card receipt (which is generally unknown for all customers), the features were not constructed using information about any received gift card that was leveraged towards a purchase in the target month.
  Many of the features are further broken down by (a) distribution channel (First-Party, B2B, Retail
    In-Store, Retail Online) and (b) geographic region (North America, EMEA,
    APAC, LATAM). Further, each feature associated with measuring the number of days between purchase and some event are accompanied by a binary
    missingness indicator, as some units may have never received a gift card from a certain region (or may have never received a gift card altogether). 
  \item[Historical Spend and Bookings.] Historical spend and booking features capture the volume,
    frequency, and spend amount associated with uncanceled bookings over the preceding three years. We use the same general convention for constructing these features as was leveraged for gifting-related features. This category includes
    average number of bookings per year, average and total nights booked, and the count of
    uncanceled bookings broken down by payment method (e.g., 
    credit card, wallet, bank account, etc.).
    Further, we measure the number of bookings with Friday/Saturday check-in, which is meant to capture vacation travel (vs.\ business travel). Likewise, we count the number of historical
    bookings with any discount applied.
    A full enumeration of these features can be found in Tables~\ref{tab:features_booking} and \ref{tab:features_spend}.
  \item[PDP-related.] PDP-related (or product detail page-related) features aim to capture a user's recent browsing
    behavior on listing detail pages of Airbnb's platform. These features are generally computed over 25 to 30 day
    windows prior to either a customer's first booking in the focal month, or prior to the start of the month for those who do not book. These features include counts of listing views (with and
    without search dates entered), listing impressions with dates, and the number of
    unique listings viewed. They also include point-in-time price information for the most
    recently viewed listing (accommodation nightly price, nightly price with fees and
    taxes, and total trip price), as well as cross-listing summary statistics (mean,
    interquartile range (IQR), and variance) of nightly and total prices across all recently
    viewed listings. Each price-based feature is accompanied by a binary missingness
    indicator. These features are displayed in Table~\ref{tab:features_pdp}.
\end{description}

Throughout the following table, ``NA ind.''\ denotes the binary missingness indicator for the corresponding feature.

\begin{longtable}{ll}
\caption{Gifting-related features.}
\label{tab:features_gifting} \\
\toprule
\textbf{Feature description} & \textbf{Feature description (cont.)} \\
\midrule
\endfirsthead
\multicolumn{2}{c}{\tablename\ \thetable{} -- continued} \\
\toprule
\textbf{Feature description} & \textbf{Feature description (cont.)} \\
\midrule
\endhead
\midrule
\multicolumn{2}{r}{\textit{Continued on next page}} \\
\endfoot
\bottomrule
\endlastfoot
Avg.\ days to claim                         & Avg.\ days to claim (NA ind.) \\
Avg.\ days to claim, 1st party              & Avg.\ days to claim, 1st party (NA ind.) \\
Avg.\ days to claim, B2B                    & Avg.\ days to claim, B2B (NA ind.) \\
Avg.\ days to claim, retail instore         & Avg.\ days to claim, retail instore (NA ind.) \\
Avg.\ days to claim, retail online          & Avg.\ days to claim, retail online (NA ind.) \\
Avg.\ days to claim, North America          & Avg.\ days to claim, North America (NA ind.) \\
Avg.\ days to claim, EMEA                   & Avg.\ days to claim, EMEA (NA ind.) \\
Avg.\ days to claim, Asia                   & Avg.\ days to claim, Asia (NA ind.) \\
Avg.\ days to claim, Latin America          & Avg.\ days to claim, Latin America (NA ind.) \\
Min.\ days to claim                         & Min.\ days to claim (NA ind.) \\
Min.\ days to claim, 1st party              & Min.\ days to claim, 1st party (NA ind.) \\
Min.\ days to claim, B2B                    & Min.\ days to claim, B2B (NA ind.) \\
Min.\ days to claim, retail instore         & Min.\ days to claim, retail instore (NA ind.) \\
Min.\ days to claim, retail online          & Min.\ days to claim, retail online (NA ind.) \\
Min.\ days to claim, North America          & Min.\ days to claim, North America (NA ind.) \\
Min.\ days to claim, EMEA                   & Min.\ days to claim, EMEA (NA ind.) \\
Min.\ days to claim, Asia                   & Min.\ days to claim, Asia (NA ind.) \\
Min.\ days to claim, Latin America          & Min.\ days to claim, Latin America (NA ind.) \\
Avg.\ days to redeem                        & Avg.\ days to redeem (NA ind.) \\
Avg.\ days to redeem, 1st party             & Avg.\ days to redeem, 1st party (NA ind.) \\
Avg.\ days to redeem, B2B                   & Avg.\ days to redeem, B2B (NA ind.) \\
Avg.\ days to redeem, retail instore        & Avg.\ days to redeem, retail instore (NA ind.) \\
Avg.\ days to redeem, retail online         & Avg.\ days to redeem, retail online (NA ind.) \\
Avg.\ days to redeem, North America         & Avg.\ days to redeem, North America (NA ind.) \\
Avg.\ days to redeem, EMEA                  & Avg.\ days to redeem, EMEA (NA ind.) \\
Avg.\ days to redeem, Asia                  & Avg.\ days to redeem, Asia (NA ind.) \\
Avg.\ days to redeem, Latin America         & Avg.\ days to redeem, Latin America (NA ind.) \\
Min.\ days to redeem                        & Min.\ days to redeem (NA ind.) \\
Min.\ days to redeem, 1st party             & Min.\ days to redeem, 1st party (NA ind.) \\
Min.\ days to redeem, B2B                   & Min.\ days to redeem, B2B (NA ind.) \\
Min.\ days to redeem, retail instore        & Min.\ days to redeem, retail instore (NA ind.) \\
Min.\ days to redeem, retail online         & Min.\ days to redeem, retail online (NA ind.) \\
Min.\ days to redeem, North America         & Min.\ days to redeem, North America (NA ind.) \\
Min.\ days to redeem, EMEA                  & Min.\ days to redeem, EMEA (NA ind.) \\
Min.\ days to redeem, Asia                  & Min.\ days to redeem, Asia (NA ind.) \\
Min.\ days to redeem, Latin America         & Min.\ days to redeem, Latin America (NA ind.) \\
\# GCs claimed                              & \# GCs claimed, 1st party \\
\# GCs claimed, 1st party, North America    & \# GCs claimed, 1st party, EMEA \\
\# GCs claimed, 1st party, Asia             & \# GCs claimed, 1st party, Latin America \\
\# GCs claimed, North America               & \# GCs claimed, EMEA \\
\# GCs claimed, Asia                        & \# GCs claimed, Latin America \\
\# GCs claimed, B2B                         & \# GCs claimed, B2B, North America \\
\# GCs claimed, B2B, EMEA                   & \# GCs claimed, B2B, Asia \\
\# GCs claimed, B2B, Latin America          & \# GCs claimed, retail instore \\
\# GCs claimed, retail instore, North America & \# GCs claimed, retail instore, EMEA \\
\# GCs claimed, retail instore, Asia        & \# GCs claimed, retail instore, Latin America \\
\# GCs claimed, retail online               & \# GCs claimed, retail online, North America \\
\# GCs claimed, retail online, EMEA         & \# GCs claimed, retail online, Asia \\
\# GCs claimed, retail online, Latin America & \\
Avg.\ GC value redeemed per trip            & Avg.\ GC value redeemed per GC trip \\
\% bookings using a GC                      & \% spend from GCs \\
\end{longtable}

\begin{table}[h]
\centering
\caption{Historical booking activity features.}
\label{tab:features_booking}
\begin{tabular}{ll}
\toprule
\textbf{Feature description} & \textbf{Feature description (cont.)} \\
\midrule
Avg.\ bookings per year          & Avg.\ nights per booking \\
Total nights booked              & \# bookings \\
\# bookings, Android Pay         & \# bookings, Apple Pay \\
\# bookings, bank account        & \# bookings, credit card \\
\# bookings, Geneva              & \# bookings, Klarna \\
\# bookings, PayPal              & \# bookings, Fri/Sat check-in \\
\# bookings with discounts       & \\
\bottomrule
\end{tabular}
\end{table}

\begin{table}[h]
\centering
\caption{Historical spend features.}
\label{tab:features_spend}
\begin{tabular}{ll}
\toprule
\textbf{Feature description} & \textbf{Feature description (cont.)} \\
\midrule
Avg.\ booking value per trip          & Avg.\ booking value per night \\
Total GBV                             & Total GBV, Android Pay \\
Total GBV, Apple Pay                  & Total GBV, bank account \\
Total GBV, credit card                & Total GBV, Geneva \\
Total GBV, Klarna                     & Total GBV, PayPal \\
Total GBV, Fri/Sat check-in           & Total GBV with discounts \\
\bottomrule
\end{tabular}
\end{table}

\begin{table}[h]
\centering
\caption{PDP (listing detail page) engagement and price features.}
\label{tab:features_pdp}
\begin{tabular}{ll}
\toprule
\textbf{Feature description} & \textbf{Feature description (cont.)} \\
\midrule
Listing views, 30d                                & Listing views, 30d (NA ind.) \\
Listing views with dates, 30d                     & Listing views with dates, 30d (NA ind.) \\
Listing impressions with dates, 25d               & Listing impressions with dates, 25d (NA ind.) \\
Unique listings viewed, 30d                    & \\
Last listing nightly price                        & Last listing nightly price (NA ind.) \\
Last listing nightly price with fees \& tax       & Last listing nightly price with fees \& tax (NA ind.) \\
Last listing total price                          & Last listing total price (NA ind.) \\
Avg.\ nightly price across listings               & Avg.\ nightly price across listings (NA ind.) \\
Avg.\ nightly price with fees \& tax              & Avg.\ nightly price with fees \& tax (NA ind.) \\
Avg.\ total price across listings                 & Avg.\ total price across listings (NA ind.) \\
IQR nightly price with fees \& tax                & IQR nightly price with fees \& tax (NA ind.) \\
IQR total price across listings                   & IQR total price across listings (NA ind.) \\
Variance of nightly price with fees \& tax        & Variance of nightly price with fees \& tax (NA ind.) \\
Variance of total price across listings           & Variance of total price across listings (NA ind.) \\
\bottomrule
\end{tabular}
\end{table}

\subsubsection{Selected Features for the Transferred Booking Hazards}

We conclude by noting the restricted selection of features used for training the logistic hazard models transferred from the experiment to the observational data. These features are intended to provide a small number of
representative features from each of the groups mentioned above. The 23 selected features are listed in
Table~\ref{tab:default_columns}.

\begin{table}[h]
\centering
\caption{The restricted set of features leveraged to train transferred probabilities/hazards.
}
\label{tab:default_columns}
\begin{tabular}{lll}
\toprule
\textbf{Group} & \textbf{Feature description} \\
\midrule
Gifting    & Avg.\ days to claim \\
Gifting    & Avg.\ days to claim (NA ind.) \\
Gifting    & Min.\ days to claim \\
Gifting    & Min.\ days to claim (NA ind.) \\
Gifting    & Avg.\ days to redeem \\
Gifting    & Avg.\ days to redeem (NA ind.) \\
Gifting    & \# GCs claimed \\
Gifting    & Avg.\ GC value redeemed per GC trip \\
Gifting    & \% spend from GCs \\
Gifting    & \% bookings using a GC \\
\midrule
Booking activity & Avg.\ bookings per year \\
Booking activity & Avg.\ nights per booking \\
\midrule
Spend      & Avg.\ booking value per trip \\
Spend      & Avg.\ booking value per night \\
\midrule
PDP        & Listing views, 30d \\
PDP        & Listing views, 30d (NA ind.) \\
PDP        & Last listing nightly price \\
PDP        & Last listing nightly price (NA ind.) \\
PDP        & Unique listings viewed, 30d \\
PDP        & IQR nightly price with fees \& tax \\
PDP        & IQR nightly price with fees \& tax (NA ind.) \\
PDP        & Variance of total price across listings \\
PDP        & Variance of total price across listings (NA ind.) \\
\bottomrule
\end{tabular}
\end{table}

\section{Extended Discussion}
\label{app:discuss}
The goal of this appendix is to continue the discussion of our empirical findings, which were outlined in Section~\ref{sec:case_study} and then interpreted in Section~\ref{sec:discuss}.  
As with the discussion section in the paper, this appendix is mostly qualitative in nature. 
First, in Section~\ref{app:discuss:econ_model}, we describe a behavioral model under which our findings for self-gifters (high extensive margins) can be justified even without an advertising effect. Then, in Section~\ref{app:discuss:assumptions}, we discuss potential threats to causal assumptions in the data we leverage. Lastly, in Section~\ref{app:discuss:profit}, we briefly mention how our results may be relevant to determining the profitability of gift card programs, and why we do not tackle measuring profitability in the paper.

\subsection{A Simple Behavioral Economic Model}
\label{app:discuss:econ_model}

In Section~\ref{sec:case_study:self_gift}, we found that self-gifters exhibited large extensive margins. For these units, attributing the lift in booking probability to something like an advertising effect associated with gift card receipt does not appear fully reasonable. In particular, self-gifters are potentially deal-seeking customers with pre-existing travel intent, and thus are likely aware of Airbnb's offerings. 
As noted in our discussion in Section~\ref{sec:discuss}, one might think that rational self-gifting units who are aware of the business should have zero extensive margins, in part due to the fungibility of gift card value and cash. 
While we heuristically justified that such results may be justified due to implicit discounts on gift cards available from third-party providers, we did not provide a rigorous economic justification. 

In this subsection we posit a simple behavioral economic model that corroborates our findings of high extensive margins amongst self-gifters. Further, this model will show that it is possible to have a substantial lift in booking probability while still having zero intensive margins. 
The role of this model is merely supportive, meant to provide additional economic insight into one way customers may make booking decisions that results in similar findings to those noted above.
We are not affirming that the model introduced in this section actually describes the behavior of customers, and we note that our primary identification/estimation framework only relies on \textit{statistical assumptions}, not \textit{behavioral economic} ones.

Consider customer $i$, who chooses between not booking and making a booking with gross booking value $y$ in a feasible set $\mathcal{Y}_i$. We assume that their utility of booking a property of value $y$ without a gift card is given as
\begin{equation}
\label{eq:utility_without_card}
U_i^0(y) = v_i + u_i(y) - y,
\end{equation}
where $v_i$ captures the customer's fixed preference for taking an Airbnb trip relative to the outside option, and $u_i(y)$ captures the utility obtained from the attributes of a booking costing $y$. We assume without loss of generality that the utility obtained from the outside option is normalized to zero.
Suppose that the utility maximization problem over $\calY_i$ has a solution, and set
\begin{equation}
\label{eq:optimal_booking_value}
y_i^\star \in \arg\max_{y\in\mathcal{Y}_i}\{u_i(y)-y\}
\end{equation}
to be the gross booking value associated with the highest utility option. 
We define a customer's baseline net surplus from booking to be
\begin{equation}
\label{eq:baseline_surplus}
\omega_i
:=
v_i+u_i(y_i^\star)-y_i^\star.
\end{equation}
Adopting the convention that a customer books when indifferent, we define the \textit{potential booking outcome  without a gift card} as
\begin{equation}
\label{eq:booking_without_card}
B_i(0)=\1\{\omega_i\geq 0\}.
\end{equation}

We now define a customer's utility function under a discounted gift card. Suppose that customer $i$ can acquire a gift card with face value $f_i$ at price $(1-\lambda_i)f_i$, where $\lambda_i\in[0,1]$ represents the proportional discount associated with the acquisition.\footnote{For instance, if a user gets 5\% cash back on the purchase of a \$100 gift card, we have $f_i = 100$ and $\lambda_i = 0.05$.} Let $c_i \geq 0$ denote the monetary equivalent of the search, acquisition, or transaction cost required to obtain and use the card.

For this simple economic model, we assume that (i) a customer can only ever have/leverage one gift card at a time, (ii) the customer fully redeems the card on a booking, and (iii) that the GBV of all properties in the consideration set exceed the gift card face value, i.e., that $y \geq f_i$ for all $y \in \calY_i$.\footnote{These assumptions are likely feasible in practice for low face value gift cards that are redeemed on expensive, multi-night stays.}
A customer booking a property with value $y$ using such a card pays $(1-\lambda_i)f_i$ to acquire the card and $y-f_i$ at redemption. Including the acquisition cost, the customer's total effective expenditure is
\[
(1-\lambda_i)f_i+(y-a_i)+c_i
=
y-\lambda_i f_i+c_i.
\]
Defining the \textit{net economic benefit} or \textit{lift} provided by the card as $\ell_i:=\lambda_i f_i-c_i$, the utility from booking under the gift card program is then
\begin{equation}
\label{eq:utility_with_card}
U_i^1(y)
=
v_i+u_i(y)-y+\ell_i.
\end{equation}
Because $\ell_i$ does not depend on $y$, access to a discounted gift card does not change the solution to the conditional booking problem:
\begin{equation}
\label{eq:unchanged_booking_value}
\arg\max_{y\in\mathcal{Y}_i}U_i^1(y)
=
\arg\max_{y\in\mathcal{Y}_i}U_i^0(y).
\end{equation}
In particular, $y_i^\star$ remains an optimal booking value under this simplistic economic model. The potential booking decision \textit{under the gift card program} is consequently
\begin{equation}
\label{eq:booking_with_card}
B_i(1)=\1\{\omega_i+\ell_i\geq0\}.
\end{equation}

Suppose henceforth that $\ell_i\geq0$, so that the financial benefit from the card weakly exceeds its acquisition cost. Equations~\eqref{eq:booking_without_card} and~\eqref{eq:booking_with_card} then imply $B_i(1)\geq B_i(0)$, and so the model generates monotonicity of the booking decision. Moreover, the individual extensive effect under this booking model is given by the booking difference
\begin{equation}
\label{eq:individual_extensive_effect}
B_i(1)-B_i(0)
=
\1\{-\ell_i\leq\omega_i<0\}.
\end{equation}
Customers whose baseline booking surplus lies in $[-\ell_i,0)$ do not book at the ordinary price but do book when the discounted gift card is available --- these are precisely the customers providing positive extensive margins.

Further, the potential spend of a customer with gift card receipt status $g \in \{0, 1\}$ is given by
\begin{equation}
\label{eq:potential_spend_structural}
Y_i(g):=y_i^\star B_i(g),
\qquad g\in\{0,1\}.
\end{equation}
Under our model, it follows exactly that the individual treatment effect is given by
\begin{equation}
\label{eq:individual_revenue_effect}
Y_i(1)-Y_i(0)
=
y_i^\star\1\{-\ell_i\leq\omega_i<0\}.
\end{equation}
In other words, the card changes whether or not a customer books (the extensive margin) but does not change the customer's chosen spend conditional on booking (the intensive margin). Thus, we see that our findings of strongly positive extensive margins in Airbnb's data are roughly justified from a simple economic model. We decompose the various behaviors of a customer (cannibalization, extensive-margin generating, non-participation) under our economic model in Table~\ref{tab:endogenous_customer_types}.

\begin{table}[!htb]
\centering
\begin{tabular}{p{0.22\textwidth}p{0.37\textwidth}p{0.30\textwidth}}
\toprule
\textbf{Surplus region} & \textbf{Booking behavior} & \textbf{Endogenous customer type} \\
\midrule
$\omega_i\geq0$
&
$B_i(0)=B_i(1)=1$
&
Cannibalizing customer: the customer books with or without the gift card.
\\[0.3em]
$-\ell_i\leq\omega_i<0$
&
$B_i(0)=0$ and $B_i(1)=1$
&
Extensive-margin customer: the gift card induces a booking that otherwise would not occur.
\\[0.3em]
$\omega_i<-\ell_i$
&
$B_i(0)=B_i(1)=0$
&
Nonparticipating customer: the gift card benefit is insufficient to induce a booking.
\\
\bottomrule
\end{tabular}
\caption{Customer types generated by the rational participation model.}
\label{tab:endogenous_customer_types}
\end{table}

\subsection{Validity of Causal Assumptions}
\label{app:discuss:assumptions}

Our primary identification results, Theorems~\ref{thm:identification} and \ref{thm:ident-type}, rest on four causal assumptions. Some assumptions are relatively innocuous. For instance, Assumption~\ref{ass:positivity} states that there is no region of the covariate space for which a customer receives a gift card with probability $1$, which is a relatively mild assumption and would only be violated if there existed strata of customers (determined by covariates $X$) who systematically only booked with gift cards. Further, Assumption~\ref{ass:partial-obs} is a structural assumption about the partial observability of the treatment in the observational data and hence easy to justify in our context.  However, Assumptions~\ref{ass:exogeneity} and \ref{ass:transfer} (or Assumption~\ref{ass:transfer-type} in the heterogeneous treatment setting) are more substantial and require careful consideration. While these assumptions cannot be verified or refuted with certainty from data alone, we devote this section to a systematic examination of the conditions under which each assumption may be violated. We further either argue why such violations are unlikely in our setting, or discuss where we think such assumptions may be partially violated in our use case. 

\subsubsection{Potential Violations of Conditional Exogeneity}
\label{sec:discuss:exogeneity}

Assumption~\ref{ass:exogeneity} requires that, conditional on observed covariates $X$, a customer's counterfactual spend without a gift card is independent of whether or not they actually received one. The primary threat to this condition is the existence of \textit{unobserved confounders}, latent variables that jointly impact gift card receipt status and control spending behavior. While exhausting the list of possible confounding behaviors would likely be impossible, we can enumerate what we believe to be the most likely threats to conditional exogeneity. We focus on three such examples: self-gifting vs.\ true-gifting behavior, preexisting customer travel intention, and customer wealth/price sensitivity. Naturally, these categories can overlap. As one example, many self-gifting customers likely possess preexisting travel intent. Likewise, wealthier consumers may be less prone to deal-seeking behaviors, and thus may not search for discounts on properties.
Despite these overlaps, we find it important to discuss each of these points in sequence and discuss how the observables we collect can mediate these latent variables.
We also conclude this subsection by noting practical limitations in constructing the features outlined in Section~\ref{app:case_study:features}, and how this may impact the validity of conditional exogeneity.

\paragraph{Self-Gifting vs.\ True-Gifting Behavior}
Perhaps the most notable unobserved feature in our problem is the indicator of whether or not a consumer being considered is engaging in self-gifting behavior. As a reminder, a unit is deemed to exhibit \textit{self-gifting behavior} if they procure or purchase their own gift cards. Likewise, \textit{true-gifting behavior} is characterized by a unit receiving cards from external sources, such as through friends, family, or potentially even business promotions. We naturally expect a customer's self-gifting status to impact both their probability of receiving a gift card and counterfactual spend behavior. For instance, self-gifting units may exhibit significantly higher gifting propensities, as they will regularly check for cashback or discounts on gift cards at their retailer of preference.
Meanwhile, patrons who do not participate in self-gifting may have more modest gift card receipt rates, as they would only receive cards for special occasions (e.g., a wedding or birthday) and thus may not leverage a card in a large fraction of their purchases. 
One may similarly speculate about differences in control spending behavior, perhaps expecting that many self-gifting units may have $Y(0) = 0$, indicating that they would have not made a booking (or would have purchased at a competing firm) without the implicit discount obtained through self-gifted cards.

In our empirical analysis, we leverage a variety of features that are targeted at mediating the effects of self-gifting behavior. As one example, a customer who has leveraged gift cards for a large fraction of their historical purchases (say, upwards of 40\%) or who has a large share of historical GBV consisting of gift card spend may have high self-gifting propensity. Further, customers for whom the minimum historical numbers of days between gift card purchase and claim/redemption is small (say, one or two days) may likewise be self-gifting.
We measure all of these quantities (amongst others) in an attempt to shut down potential confounding paths.

In addition to our rich measurements, in Section~\ref{sec:case_study:self_gift} we also took several steps throughout our analysis to ensure robustness to self-gifting.  First, we decomposed our estimates for both the IR and IC by channel, separately reporting numbers for channels where self-gifting is less likely (First-Party and Retail In-Store channels) and more likely (Retail Online and B2B channels). Second, we also measured both the IR and IC when we restricted our attention to inferred self-gifting units, who were defined as those who made a booking using a gift card that was less than seven days old. As noted above, the number of days between gift card purchase and redemption serves as a natural proxy for self-gifting behavior.

\paragraph{Customer Travel Intent}
Another unobserved confounder in our setting is whether or not the customer had preexisting travel intent. A customer's travel intent naturally impacts both their likelihood of receiving a gift card and how much they would spend without a gift card. 
As one example, if a customer is determined to visit a certain location (prior to gift card receipt), they may exhibit lower sensitivity to prices and will likely not have $Y(0) = 0$ (corresponding to not booking). Further, they may express travel desire to family and friends, and thus may be significantly more likely to receive a gift card. On the other hand, if a customer does not have preexisting intent, they may not book at all absent a gift card, and thus would have zero counterfactual spend. Further, gift card receipt may be less likely amongst individuals with low travel intent, as self-gifting would be improbable and friends may provide other gifts that they believe would be of higher utility to the customer in question.

While a customer's travel intent is not directly measurable, we do incorporate a variety of features aimed at capturing this unobservable into our estimation pipeline. We construct a large number of \textit{product details page (PDP)} features, which capture a customer's interaction with Airbnb's website. As examples, we measure the variability and interquartile range of both the total and nightly price of the properties that a customer views. Additionally, we measure the average nightly price and the number of unique property viewings. We detailed the full collection of such features in Table~\ref{tab:features_pdp} above. In addition to these PDP features, we also collect spend and booking-related features that may provide insight into a customer's travel intention. For instance, we measure the average number of bookings and average nights per booking, both over the previous three years. The aforementioned features (amongst others) all serve to mediate the effect of travel intention on both gift card receipt and customer spend amount.

\paragraph{Customer Mental Accounting}
Many existing works evaluating gift card programs note the diverse ways customers may view gift card value.
For instance, some units may exhibit a high propensity towards cannibalization, with any obtained gift card value directly replacing potential out-of-pocket spend. Likewise, other units may view obtained gift card value as \textit{windfall gains}, thus appending this value onto what they would have otherwise spent out-of-pocket. Further, some units may exhibit \textit{deal-seeking behavior}, only booking trips when there is some perceived discount on the property (with the discount either being offered directly through Airbnb or implicitly through a gift card available from a rewards program). All of these behaviors would result in different customer booking behaviors, and we have no way of directly ascertaining the precise behavior patterns of the customers. 

We again construct a rich set of features to capture the above modes of mental accounting. To determine whether a customer is likely to cannibalize or view gift card value as ``free money,'' we measure both the historical percentage of GBV that was made up of gift card spend as well as the average gift card value redeemed per trip. Likewise, to capture deal-seeking behavior, we measure both the number of historical bookings with discounts and the total historical GBV booked by a customer on trips that involved a discount. 
Further, by breaking our analysis down on a channel-by-channel basis, we also capture any systematic differences in customer behavior that may exist across channels. In fact, in our decomposition of the IR into intensive, extensive, and cannibalization shares, we find that there are systematic differences in cannibalization behavior across channels. In particular, third-party channels exhibit much smaller cannibalization shares than first party channels. 

\paragraph{Limitations of Feature Construction}
In Sections~\ref{sec:case_study:empirical_setup} and \ref{app:case_study:features}, we detailed how we constructed the various features used in our analysis. In particular, we noted that there were differences in feature construction between the experimental and observational datasets. For the former, feature construction was rather natural --- we could construct all average and aggregate features up to the start of the experiment. 
However, for the latter, there is no natural point prior to the (generally unknown) time of gift card receipt that we could use as a reference. Thus, to be relatively consistent in our analysis, we chose to construct all features up to the point of a customer's first booking in the target month, leveraging the start of the period for those units who did not book.

We note that, by making this choice, one may suspect that the constructed features serve as mediators on the causal pathway from gift card receipt $G$ to a booking decision $B$.
We first argue that this is not the case for historical gifting, booking, and customer spend features.
Suppose we operate under the assumption that, if a customer possesses a gift card, they will use it in their next transaction. 
Then, any transaction that is counted towards the aforementioned features must have been executed \textit{prior to the receipt} of the gift card that was counted as treatment. In particular, historical gifting features were explicitly constructed to exclude any gift card that was redeemed in the target period, and historical spend/booking features do not count bookings made in the target month.

The features that are most likely to serve as mediators between $G$ and $B$ in our setting are the \textit{PDP features}, as they measure the browsing intensity of a customer in a time period leading up to their booking decision. We have included these features because we view them as invaluable in describing potentially confounding customer behavior, such as travel intent and price sensitivity. This is because they provide contemporaneous proxies for a customer's latent travel intent. Under the natural behavioral assumption that gift card receipt serves to increase browsing rates and that increased browsing rates in turn provide increased booking rates, conditioning on these features would serve to block a positive mediated effect of gift card receipt $G$ on the booking variable $B$. Consequently, to the extent that these features do mediate the effect of gift card receipt on booking rates, we would expect under the above behavioral assumption their inclusion to provide downward bias for our incrementality estimates, in particular, through lower estimated extensive margins. We emphasize that this argument does not imply that conditioning on post-treatment variables is generically conservative; rather, it describes the expected direction of bias when conditioning on the PDP features noted above.

We note that another way one could hope to construct features in the observational dataset is as follows. For units who book without gift cards, we continue to measure features up until the point of booking. Likewise, since we know the date of gift card purchase for those who book with a gift card, we can measure historical behaviors up until that point in time. There are, however, two major problems with this approach. First, this does not resolve the problem of collecting measurements for individuals who do not book in the target period. While we could still leverage the start of the focal period for such units, this can result in features being seriously misaligned, especially between bookers and non-bookers who possess very old cards. Second, a unit's date of gift card receipt need not coincide with the date of  purchase, especially for true-gifted individuals. Because of these issues, we prefer to keep the current definitions of the features in our analysis.

\subsubsection{Potential Violations of Transferability}
\label{sec:discuss:transfer}

We now discuss potential threats to our main transferability assumptions (Assumptions~\ref{ass:transfer} and~\ref{ass:transfer-type}), which require that the ratios between the booking probability (conditional on covariates) with and without gift cards remain invariant between the experimental and observational distributions.

\paragraph{Differences in Mode of Gift Card Receipt.} In the experiment, individuals who received gift cards were notified via an email from Airbnb, with those who do not receive a gift card having their experience at Airbnb unchanged. Thus, we view the method of gift card acquisition in the experiment as being closest to a \textit{true-gifting} scenario, wherein an individual's colleagues provide them with either a physical or digital gift card. The experiment then most closely matches the gift card receipt and redemption patterns found in channels where true-gifting is commonplace, such as First-Party and Retail In-Store channels. The experiment is somewhat less representative of channels where self-gifting is highly incentivized (via rewards programs or point redemption), with Retail Online and B2B channels being notable examples. 
We take several steps in our analyses to deal with these discrepancies. As noted repeatedly above, we collect a rich set of features capturing customer self-gifting behavior through historical spend and gift card usage patterns.  In particular, these features contain several proxies for self-gifting behavior, such as the fraction of historical purchases made with gift cards and the minimum/average number of days between gift card receipt and redemption. 
We also conduct sensitivity analyses in segments of the population where we believe the treatment mechanism to differ significantly from the experiment, namely Retail Online and B2B.

\paragraph{Temporal Drift and Seasonality.} There are two potential temporal threats to the validity of our transferability assumption. First, there is the concern of temporal drift. The experimental and observational datasets were gathered at distinct points in time, with the experiment being conducted in 2022 and the observational data being collected in 2024. One may then speculate that there are systematic differences in booking rates across time. For instance, one may expect conditional booking rates to increase over time as a company scales and better advertises to its customers. To respond to this, we emphasize that our transferability assumption does not require the exact conditional booking probabilities to be consistent across the populations. Rather, only the ratios of booking probabilities must remain invariant. As noted in Equation~\eqref{eq:f_transfer} immediately following Assumption~\ref{ass:transfer}, there could be some function $f(x)$ satisfying 
\[
\P_{\obs}(B = 1 \mid X = x, G = 0) = f(x)p_0(x) \quad \text{and} \quad \P_{\obs}(B = 1 \mid X =x, G = 1) = f(x)q_0(x)
\]
and our assumption will still hold. This function $f$ could capture systematic, covariate-dependent changes in booking rates induced by temporal drift, such as general increases or decreases in booking rates over time. Using the same logic, one can also see that seasonal changes in booking rates can also be captured through this scaling function $f$. 

\paragraph{Geographic Extrapolation.} The experiment was conducted using mostly North American customers, with a small number of customers from EMEA (Europe, Middle East, and Africa) being included as well. 
One may then speculate that transferability would fail or be less justifiable for geographically or culturally dissimilar regions. To address this, we focus our analysis entirely on units who have received gift cards through North American channels. Further, if differences in conditional booking rates systematically impact both the gifted and ungifted populations, then we may believe our transferability assumption to be equally valid across geographies.

\subsection{On the Profitability of Gift Card Programs}
\label{app:discuss:profit}

In this paper, we do not directly address the underlying question of the profitability of gift card programs across geographies and channels. Instead, we focus on just determining \textit{incrementality}, measured either through the fraction of GBV amongst gift card recipients that is attributable to the existence of the gift card program (the IR) and the number of incremental dollars of GBV booked for each one dollar of gift card value redeemed (the IC). Our choice in doing this is intentional.

First, determining the  profitability of gift card programs requires a bespoke, business-by-business and channel-by-channel analysis. Naturally, different channels exhibit different cost structures. First-party sales are often the cheapest for a business to operate, as there are fewer distributional costs. Third-party providers sometimes require gift cards to be sold to them at below face value so that they can make a profit while distributing the product.
These discount rates can vary widely across channels, and even within a channel. Looking at Retail In-Store as an example, Store A may require the business to provide its gift cards at a 10\% discount, whereas Store B may require a more substantial 15\% discount. 
Further, businesses may require a third-party platform to manage gift card transactions and monitor for/manage fraud.
Thus, gift cards may exhibit significantly higher transaction fees than other forms of payment.
Likewise, treatment of breakage may subtly impact downstream calculations of profitability.
To recall, \textit{breakage} is the fraction of unredeemed gift card value that a business can claim as profit after the card's expiration date. Breakage rates and laws naturally differ across geographies, and even within the United States different states possess drastically different breakage laws.
While we could propose a toy framework under which one could conduct said profitability calculations, doing this may lead to misleading results.

Second, the question of profitability given our point estimates for the IR and IC is no longer a counterfactual or causal problem. A given business could leverage our framework for estimating the IR and IC alongside their internal cost structure to determine the profitability of their gift card program on a channel-by-channel basis. Attempting to develop a general framework through which all businesses could assess profitability would not only be a substantially different task (and a very difficult one at that), but would also distract from our proposed estimands and causal identification framework. We thus view the development of such frameworks for determining profitability as interesting directions for future work.

\end{document}